\newif\ifabstract
\abstracttrue
\newif\iffull
\ifabstract \fullfalse \else \fulltrue \fi

\documentclass[11pt]{article}
\usepackage{amsfonts}
\usepackage{amssymb}
\usepackage{amstext}
\usepackage{amsmath}
\usepackage{xspace}
\usepackage{ntheorem}
\usepackage{graphicx}
\usepackage{url}
\usepackage{graphics}
\usepackage{colordvi}
\usepackage{hyperref}
\usepackage{xcolor}
\usepackage{cleveref}
\usepackage{subcaption}
\usepackage{tikz}
\usetikzlibrary{calc}

\advance \oddsidemargin by -0.8in
\newcommand{\myparskip}{3pt}
\parskip \myparskip

\newcommand{\procprocessU}{\ensuremath{\operatorname{ProcProcessU}}\xspace}
\newcommand{\procprocess}{\ensuremath{\operatorname{ProcProcess}}\xspace}

\newcommand{\procpromote}{\ensuremath{\operatorname{ProcPromote}}\xspace}

\newcommand{\procupdate}{\ensuremath{\operatorname{ProcUpdate}}\xspace}

\newcommand{\procrematchA}{\ensuremath{\operatorname{ProcRematchA}}\xspace}
\newcommand{\procrematchBU}{\ensuremath{\operatorname{ProcRematchBU}}\xspace}

\newcommand{\vhigh}{V^{\operatorname{high}}}
\newcommand{\vlow}{V^{\operatorname{low}}}
\newcommand{\vmed}{V^{\operatorname{med}}}

\newcommand{\ceil}[1]{\ensuremath{\left\lceil#1\right\rceil}}
\newcommand{\floor}[1]{\ensuremath{\left\lfloor#1\right\rfloor}}

\newcommand{\attime}[1][\tau]{^{(#1)}}

\newcommand{\set}[1]{\left\{ #1 \right\}}

\newcommand{\sset}{{\mathcal{S}}}

\newcommand{\be}{\begin{enumerate}}
\newcommand{\ee}{\end{enumerate}}
\newcommand{\bd}{\begin{description}}
\newcommand{\ed}{\end{description}}
\newcommand{\bi}{\begin{itemize}}
\newcommand{\ei}{\end{itemize}}

\newtheorem{theorem}{Theorem}[section]
\newtheorem{lemma}[theorem]{Lemma}
\newtheorem{observation}[theorem]{Observation}
\newtheorem{corollary}[theorem]{Corollary}
\newtheorem{claim}[theorem]{Claim}

\newtheorem{definition}[theorem]{Definition}
\newenvironment{proof}{\par \smallskip{\bf Proof:}}{\hfill\stopproof}
\def\stopproof{\square}
\def\square{\vbox{\hrule height.2pt\hbox{\vrule width.2pt height5pt \kern5pt
\vrule width.2pt} \hrule height.2pt}}

        {\hfill\stopproof}

\newenvironment{prog}[1]{
\begin{minipage}{5.8 in}
\begin{center}
{\sc #1}
\end{center}
}
{
\end{minipage}
}

\newcommand{\program}[3]{\begin{figure} \fbox{\vspace{2mm}\begin{prog}{#1} #3 \end{prog}\vspace{2mm}} 
			\caption{#1 \label{#2}} \end{figure}}

\renewcommand{\phi}{\varphi}
\newcommand{\eps}{\epsilon}

\newcommand{\poly}{\operatorname{poly}}

\newenvironment{properties}[2][0]
{
\begin{enumerate} \setcounter{enumi}{#1}}{\end{enumerate}}

\renewcommand{\path}{\mbox{\sf{path}}}

\begin{document}

\begin{titlepage}	
\title{A Faster Deterministic Algorithm for Fully Dynamic Maximal Matching\footnote{To appear in STOC 2026.}}
\author{Julia Chuzhoy\thanks{Toyota Technological Institute at Chicago, IL, USA. Email: {\tt cjulia@ttic.edu}. Supported in part by NSF award CCF-2402283 and NSF HDR TRIPODS award 2216899.} \and Sanjeev Khanna\thanks{Courant Institute School of Mathematics, Computing, and Data Science, New York University, NY, USA. Email: {\tt sanjeev.khanna@nyu.edu}. Supported in part by NSF award CCF-2625203 and AFOSR award FA9550-25-1-0107.} \and Junkai Song\thanks{Courant Institute School of Mathematics, Computing, and Data Science, New York University, NY, USA. Email: {\tt junkaisong@nyu.edu}.}}

\pagenumbering{gobble}

\maketitle
\thispagestyle{empty}

\begin{abstract}
In the fully dynamic maximal matching problem, the goal is to maintain a maximal matching in a graph undergoing an online sequence of edge insertions and deletions, while minimizing the update time. The problem has been studied extensively in the \emph{oblivious-adversary} setting, where randomized algorithms with polylogarithmic worst-case and constant amortized update time have been known for some time. A major challenge in this area has been designing an algorithm with non-trivial update time against an \emph{adaptive} adversary, who may explicitly tailor the update sequence to the algorithm's choices. In a recent breakthrough, Bernstein, Bhattacharya, Kiss, and Saranurak (STOC~2025; hereafter, BBKS25) obtained the first algorithms with sublinear in $n$ update time for this setting: namely, a randomized algorithm with $\tilde{O}(n^{3/4})$ amortized update time, and a deterministic algorithm with $\tilde{O}(n^{8/9})$ amortized update time. Our main result is a deterministic algorithm for fully dynamic  maximal matching with amortized update time $n^{1/2+o(1)}$.  

A powerful tool in dynamic matching is the use of matching
sparsifiers: sparse subgraphs that preserve enough information to
recover matchings with desired properties. Sparsifiers have been
successfully used for approximate maximum matching, yielding
sublinear update-time algorithms even against adaptive adversaries.
For maximal matching, however, this paradigm is not as natural, since
maximality must hold with respect to the entire graph, and so the
algorithm must be able to detect and repair violations across all edges.
Nevertheless, BBKS25 showed that the EDCS data structure can be
ingeniously repurposed as a verification-and-repair mechanism for
fully dynamic maximal matching against adaptive adversaries.

We introduce a new deterministic framework, referred to as the \emph{subgraph system}, which, in contrast to the EDCS data structure used by BBKS25, is purpose-built for verification and maintenance of maximality. 
The structure of the subgraph system is also carefully designed to allow efficient recursive refinements leading to stronger and stronger parameters.
This recursive approach yields our deterministic algorithm with $n^{1/2+o(1)}$ amortized  update time, and provides a new deterministic framework for one of the central graph optimization problems in the dynamic setting.

\end{abstract}

\end{titlepage}

\tableofcontents

\newpage
\pagenumbering{arabic}

\section{Introduction}

We study the fully dynamic maximal matching problem: given an $n$-vertex graph $G$ that undergoes an online sequence of edge insertions and deletions, the goal is to maintain a matching $M$, that is a maximal matching with respect to the current graph $G$ at all times. 
Maximal matching is a fundamental graph theoretic notion, and algorithms for computing it serve as a key subroutine in dynamic algorithms for coloring, scheduling, and approximate maximum matching.
As such, the problem has been studied extensively in various settings, including that of dynamic algorithms.

In the \emph{oblivious} adversary scenario, where the update sequence is fixed in advance,
fully dynamic maximal matching is by now a well-understood problem.
Baswana, Gupta, and Sen~\cite{BGS11,BGS18} designed the first algorithm with $O(\poly\log n)$ amortized update time,
that was later improved to an $O(1)$ amortized update time by Solomon~\cite{Sol16},
and to $O(\poly\log n)$ worst-case update time by Bernstein, Forster, and Henzinger~\cite{BFH21}.
In contrast, in the setting of the \emph{adaptive adversary}, where the update sequence may depend on the algorithm’s decisions,
progress was long stalled. 
Until very recently, the best known algorithm in this setting, due to Neiman and Solomon~\cite{NS15}, achieved an $O(\sqrt{m})$ worst-case update time, improving upon the earlier $O((n+m)^{1/\sqrt{2}})$ result of Ivković and Lloyd~\cite{IL93}; 
both these algorithms are deterministic. Thus in dense graphs, until very recently, even with randomization, no algorithm was known to achieve a better than $\Theta(n)$ amortized update time against an adaptive adversary—the same update time bound as the trivial approach of restoring maximality by exhaustively searching the neighborhood of a vertex affected by a deletion.
To summarize, while algorithms with $O(1)$ amortized update time have been known for the fully dynamic maximal matching problem in the oblivious adversary setting for some time, in the adaptive adversary setting, even the problem of achieving a sublinear in $n$ amortized update time was open until very recently.

To gain intuition for this stark discrepancy, consider a vertex $v$ of degree $\Omega(n)$. The most challenging updates are deletions of edges that lie in the current matching $M$ maintained by the algorithm, since such deletions require that the affected vertices are rematched. If an edge that is incident to $v$ in the matching $M$ that the algorithm maintains is selected uniformly at random among the $\Omega(n)$ incident edges of $v$ in $G$, then in \emph{expectation}, an oblivious adversary must perform $\Omega(n)$ deletions before removing the matching edge incident on $v$. Thus, even if the algorithm spends $\Theta(n)$ time to rematch $v$, it still achieves an $O(1)$ amortized update time. In contrast, an adaptive adversary can deliberately target matching edges incident to high-degree vertices, forcing  updates that require $\Omega(n)$ time if a naive rematching strategy is used.

As every maximal matching is a $1/2$-approximation to the maximum matching, and since the past decade has seen major progress on algorithms for \emph{approximate maximum matching} in dynamic graphs (e.g., \cite{OR10,GP13,BS15,BS16,BHN16,BCH17,BHI18,ACC+18,Waj20,BDL21,Kiss22,RSW22,BK22,Beh23,ABKL23,BKS23,BKSW24a,BKSW24b,ABR24,Liu24,BG24,AKK25}), it is natural to ask whether techniques employed in these algorithms can be extended to the maximal matching setting.  
A common tool underlying most these algorithms is the notion of a \emph{matching sparsifier}—a sparse subgraph that compresses the input fully dynamic graph $G$, while preserving enough information, so that an approximate matching can be recovered using only the edges of the sparsifier.
Prominent examples of this approach include the use of \emph{kernel} in algorithms for $(1/2-\epsilon)$-approximate matchings~\cite{BHI18}, the use of the \emph{Edge-Degree Constrained Subgraph (EDCS)}
in algorithms for $(2/3-\epsilon)$-approximate matchings~\cite{BS15,BS16}, and algorithms based on Szemerédi’s regularity lemma and
Ruzsa--Szemerédi graphs, that achieve $(1-o(1))$-approximation~\cite{ABKL23,BG24,AKK25}. 
These structures lead to a sublinear in $n$ update time, because a matching of the desired quality can be found while restricting the attention to the sparsifier graph, which is much sparser than $G$. Moreover, these sparsifier structures have proven robust enough to yield sublinear update times even against an adaptive adversary.  
In the context of maintaining a maximal matching, however, this sparsification paradigm is less natural, since maximality must hold with respect to the \emph{entire} graph, and hence requires the ability
to detect and repair violations that may involve edges outside of the
sparsifier. The core algorithmic challenge thus becomes one of
\emph{efficient verification and repair} of maximality.
An additional impediment is that, unlike the setting of approximate matchings, where, for example, a valid $(1 - \Theta(\eps))$-approximate matching remains so even if one ignores the next $\Theta(\eps\cdot n)$ updates~\cite{GP13,AKL19,Kiss22}, maximality is fragile: it can be violated by the insertion or the deletion of even a single edge.

A major breakthrough was achieved very recently by Bernstein, Bhattacharya, Kiss, and Saranurak~\cite{BBKS25},
who obtained a randomized algorithm with $\tilde{O}(n^{3/4})$ amortized update time and a deterministic algorithm with $\tilde{O}(n^{8/9})$ amortized update time; both algorithms work against an adaptive adversary. In their algorithm, \cite{BBKS25} ingeniously repurpose the EDCS data structure as a verification-and-repair mechanism, using it as the algorithmic backbone for maintaining a dynamic maximal matching. 
While \cite{BBKS25} demonstrate that the EDCS framework can be adapted to the maximal matching setting, their approach inherits some structural limitations, since EDCS was originally designed for approximate matchings rather than for enforcing maximality. In particular, an EDCS acts as a \emph{sparse subgraph} of the underlying graph, relying on degree constraints that ensure approximate matchings can be maintained efficiently. 
Some of the challenges of using EDCS as a backbone include a high update time for maintaining the EDCS data structure in the fully dynamic setting~\cite{GSSU22}, and difficulty of finding short augmenting paths efficiently in a dynamically changing graph.

\paragraph{Our Results.} Our main result is a \emph{deterministic} algorithm for maintaining a maximal matching in a fully dynamic graph with $n^{1/2+o(1)}$ amortized update time.
Our approach is different from that of  \cite{BBKS25}: we introduce a new deterministic framework referred to as the \emph{subgraph system}, 
designed specifically to allow efficient \emph{verification and repair of maximality} directly, rather than to sparsify the graph while preserving an approximate matching. Our subgraph system design serves two key purposes.  
First, it allows for efficient deterministic construction (amortized over a batch of updates) and efficient maintenance of some key data structures as the underlying graph evolves. In particular, the subgraph system is designed to allow a recursive composition that gradually yields stronger and stronger properties. Second, it offers a structured way to detect and restore violations of maximality through bounded local operations. In this respect, the subgraph system provides for maximal matching what matching sparsifiers provide for approximate matching: an explicit, combinatorial framework that retains all information necessary to maintain maximality.

\paragraph{Adaptive vs. Oblivious Adversary.}
It is worth noting that maximal matching is not unique in suffering a large performance gap between the oblivious and the adaptive adversary settings. A similar dichotomy
arises in other dynamic 
 problems, such as maximal independent set and $(\Delta+1)$-coloring.
In the oblivious adversary setting, algorithms with polylogarithmic or even constant amortized update times are known for both problems
\cite{BDH+19,CZ19,BGK+22,HP22},
while in the adaptive adversary setting, the best known update times are polynomial in $n$
\cite{GK18,DZ18,BRW25,FH25}.
More generally, such gaps are a recurrent phenomenon in dynamic graph problems. 
A very recent work by Bernstein, Bhattacharya, Fischer, Kiss, and Saranurak~\cite{BBFKS25} establishes conditional update-time separations between adaptive and oblivious adversaries for multiple dynamic symmetry-breaking problems, including maximal independent set (MIS) and maximal clique. 
Specifically, they show that, under the combinatorial Boolean Matrix Multiplication (BMM) hypothesis, every algorithm for the \emph{incremental MIS} problem that works against an adaptive adversary must incur $n^{1-o(1)}$ amortized update time, whereas, in the oblivious adversary setting, polylogarithmic update time algorithms are known. 
Similarly, \cite{BBFKS25} show that, under the 3SUM or the APSP hypotheses, any algorithm for \emph{decremental maximal clique} that works against an adaptive adversary requires \( \Omega(\Delta / n^{o(1)}) \) amortized update time when the initial maximum degree is \( \Delta \le \sqrt{n} \), while in the oblivious adversary setting, a $\poly\log (n)$ update time is achievable. 
Together, these results show that the performance gap between oblivious and adaptive adversaries is a pervasive barrier across dynamic symmetry-breaking and related graph problems.

\subsection{Our Results and Techniques}
\label{subsec: techniques}
Our main result is summarized in the following theorem.

\begin{theorem}\label{thm: main}
There is a deterministic algorithm that, given an $n$-vertex graph $G$ initially containing no edges and undergoing an online sequence of edge insertions and deletions, maintains a maximal matching in $G$ with amortized update time $n^{1/2+o(1)}$.
\end{theorem}

We now provide a high-level overview of the proof of \Cref{thm: main} and compare it with the approach of \cite{BBKS25}. For clarity of exposition, we consider a slightly different setting, where the input $n$-vertex graph $G$ initially may contain edges, but the only updates that it undergoes are edge deletions, so $G$ is decremental and not fully dynamic. We also assume that the update sequence is sufficiently long, so, for example it contains at least $n^{4/3}$ updates. This setting seems to capture the main challenges of the problem, and extending the algorithm to the fully dynamic setting with an arbitrary number of updates is not difficult.
We start by describing an algorithm with $\tilde O(n^{2/3})$ amortized update time, since this algorithm allows us to convey some of our main ideas, and to motivate the more involved data structures used in the algorithm with $n^{1/2+o(1)}$ amortized time.

\subsection*{Algorithm with $\tilde O\left(n^{2/3}\right )$ Amortized Update Time}
The main new combinatorial object that we introduce is the \emph{$z$-subgraph system}, that we describe next; for clarity of exposition, we omit some technical details that are not essential.

\paragraph{The $z$-subgraph system.}
Let $G$ be a graph and let $1\leq z\leq n$ be an integral parameter. A $z$-subgraph system for $G$ consists of a subset $M\subseteq E(G)$ of edges, and a partition $(A,B,U)$ of the vertices of $G$; we denote by $S=A\cup B$. We require that every vertex in $S$ is incident to exactly $z$ edges of $M$, and every vertex in $U$ is incident to at most $z$ edges of $M$. We also require that, for every vertex $u\in U$, at most $z$ neighbors of $u$ in $G$ may lie in $U$, so $|N_G(u)\cap U|\leq z$. Additionally, we require that the following two crucial properties hold:

\begin{properties}{P}
	\item \label{prop intro UB edges} For every vertex $u\in U$, $|N_G(u)\cap B|\leq 2z$; and
	\item \label{prop: intro A vertices} For every vertex $a\in A$, for every edge $(a,v)\in M$ incident to $a$, $v\in S$ must hold.
\end{properties}

Lastly, the subgraph system must contain, for every vertex $u\in U$, the list $\Lambda(u)=N_G(u)\cap (B\cup U)$ of vertices, whose length must be $O(z)$, and, for every vertex $a \in A$, the list $L(a)=N_G(a)\cap U$ of vertices, whose length may be arbitrary.
We denote by $\Lambda=\set{\Lambda(u)}_{u\in U}$ and $L=\set{L(v)}_{v\in A}$, and we denote the entire $z$-subgraph system by $(S,B,U,M,\Lambda,L)$.

We note that the notion of the $z$-subgraph system is somewhat similar to that of the \emph{Edge-Degree Constrained Subgraph} (EDCS) data structure, that was used in the algorithm of \cite{BBKS25} for dynamic maximal matching. The EDCS data structure was introduced  by  Bernstein and Stein \cite{BS15,BS16} in the context of computing approximate maximum matchings, and it was used in numerous algorithms since (see \cite{BBKS25} for further references). 
Given a pair $0<B^-<B$ of parameters, a $(B, B^{-})$-EDCS of a
graph $G$ is a subgraph $G'\subseteq G$ with $V(G')=V(G)$, that has the following two properties. First, for every edge $(u,v) \in E(G')$, $\deg_{G'}(u) + \deg_{G'}(v) \leq B$ must hold. Second, for every edge $(u,v) \in E(G) \setminus E(G')$,  $\deg_{G'}(u) + \deg_{G'}(v) \geq B^{-}$. Note that the maximum vertex degree in $G'$ is at most $B$, ensuring that $G'$ is a  sparse graph, regardless of the density of $G$, provided that $B$ is chosen to be sufficiently small. The algorithm of \cite{BBKS25} for fully dynamic maximal matching maintains a $(B,(1-\epsilon)B)$-EDCS $G'$ of the input fully dynamic graph $G$, for some parameters $0<\epsilon<1$ and $B>0$,  by using the deterministic algorithm of \cite{GSSU22}, whose worst-case update time is $O\left(\frac{n}{\eps\cdot B}\right)$. In the algorithm of \cite{BBKS25}, the vertices of $G$ are partitioned into three subsets, based on their degree in $G'$: set $\vhigh$ contains all vertices $v$ with $\deg_{G'}(v)\ge \left(\frac{1}{2}+\Theta(\epsilon)\right)B$; set $\vlow$ contains all vertices $v$ with $\deg_{G'}(v)\leq \left(\frac{1}{2}-\Theta(\epsilon)\right)B$; and set $\vmed$ contains all remaining vertices.
The properties of the EDCS data structure are then exploited by \cite{BBKS25} in order to maintain a matching $M'$, in which every vertex of $\vhigh$, and the vast majority of the vertices of $\vmed$ are matched. Additionally, they maintain a maximal matching $M''$ in $G$, that is obtained by extending the matching $M'$ to the vertices of $\vlow$. In order to do so, they exploit the fact that each such vertex may only have relatively few neighbors in $\vlow$. One can think of the parameter $B$ in their EDCS data structure as playing a similar role to the parameter $z$ in the $z$-subgraph system, with vertex sets $\vhigh,\vmed$ and $\vlow$ corresponding to sets $A,B$ and $U$, respectively, in the $z$-subgraph system. However, there are several crucial differences between the two data structures. The first difference is the sharper degree bounds that are required for the vertices of $S$ in the subgraph system. For example, this allows us to decompose the set $M$ of edges into $(z+1)$ disjoint matchings $M_1,\ldots,M_{z+1}$ so that, for every vertex $v\in S$, all but at most one of the matchings $M_i$ contain an edge incident to $v$. Additionally, Properties \ref{prop intro UB edges} and
\ref{prop: intro A vertices} of the $z$-subgraph system are specifically designed so as to make the task of maintaining a maximal matching in $G$ easier, and it is not clear how to ensure these properties in the EDCS data structure. We also believe that the $z$-subgraph system is a simpler data structure, in that, it is both easier to construct and manipulate. In fact our final algorithm for maximal matching with $n^{1/2+o(1)}$ amortized update time uses a \emph{multi-level} generalization of this data structure, that is constructed from the basic data structure using a recursive algorithm. The multi-level subgraph system is specifically designed to allow an efficient recursive refinement of the parameters of the subgraph system.

Our algorithm for dynamic maximal matching uses two main algorithmic tools. The first tool allows us to construct a $z$-subgraph system of a graph $G$ in time $\tilde O(|V(G)|+|E(G)|)$. The second tool allows us to maintain a maximal matching in a graph $G$ undergoing a bounded number of edge deletions, given a $z$-subgraph system for the initial graph $G$. We now describe each of these tools in turn.

\paragraph{Constructing a $z$-subgraph system.}
We provide a simple algorithm, that, given an $n$-vertex and $m$-edge graph $G$, constructs a $z$-subgraph system $(S,B,U,M,\Lambda,L)$ for $G$. 
 The algorithm consists of two steps. In the first step, we construct a $z$-subgraph system for $G$ that has all required properties except for, possibly, Property \ref{prop intro UB edges}. Then in the second step we ``fix'' this subgraph system, to ensure that it has all required properties.

The algorithm for the first step is quite straightforward. We construct a maximal set $M\subseteq E(G)$ of edges with the following property: every vertex of $G$ has at most $z$ edges incident to it in $M$. Such an edge set can be constructed in $O(n + m)$ time using a straightforward greedy algorithm. For every vertex $v\in V(G)$, let $m(v)$ denote the number of edges of $M$ incident to $v$. We then let $S$ contain all vertices $v$ with $m(v)=z$ and $U$ contain all remaining vertices of $G$. From the maximality of $M$, it is immediate to verify that, for every vertex $u\in U$, $|N_G(u)\cap U|\leq z$. We also partition the vertex set $S$ into two subsets: set $A$ containing every vertex $v\in S$ such that every edge $(v,v')$ of $M$ incident to $v$ connects $v$ to a vertex in $S$; and set $B$ containing all remaining vertices of $S$. Finally, we initialize the lists $\set{\Lambda(u)}_{u\in U}$ and $\set{L(v)}_{v\in A}$, in time $O(n + m)$. It is easy to verify that this data structure has all properties required from the $z$-subgraph system, except for, possibly, Property  \ref{prop intro UB edges}.

In the second step, we modify the subgraph system computed in the first step, in order to ensure that it has Property  \ref{prop intro UB edges}. In order to do so, we process every vertex of $U$ one by one. When a vertex $u\in U$ is processed, we check whether $|N_G(u)\cap B|\ge z$ holds. If so, we insert edges connecting $u$ to vertices of $B$ into $M$, until $m(u)=z$ holds, and $u$ can be moved from $u$ to $S$, where it joins the set $B$. Note that, if an edge $(u,v)$ connecting $u$ to a vertex $v\in B$ is inserted into $M$, then, in order to ensure that $m(v)=z$ continues to hold, we need to delete another edge connecting $v$ to a vertex of $U$ from $M$. Whenever a vertex $v\in B$ loses all edges of $M$ connecting it to vertices of $U$ (that is, Property \ref{prop: intro A vertices} holds for $v$), we move it from $B$ to $A$. This ensures that, as long as a vertex $v$ remains in $B$, there is an edge of $M$ connecting it to a vertex of $U$; this, in turn, allows us to perform ``edge switching'' described above, when an edge connecting $v$ to some new vertex $u\in U$ is inserted into $M$, and an edge connecting $v$ to another vertex of $U$ is deleted from $M$. It is not difficult to efficiently implement this step, and to update the lists $\set{\Lambda(u)}_{u\in U}$ and $\set{L(v)}_{v\in A}$ at the end of this step, in time $\tilde O(n+m)$. At the end of this procedure we obtain a valid $z$-subgraph system for $G$.

Our algorithm uses the parameter $z=\floor{n^{2/3}}$. Unlike the algorithm of \cite{BBKS25}, that maintains the EDCS data structure for the input dynamic graph $G$ over the course of the entire algorithm, we only compute the $z$-subgraph system for the graph $G$ from time to time. Specifically, we partition our algorithm into \emph{phases}, each of which spans exactly $r=\floor{n^{4/3}}$ updates by the adversary (except for the last phase that may be shorter). We compute the $z$-subgraph system for $G$ from scratch at the beginning of every phase, in time $\tilde O(|E(G)|+|V(G)|)\leq \tilde O(n^2)$. Since this computation is amortized over $r$ updates to $G$ that occur during the phase, we get that computing the $z$-subgraph system at the beginning of every phase can be executed in amortized $\tilde O(n^{2/3})$ update time. Our second main tool is an algorithm that exploits the $z$-subgraph system in the initial graph $G$ (at the start of a phase) in order to maintain a maximal matching in $G$ over the course of a phase spanning $r$ edge deletions. We describe this algorithm next.

\paragraph{Maintaining a maximal matching via the $z$-subgraph system.} 
We now describe our algorithm for implementing a single phase.
We assume that we are given an $n$-vertex graph $G$ (corresponding to the input dynamic graph $G$ at the beginning of the current phase), that undergoes an online sequence of $r\geq n$ edge deletions. We also assume that we are given a $z$-subgraph system $\sset=(S,B,U,M,\Lambda,L)$ for the initial graph $G$. We now describe an algorithm that exploits this subgraph system in order to maintain a maximal matching in $G$. As our first step, we use the almost linear-time \emph{deterministic} algorithm of \cite{ABB+26} for edge coloring in order to compute, in time $O\left(n^{1+o(1)}\cdot z\right )$, a partition $(M_1,\ldots,M_{z+1})$ of the set $M$ of edges from the subgraph system into $z+1$ disjoint matchings (here, each color class defines a matching). Recall that, from the properties of the subgraph system, every vertex $v\in S$ has $z$ edges incident to it in $M$; so there is at most one index $1\leq i\leq z+1$, such that $v$ is not matched in $M_i$. Therefore, in a typical matching $M_i$, all but at most $\frac{2|S|}{z}$ vertices of $S$ are matched. 
We assume without loss of generality that matching $M_1$ has this property. Over the course of the algorithm, we may modify the matching $M_1$ by deleting some edges from it, and by inserting edges of $M$ into it, so $M_1\subseteq M$ will always hold. The remaining matchings $M_2,\ldots,M_{z+1}$ only undergo the deletions of edges that the adversary deletes from the graph $G$. 
Since $G$ undergoes a sequence of $r$ edge deletions, for most indices $2\leq i\leq z+1$, the number of edges of $M_i$ that are deleted by the adversary over the course of the algorithm is bounded by $\frac{2r}{z}$. It is then not hard to see that, at all times, there is an index $2\leq i\leq z+1$, such that the number of vertices of $S$ that are not matched by $M_i$ is bounded by $\frac{2|S|}{z}+\frac{4r}{z}\leq \frac{32r}{z}$; denote this bound by $\tau$. 
We partition each phase of the algorithm into \emph{subphases}, each of which spans $\floor{\frac{r}{z}} \approx n^{2/3}$ adversarial updates to graph $G$, so that the number of subphases is bounded by $2z$. We will ensure that, at the beginning of every subphase, all but at most $\tau=\frac{32r}{z}$ vertices of $S$ are matched by $M_1$. Over the course of each subphase, our algorithm may delete edges from $M_1$, but it ensures that every edge deletion by the adversary may lead to at most two edge deletions from $M_1$. Since at most $\frac{r}{z}$ adversarial edge deletions may occur during a subphase, this ensures that, overall, at most $\frac{3r}{z}$ edges may be deleted from $M_1$ over the course of a subphase, and so at most $\frac{6r}{z}$ additional vertices of $S$ may become unmatched during a subphase. Therefore, the following crucial invariant holds over the course of the subphase.

\begin{properties}{I}
\item	\label{inv: vertices of S} Throughout the subphase, all but at most $2\tau=\frac{64r}{z}$ vertices of $S$ are matched by $M_1$.
\end{properties}

Over the course of a subphase, our algorithm maintains a matching $M^*$ that contains all edges of $M_1$, such that, at all times, $M^*$ is a maximal matching in $G$.
Recall that our algorithm ensures that $M_1\subseteq M$ holds at all times. Since, from Property \ref{prop: intro A vertices} of the subgraph system, for every vertex $a\in A$, every edge of $M$ incident to $a$ connects it to a vertex of $S$, from Invariant \ref{inv: vertices of S}, we get that all but at most $2\tau$ vertices of $A$ are matched to vertices of $S$ by $M_1$. Recall that $M_1\subseteq M^*$, so the only way that a vertex $a\in A$ is matched to a vertex of $U$ by $M^*$ is if $a$ is not matched by $M_1$. Therefore, the following property is also guaranteed to hold over the course of the subphase:

\begin{properties}[1]{I}
	\item	\label{inv: vertices of A} Throughout the subphase, the number of vertices of $A$ that are matched by $M^*$ to vertices of $U$ is bounded by $2\tau=\frac{64r}{z}$.
\end{properties}

Recall that, at the beginning of a subphase, we are given a matching $M_1$ in $G$, such that all but at most $\tau$ vertices of $S$ are matched by  $M_1$. Throughout the algorithm, we denote by $\hat S$ the collection of all vertices $v\in S$ that are currently not matched by $M^*$. From Property \ref{inv: vertices of S}, $|\hat S|\leq 2\tau$ always holds. We start by setting $M^*=M_1$, and then we extend $M^*$ to be a maximal matching in a relatively straightforward manner, by inspecting all vertices $u\in U$ and checking whether we can insert an edge connecting $u$ to vertices of $U\cup B\cup \hat S$ into $M^*$. We also iteratively check whether an edge connecting a pair of vertices in $\hat S$ can be inserted into $M^*$. Since $|\hat S|\leq O(r/z)$ and, for every vertex $u\in U$, $|N_G(u)\cap U|\leq |\Lambda(u)|\leq O(z)$, this can be done in time $\tilde O(nz+nr/z)$.

Throughout the algorithm, we say that a vertex $v\in V(G)$ is \emph{settled} if either (i) some edge of $M^*$ is incident to $v$; or (ii) for every vertex $y\in N_G(v)$, some edge of $M^*$ is incident to $y$. It is easy to verify that matching $M^*$ is maximal if and only if every vertex of $G$ is settled. Whenever an edge $e=(u,v)$ is deleted from $M^*$, we execute a procedure that attempts to \emph{rematch} each of its endpoints; in order to rematch a vertex $x$, we need to either insert an edge incident to $x$ into $M^*$, or verify that every vertex in $N_G(x)$ is currently matched by $M^*$. In other words, we need to ensure that $x$ is settled. 

Intuitively, it is easy to rematch vertices of $U$ efficiently. Specifically, consider a vertex $u\in U$ that needs to be rematched. Recall that $|N_G(u)\cap U|\leq O(z)$, so we can check, in time $O(z)$, whether some vertex $u'\in N_G(u)\cap U$ is currently not matched by $M^*$, and, if so, insert the edge $(u,u')$ into $M^*$. Additionally, we can scan every vertex of $\hat S$ (recall that $\hat S$ contains all vertices of $S$ that are currently not matched, and that $|\hat S|\leq O(\tau)\leq O(r/z)$), to check whether $u$ can be matched to a vertex of $S$ that is currently unmatched. Therefore, the procedure for rematching a vertex of $U$ can be implemented in time $\tilde O\left(z+\frac{r}{z}\right )$. Moreover, rematching the vertices of $B$ is also relatively easy, as long as we maintain an appropriate data structure: namely, a directed graph $H$, whose vertex set is $B\cup U$, and whose edge set contains, for every vertex $u\in U$ that is currently unmatched in $M^*$, an edge connecting $u$ to every vertex in $N_G(u)\cap (B\cup U)=\Lambda(u)$. Recall that $|\Lambda(u)|\leq O(z)$, so, whenever a vertex $u$ switches its status between matched and unmatched in $M^*$, we can insert all edges connecting $u$ to vertices of $\Lambda(u)$, or delete them from $H$, as needed, in time $\tilde O(z)$. Whenever a vertex $b\in B$ becomes unmatched, we can then check whether $H$ contains an edge $(u,b)$ entering $b$ in $H$. If so, then vertex $u$ must be currently unmatched by $M^*$, and we can insert the edge $(u,v)$ into $M^*$. If $b$ has no incoming edges in $H$, then we are guaranteed that every vertex in $N_G(b)\cap U$ is currently matched by $M^*$. We then scan the vertices of $\hat S$ in order to check whether $b$ can be matched to a vertex of $S$ that is currently unmatched by $M^*$, in time $O(|\hat S|)\leq O(r/z)$. Overall, we can execute a rematching procedure for a vertex of $B\cup U$ in time $\tilde O(z+r/z)$, and we can maintain the graph $H$ that allows us to rematch the vertices of $B$ efficiently in time $\tilde O(z)$ per update.

The main challenge of the algorithm is in rematching the vertices of $A$. The problem is that a vertex $a\in A$ may be incident to a large number of vertices of $U$ in $G$ (recall that all such vertices are stored in list $L(a)$), and, moreover, every vertex of $U$ may be incident to a large number of vertices of $A$. Therefore, whenever a vertex $a\in A$ loses an edge of $M^*$ that is incident to it, it is not clear how to check whether $a$ can be matched to a vertex of $U$ efficiently. Our main idea is to give the vertices of $A$ a \emph{priority} in rematching, exploiting the fact that vertices of $B\cup U$ are relatively easy to rematch. Therefore, when attempting to rematch a vertex $a\in A$, we allow ourselves to delete at most one edge from $M^*$ (and from $M_1$), provided that this edge deletion only affects vertices in $B\cup U$, which we can then rematch efficiently. Specifically, recall that, from Invariant \ref{inv: vertices of A}, at most $2\tau=\frac{64r}{z}$ vertices of $U$ may be matched to vertices of $A$ at all times. Therefore, 
either $L(a)=N_G(a)\cap U$ contains at most $\frac{64r}{z}$ vertices, or among the first $\frac{64r}{z}+1$ vertices of $L(a)$, there must be some vertex $u$ that is not currently matched to a vertex of $A$ by $M^*$. In the latter case, we insert the edge $(a,u)$ into $M^*$, and, if another edge $(u,u')$ incident to $u$ lies in $M^*$, we delete this edge from $M^*$ (and from $M_1$, if $(u,u')\in M_1$). Notice that $u'\in B\cup U$ must hold, so we can rematch $u'$ efficiently. Overall, in order to rematch a vertex  $a\in A$, we only need to scan $O(r/z)$ vertices of $L(a)$, and $O(r/z)$ vertices of $\hat S$. Therefore, as long as Invariant \ref{inv: vertices of A} holds, we can maintain a maximal matching in $M^*$ in amortized update time $\tilde O\left (z+\frac{r}{z}\right )$.

From our discussion so far, in order to guarantee that Invariant \ref{inv: vertices of A} holds, it is enough to ensure that, at the beginning of every subphase, all but at most $\tau=\frac{32r}{z}$ vertices of $S$ are matched by $M_1$. As shown above, we are guaranteed that, at all times, there is an index $2\leq i\leq z+1$, such that all but at most $\tau$ vertices of $S$ are matched by $M_i$. Therefore, the most natural approach would be to simply replace $M_1$ with a matching $M_i$ that has this property at the beginning of every subphase. The problem is that we would then  need to extend this matching to a maximal one by inspecting the vertices of $U$ as before, and to construct the data structure $H$ from scratch, which is too time consuming (we can afford to do so once at the beginning of the algorithm, but not at the beginning of every phase). Instead, we use the graph $M_i\cup M_1$ in order to gradually augment the matching $M_1$ via augmenting paths originating at vertices of $S$ that are unmatched by $M_1$, so as to reduce the number of such vertices. This approach ensures that only $O(r/z)$ vertices of $G$ may switch their status between matched and unmatched in $M_1$, which, in turn, allows us to repair the resulting matching $M^*$ efficiently. At the end of this ``repair'' procedure, we are guaranteed that all but at most $\tau$ vertices of $S$ are matched by $M_1$, and that $M^*$ is a maximal matching for $G$ with $M_1\subseteq M^*$.

Overall, our algorithm spends $O\left(n^{1+o(1)}\cdot z\right )$ time
in order to compute the initial collection $M_1,\ldots,M_{z+1}$ matchings. After that, it spends $\tilde O\left (z+\frac{r}{z}\right )$ time per update. Intuitively, we spend $O(z)$ time per update in order to rematch vertices of $U$ (as each such vertex may have up to $\Theta(z)$ neighbors in $U\cup B$); equivalently, every time a vertex of $U$ switches its status between matched and unmatched in $M^*$, we may need to spend up to $\Theta(z)$ time in order to update its outgoing edges in graph $H$. Additionally, since the adversary may delete $r$ edges over the course of the algorithm, it is possible that, for every matching $M_i$, $\Omega(r/z)$ vertices of $S$ become unmatched in it. Therefore, in particular, $M_1$ may contain as many as $\Theta(r/z)$ vertices of $S$ that are currently unmatched. Whenever any vertex $v$ of $G$ becomes unmatched, we need to inspect all vertices of $\hat S$, in order to check whether $v$ can be matched to these vertices, which may require $\Theta(r/z)$ time. (Note  a vertex of $U$ can be incident to an arbitrarily large number of vertices of $A$ in $G$ and vice versa, so we do not have a more efficient way to check whether a newly unmatched vertex can be rematched to vertices of $A$).
We note that it is not hard to adapt this algorithm to work with a slightly more relaxed definition of the $z$-subgraph system, where we only require that, for every vertex $v\in S$, $z-k\leq m(v)\leq z$ holds for some parameter $k<z$; in this case, the amortized update time of the algorithm is bounded by $\tilde O\left (z+\frac{rk}{z}\right )$. We will use this relaxed definition of the subgraph system later.

Overall, including the time required to compute the $z$-subgraph system at the beginning of every phase, our algorithm spends time $\tilde O\left(n^{2}+n^{1+o(1)}\cdot z+r\cdot \left(z+\frac{r}{z}\right )\right )$ per phase, which leads to \\$\tilde O\left(\frac{n^{2}+n^{1+o(1)}\cdot z}{r}+z+\frac{r}{z}\right )$ amortized update time, as every phase contains $r$ adversarial updates to $G$. Substituting $z=\floor{n^{2/3}}$ and $r=\floor{n^{4/3}}$, we get $\tilde O\left (n^{2/3}\right )$ update time overall.

\subsection*{Obtaining a Faster Algorithm}

Observe that our algorithm for maintaining a maximal matching from a $z$-subgraph system has amortized update time $\tilde O\left(\frac{r}{z}+z\right )$; at a very high level, additive factor $z$ accounts for rematching vertices of $U$, since each such vertex may have up to $\Theta(z)$ neighbors that lie in $U$, while additive factor $\frac{r}{z}$ accounts for the fact that every matching $M_i$ may have up to $\Theta\left(\frac{r}{z}\right )$ vertices of $A$ that are not matched in it. Both these bounds appear to be fundamental bottlenecks for this approach. Therefore, in order to obtain a faster amortized update time, we need to make sure that both the parameters $z$ and $r$ are significantly smaller; in particular, the length $r$ of every phase must be reduced. However, computing the $z$-subgraph system from scratch at the beginning of every phase requires $\Omega(|E(G)|)$ time, which may be as high as $\Omega(n^2)$, and decreasing the length of each phase means that this computation needs to be repeated more often, which will lead to a higher amortized update time.

In order to overcome this difficulty, we use a recursive approach. We use a parameter $r_1$ that is sufficiently high, and partition the timeline into \emph{level-1 phases}, each of which spans $r_1$ adversarial edge deletions (except for possibly the last phase that may be shorter). At the beginning of every level-$1$ phase, we construct a $z_1$-subgraph system $\sset$ for $G$ exactly as before; in order to ensure that this subgraph system does not sustain too much damage over the course of the phase due to edge deletions, we need to ensure that the parameter $z_1$ is sufficiently high. However, 
we do not use this subgraph system directly in order to maintain a maximal matching in $G$. Instead, we design an algorithm that, given a $z_1$-subgraph system for $G$, constructs a new $z_2$-subgraph system, for any chosen parameter $z_2<z_1$, in time roughly $O\left(n^{1+o(1)}\cdot z_1\right )$; note that, if the input graph $G$ is sufficiently dense and the parameter $z_1$ is not too high, this running time may be significantly lower than $\Theta(|E(G)|)$. We then select an appropriate parameter $r_2<r_1$ and partition every level-1 phase into \emph{level-2 phases}, each of which spans $r_2$ adversarial edge deletions. Consider now a level-$1$ phase $\Phi$, and let $\sset$ be the $z_1$-subgraph system computed at the beginning of Phase $\Phi$. Let $\Phi'\subseteq \Phi$ be a level-$2$ phase that is contained in $\Phi$.
We denote by $\tau$ and $\tau'$ the times when the phases $\Phi$ and $\Phi'$ started,  respectively, and we denote by $G\attime$ and $G\attime[\tau']$ the graph $G$ at time $\tau$ and $\tau'$, respectively.
 Let $E_D$ denote the set of all edges that the adversary deleted from $G$ during the time interval $[\tau,\tau')$. We do not delete the edges of $E_D$ from the subgraph system $\sset$. Instead, we use $\sset$ in order to construct a subset $E'_D\subseteq E_D$ of at most $\frac{z_2}{z_1}\cdot |E_D|$ edges, and a $z_2$-subgraph system $\sset'$ for the graph $G'=G\attime\setminus (E_D\setminus E'_D)$ (note that, equivalently, $G'$ is the graph that is obtained from the current graph $G\attime[\tau']$ by inserting the edges of $E'_D$ into it). Then we apply the algorithm for maintaining a maximal matching in $G'$, given a $z_2$-subgraph system $\sset'$ for it. This algorithm first processes the deletions of the edges of $E'_D$ from $G'$ as adversarial updates to ensure that $G'=G\attime[\tau']$, before processing the edge deletions from $G$ that the adversary performs during the level-$2$ phase $\Phi'$. Since our algorithm for constructing a $z_2$-subgraph system from a $z_1$-subgraph system is more efficient than the $\tilde O(|E(G)|)$-time algorithm that constructs a subgraph system for $G$ from scratch (if $G$ is sufficiently dense), we can afford to execute it more often. This allows us to ensure that the length $r_2$ of level-2 phases is shorter, and the corresponding parameter $z_2$ is smaller, leading to a faster running time of the algorithm that maintains a maximal matching given a subgraph system.

We now briefly describe our algorithm for constructing a $z_2$-subgraph system from a $z_1$-subgraph system, in order to motivate the multi-level subgraph system that our algorithm uses. Recall that we are given as input a graph $G$ (which we previously denoted by $G'$), a $z_1$-subgraph system $\sset=(S,B,U,M,\Lambda,L)$ for $G$, and a subset $E_D\subseteq E(G)$ of edges (these are the edges that were already deleted from $G$ by the adversary but we do not implement these edge deletions yet, that is, $E_D\subseteq E(G)$ currently holds). Our goal is to construct a subset $E'_D\subseteq E_D$ of at most $\frac{z_2}{z_1}\cdot |E_D|$ edges, and a $z_2$-subgraph system $\sset'=(S',B',U',M',\Lambda',L')$ for the graph $\tilde G= G\setminus (E_D\setminus E'_D)$.
In order to do so, we use the almost linear-time deterministic algorithm of \cite{ABB+26} for edge coloring in order to compute, in time $O\left(n^{1+o(1)}\cdot z_1\right )$, a partition $(M_1,\ldots,M_{z_1+1})$ of the set $M$ of edges from the subgraph system $\sset$ into $z_1+1$ disjoint matchings. For all $1\leq i\leq z_1+1$, we denote by $w_i=|M_i\cap E_D|$, and we reindex the matchings so that $w_1\leq w_2\leq\cdots\leq w_{z_1+1}$ holds. We then let $M'=\bigcup_{i=1}^{z_2}M_i$ and $E'_D=E_D\cap M'$. It is easy to verify that $|E'_D|\leq\frac{z_2}{z_1}\cdot |E_D|$, as required. Let $\tilde G=G\setminus (E_D\setminus E'_D)$, and let $\sset'=(S',B',U',M',\Lambda',L')$ be a subgraph system for $\tilde G$, where $M'$ is the set of edges we have just constructed, $S'=S$, $B'=B$, and $U'=U$; for now we ignore the data structures $\Lambda'$ and $L'$ but we will revisit them later. For every vertex $v\in V(\tilde G)$, let $m'(v)$ denote the number of edges of $M'$ incident to $v$. Recall that, from the properties of the $z_1$-subgraph system $\sset$, for every vertex $v\in S$, $m(v)=z_1$ held, so there is at most one matching $M_i$ that does not contain an edge incident to $v$. It is then easy to verify that, for every vertex $v\in S'$, $z_2-1\leq m'(v)\leq z_2$ must hold, and for every vertex $v\in V(\tilde G)$, $m'(v)\leq z_2$. The subgraph system $\sset'$ therefore \emph{almost} has all required properties, except for the following: first, it is required that, for every vertex $u\in U$, $|N_{\tilde G}(u)\cap U|\leq z_2$ holds, but we are only guaranteed that 
$|N_{\tilde G}(u)\cap U|\leq z_1$ from the properties of the $z_1$-subgraph system $\sset$. Additionally, it is required that, for every vertex $u\in U$, $|N_{\tilde G}(u)\cap B|\leq z_2$ holds, but we are only guaranteed that
$|N_{\tilde G}(u)\cap B|\leq 2z_1$ from the properties of the $z_1$-subgraph system $\sset$. In our next step, we ``fix'' this problem by processing every vertex $u\in U$ one by one. Recall that, for each such vertex $u$, we are given the list $\Lambda(u)=N_{G}(u)\cap (B\cup U)$ of vertices as part of the $z_1$-subgraph system $\sset$. If $\Lambda(u)$ contains at least $ z_2$ vertices of $B'\cup U'$, we insert into $M'$ additional edges connecting $u$ to vertices of $U'\cup B'$, so that $m'(u)=z_2$ holds, and then move $u$ from $U'$ to $S'$, where it joins the vertex set $B'$. As before, if we insert into $M'$ an edge $(u,v)$ with $v\in B'$, then we may need to delete another edge connecting $v$ to a vertex of $U'$ from $M'$ to ensure that $m'(v)\leq z_2$ continues to hold. As before, every vertex of $B'$ that loses all edges of $M'$ connecting it to vertices of $U'$ is ``promoted'' to set $A'$. It is not difficult to implement this step to run in time $\tilde O(nz_1)$, by exploiting the lists $\Lambda(u)$ for vertices $u\in U$ that are given as part of the subgraph system $\sset$. 

Generally, this approach allows us to compute the desired $z_2$-subgraph system $\sset'$ for the graph $\tilde G$, except for the following critical issue: at the end of this procedure we need to construct, for every vertex $v\in A'$, a list $L'(v)=N_{\tilde G}(v)\cap U'$ of vertices. In order to do so, we can exploit the data structures $L\cup \Lambda$ that are given as part of the $z_1$-subgraph system $\sset$. The problem is that, for every vertex $a\in A$, the initial list $L(v)$ may be arbitrarily long, and it is possible that every vertex $u\in U$ has a large number of neighbors in $\tilde G$ that lie in $A$. Therefore, if $u$ is promoted from $U'$ to $S'$, we may need to update a large number of the lists $L'(v)$, which we cannot afford. At the same time, these lists are crucially used by our algorithm for maintaining a maximal matching in $G$, so it is important to maintain them correctly.

In order to overcome this difficulty, we slightly change the definition of the subgraph system, to a \emph{multi-level subgraph system}. The new definition allows us to avoid the difficulty outlined above, so we do not need to update the lists $L(v)$ for vertices that lie in the initial set $A$. We then slightly modify the algorithm for maintaining a maximal matching in $G$, so that it can work with this modified definition of the subgraph system.

We now provide the description of the 2-level subgraph system. In order to provide intuition and motivation for its construction, consider the initial $z_1$-subgraph system $\sset=(S,B,U,M,\Lambda,L)$, and the (incomplete) new $z_2$-subgraph system $\sset'=(S',B',U',M')$ that we constructed above (recall that we did not construct the required data structures $\Lambda'$ and $L'$).

The $2$-level $z$-subgraph system  for a graph $G$ is defined very similarly to the basic $z$-subgraph system: it consists of a partition $(A,B,U)$ of $V(G)$, where we denote by $S=A\cup B$, and a set $M$ of edges of $G$. As before, for every vertex $v\in V(G)$, we denote by $m(v)$ the number of edges of $M$ incident to $v$. We require that, for every vertex $v\in S$, $z-1\leq m(v)\leq z$ holds, and for every vertex $u\in U$, the following three properties hold: $m(u)\leq z$; $|N_G(u)\cap U|\leq z$; and $|N_G(u)\cap B|\leq 2z$.
For every vertex $u\in U$, the subgraph system also needs to contain the list $\Lambda(u)=N_G(u)\cap (B\cup U)$ of vertices, whose length must be $O(z)$. So far, the definition of the $2$-level $z$-subgraph system is essentially identical to that of the basic $z$-subgraph system. The main new ingredient is that the 2-level subgraph system must contain a partition $(A_1,A_2)$ of $A$ (for intuition, we can think of $A_1$ as the set $A$ of vertices in the initial subgraph system $\sset$, and of $A_2$ as the set $A'\setminus A$ of all vertices that were added to $A$ when constructing $\sset'$). Additionally, the subgraph system must contain a subset $N_1\subseteq A_2\cup B$ of vertices, such that, for every vertex $v\in A_1$, for every edge $(v,v')\in M$ that is incident to $v$, it must be the case that $v'\in A_1\cup N_1$. Intuitively, $N_1$ corresponds to the set $B$ of vertices in the initial subgraph system $\sset$. We then denote by $R_1=V(G)\setminus (A_1\cup N_1)$, and require that, for every vertex $v\in A_1$, the list $L(v)$ contains all vertices in $N_G(v)\cap R_1$. Intuitively, $R_1$ corresponds to the initial set $U$ of vertices in the subgraph system $\sset$, and this definition is designed so that the lists $L(v)$ for vertices $v\in A_1$ do not need to undergo major updates when constructing the $z_2$-subgraph system $\sset'$ from the $z_1$-subgraph system $\sset$. For every vertex $v\in A_2$, we require, as before, that, for every edge $(v,v')\in M'$ incident to $v$, $v'\in S$ holds, and that $L(v)$ contains all vertices of $N_G(v)\cap U$. For consistency of notation, we denote $N_2=B$ and $R_2=U$. This definition ensures that the algorithm we outlined above can indeed be used in order to construct a $z_2$-subgraph system from a $z_1$-subgraph system in time $O(n^{1+o(1)}\cdot z_1)$, though the resulting $z_2$-subgraph system will be a $2$-level system. We then adapt our algorithm for maintaining a maximal matching from a subgraph system, so that it works with a $2$-level subgraph system.

The latter algorithm remains largely unchanged, except that Invariant \ref{inv: vertices of S}  now implies that, in addition to Invariant \ref{inv: vertices of A}, the following invariant also holds:

\begin{properties}[2]{I}
	\item	\label{inv: vertices of A1} Throughout the phase, the number of vertices of $A_1$ that may be matched by $M^*$ to vertices of $R_1$ is bounded by $2\tau=\frac{64r}{z}$.
\end{properties}

The remainder of the algorithm remains unchanged, except in how it handles rematching the vertices of $A$. The idea here is that it is easy to rematch vertices of $B\cup U$ using the same algorithm as before. Whenever we need to rematch a vertex $a\in A_2$, we also employ the same algorithm as before, that may delete an edge that unmatches a vertex of $B\cup U$, which is then rematched as before. However, rematching vertices of $A_1$ is done slightly differently. In order to rematch a vertex $a\in A_1$, we scan the first $2\tau+1$ vertices of the list $L(a)=N_G(a)\cap R_1$, until we encounter the first vertex $v$ that is not matched to a vertex of $A_1$. From Invariant \ref{inv: vertices of A1}, some vertex $v$ from among the first $2\tau+1$ vertices of $L(a)$ must have this property. We then insert the edge $(a,v)$ into $M^*$, and, if another edge $(v,v')$ incident to $v$ currently lies in $M^*$, we delete it from $M^*$ (and from $M_1$ if it lies there). Notice that, from the choice of $v$, it may not lie in $A_1$, so we can rematch it using one of the procedures outlined above.

In order to optimize the final running time of the algorithm, we use a natural extension of the above $2$-level construction to an arbitrary number $k$ of levels. Our final algorithm employs $k=O(\log n)$ levels, and parameters $z_1>z_2>\cdots>z_k$, where $z_1=n$, $z_k$ is close to $\sqrt{n}$, and for all $1< i\leq k$, $z_i=z_{i-1}/2$. As before, we partition the update sequence into a number of level-$1$  phases of an appropriate length, and, at the beginning of every level-$1$ phase, we construct a basic $z_1$-subgraph system using the $\tilde O(m+n)$-time algorithm outlined above. Generally, for all $1<i\leq k$, every level-$(i-1)$ phase is partitioned into two level-$i$ phases of roughly identical length, and, at the beginning of every level-$i$ phase, we construct an $i$-level $z_i$-subgraph system using the $(i-1)$-level $z_{i-1}$ subgraph system constructed at the beginning of the current level-$(i-1)$ phase. This way we obtain, at the beginning of every level-$k$ phase, a $k$-level $z_k$-subgraph system. Our recursive construction ensures that $z_k$ is close to $\sqrt{n}$, and that every $k$-level phase spans roughly $n$ adversarial updates. We then employ the algorithm we outlined above in order to maintain a maximal matching in $G$ over the course of a level-$k$ phase, using the $k$-level $z_k$-subgraph system for the graph $G$ at the beginning of that phase.

\section{Preliminaries}
\label{sec: prelims}

All logarithms in this paper are to the base of $2$. Given an integer $k\geq 1$, we denote by $[k]=\set{1,\ldots,k}$. 

By default, all graphs considered in this paper are undirected, unless stated otherwise. Given a graph $G=(V,E)$, for every vertex $v\in V$, we denote by $\delta_G(v)$ the set of all edges of $E$ that are incident to $v$, by $N_G(v)=\set{u\in V\mid (u,v)\in E}$ the set of neighbors of $v$ in $G$, and by $\deg_G(v)=|\delta_G(v)|=|N_G(v)|$ the degree of $v$ in $G$.

Given a graph $G=(V,E)$, a \emph{matching} in $G$ is a subset $M\subseteq E$ of edges, such that, for every vertex $v\in V$, at most one edge in $M$ is incident to $v$. We say that a vertex $v$ is \emph{matched} by $M$, if $M$ contains an edge incident to $v$, and we say that it is \emph{not matched}, or \emph{unmatched} by $M$ otherwise. A matching $M$ in $G$ is \emph{maximal} if, for every edge $e\in E\setminus M$, the edge set $M\cup \set{e}$ is not a valid matching; equivalently, for every pair $u,v\in V$ of vertices that are not matched by $M$, $(u,v)\not\in E$ must hold.

For a \emph{directed} graph $G=(V,E)$ and a vertex $v\in V$, the vertices in set $N^{-}_{G}(v)=\{u\in V\mid (u,v)\in E\}$ are called the \emph{in-neighbors of $v$}, and the vertices in  set  $N^{+}_{G}(v)=\{u\in V\mid (v,u)\in E\}$ are called the \emph{out-neighbors of $v$}.

Throughout the paper, all lists are implemented as binary search trees to support $O(\log n)$-time insertion, deletion and lookup operations. All dynamic graphs are stored in the adjacency-list representation, with each vertex maintaining its neighbors using a binary search tree.

\paragraph{Fully dynamic graphs.} A fully dynamic graph is a graph $G$ that undergoes an online sequence of adversarial updates, where every update either deletes a single edge from $G$ or inserts an edge into $G$. We usually denote by $G\attime[0]$ the \emph{initial graph $G$}, before any updates, and, for an integer $\tau>0$, we denote by $G\attime$ the graph $G$ at time $\tau$, that is, after the first $\tau$ adversarial updates are applied to it.

\subsection{Settled Vertices and Maximal Matching}
Our algorithm uses the notion of settled vertices that we define next.

\begin{definition}[Settled vertices]\label{def: settled}
	Let $G$ be a graph and let $M$ be a matching in $G$. We say that a vertex $v\in V(G)$ is \emph{settled} with respect to $M$, if either (i) $M$ contains an edge incident to $v$; or (ii) for every vertex $y\in N_G(v)$, $M$ contains an edge incident to $y$.
\end{definition}

It is immediate to verify that a matching $M$ in a graph $G$ is maximal if and only if every vertex of $G$ is settled with respect to $M$. We will also use the following simple observation.

\begin{observation}\label{obs: settled if neighbors settled}
	Let $G$ be a graph, let $M$ be a matching in $G$, and let $v\in V(G)$ be a vertex of $G$. Assume further that every vertex in $N_G(v)$ is settled with respect to $M$. Then $v$ is settled with respect to $M$.
\end{observation}
\begin{proof}
	Clearly, if every vertex in $N_G(v)$ is matched by $M$, then $v$ is settled with respect to $M$. Assume now that this is not the case, so some vertex $y\in N_G(v)$ is not matched by $M$. Since $y$ is settled with respect to $M$ but it is not matched by $M$, every neighbor of $y$ must be matched by $M$, so in particular $v$ must be matched by $M$.
\end{proof}

Lastly, consider the scenario where we are given a graph $G$ and a matching $M$ that is maximal in $G$. Assume now that the matching $M$ undergoes a sequence of edge deletions and insertions, while the graph $G$ remains fixed, and denote the resulting matching by $M'$. Let $E_D$ be the set of edges that were deleted from $M$, and let $V_D$ be the set of all vertices that serve as endpoints of $E_D$. In the following simple observation we prove that, if every vertex of $V_D$ is settled with respect to $M'$, then $M'$ must be a maximal matching.

\begin{observation}\label{obs: maximal matching transformation}
	Let $G$ be a graph and let $M,M'$ be a pair of matchings in $G$, where $M$ is a maximal matching. Denote by $E_D=M\setminus M'$, and by $V_D$ the set of all vertices that serve as endpoints of the edges of $E_D$. Assume further that every vertex in $V_D$ is settled with respect to $M'$. Then $M'$ is a maximal matching in $G$.
\end{observation}
\begin{proof}
	Assume otherwise. Then there is an edge $(u,v)\in E(G)$, such that both $u$ and $v$ are unmatched by $M'$. Since matching $M$ is maximal in $G$, at least one of the vertices $u,v$ is matched by $M$; assume w.l.o.g. that it is $u$. Since $u$ is not matched by $M'$, the unique edge of $M$ incident to $u$ lies in $E_D$, and $u\in V_D$ must hold. Since vertex $u$ is settled with respect to $M'$, but it is not matched by $M'$, every vertex in $N_G(u)$ must be matched by $M'$. But then $v$ must be matched by $M'$, a contradiction.
\end{proof}

\subsection{Edge Coloring} 
Given a graph $G$, an \emph{edge coloring of $G$ with $k$ colors} is an assignment of a color $c(e)\in \set{1,\ldots,k}$ to every edge $e\in E(G)$, such that, for every vertex $v\in V(G)$, all its incident edges are assigned distinct colors. Equivalently, for all $1\leq i\leq k$, the set $\set{e\in E(G)\mid c(e)=i}$ of edges must define a matching. 

Vizing's theorem \cite{Viz64} states that any simple $n$-vertex and $m$-edge graph with maximum degree $\Delta$ can be edge-colored using at most $\Delta+1$ colors, and the algorithmic proof yields a deterministic $O(mn)$-time algorithm for finding such a coloring. Subsequent work has focused on improving this running time, culminating in a randomized $O(m\log \Delta)$-time algorithm \cite{edge-coloring} and a very recent deterministic $O(m^{1+o(1)})$-time algorithm \cite{ABB+26}. We use the latter algorithm, summarized in the following theorem.

\begin{theorem}[Theorem 1.1 of \cite{ABB+26}] \label{thm: edge coloring}
There is a deterministic algorithm, that, given a simple $n$-vertex and $m$-edge graph $G$ with maximum vertex degree $\Delta$, computes an edge coloring of $G$ with $(\Delta+1)$ colors. The running time of the algorithm is $O\left(m\cdot 2^{O(\sqrt{\log \Delta})}\cdot \log n\right ) \leq O\left(m\cdot n^{o(1)}\right )$.
\end{theorem}

We remark that the amortized update time of our algorithm for dynamic maximal matching can be improved from $O(n^{1/2+o(1)})$ to $\tilde O(n^{1/2})$ (against an adaptive adversary) by using the randomized $(\Delta+1)$-edge coloring algorithm of \cite{edge-coloring} instead.

\section{An Algorithm for Decremental Graphs with $\tilde O\left(n^{2/3}\right)$ Amortized Update Time}
\label{sec: simple alg}

In this section we provide a relatively simple deterministic algorithm for the dynamic Maximal Matching problem in a decremental $n$-vertex graph that undergoes an online sequence of at least $n^{4/3}$ edge deletions; the amortized update time of the algorithm is $\tilde O\left(n^{2/3}\right )$ --- somewhat higher than that required by \Cref{thm: main}. We  provide this algorithm in order to introduce some of our ideas and to provide the intuition and the motivation for the more complex data structures that are used in the proof of \Cref{thm: main}, that is provided in \Cref{sec: main alg}. The setting of decremental graphs that undergo a sufficiently long sequence of edge deletions captures most of the complexity of the problem, and extending this algorithm to the general fully dynamic setting is not difficult. This simpler setting allows us to introduce some of our central notions and ideas and to provide intuition in a clean setting.

 The main combinatorial object that we use  is a \emph{subgraph system}, that we define below. This object can be thought of as being somewhat similar to the EDCS data structure that was used by \cite{BBKS25}. Like in \cite{BBKS25}, our algorithm is partitioned into phases, each of which spans a limited number of updates by the adversary. Unlike the algorithm of \cite{BBKS25}, that maintains the EDCS data structure of the dynamic graph $G$ over the entire algorithm, we simply compute the subgraph system from scratch at the beginning of every phase. We then exploit the subgraph system within each phase in order to maintain a maximal matching in $G$ over the course of the phase. While our subgraph system can be thought of as being similar to EDCS, it is designed specifically so as to make the task of maintaining a maximal matching in $G$  easier. We now formally define the subgraph system. We note that in \Cref{sec: main alg} we define a more general object, called a \emph{multi-level subgraph system}, so to avoid confusion, the subgraph system used in this section is called \emph{basic}.

\begin{definition}[Basic Subgraph System]\label{def: basic subgraph structure}
	Let $G$ be an $n$-vertex graph, and let $1\leq z\leq n$ be an integral parameter. A \emph{basic $z$-subgraph system for $G$} consists of the following ingredients:
	\begin{itemize}
		\item A subset $M\subseteq E(G)$ of edges. For every vertex $v\in V(G)$, we denote by $m(v)$ the number of edges of $M$ that are incident to $v$, and we require that the following property holds:
		
		\begin{properties}{K}
			\item \label{prop: small degrees in M basic} For every vertex $v\in V(G)$, $m(v)\leq z$;
		\end{properties}
		
		\item A partition $(S,U)$ of $V(G)$, for which the following properties hold:
		\begin{properties}[1]{K}
			\item \label{prop: no edge of M w both endpoints in U basic} No edge of $M$ has both endpoints in $U$;
			
			\item \label{prop: edges inc to S basic} For every vertex $v\in S$, $m(v)=z$; and
			\item \label{prop: degrees in U basic} For every vertex $u\in U$, $|N_G(u)\cap U|\leq z$.
		\end{properties}
		\item A partition $(A,B)$ of $S$ that satisfies the following properties:
		\begin{properties}[4]{K}
			\item \label{prop: U neighbors in B basic} For every vertex $u\in U$, $|N_G(u)\cap B|\leq 2z$; and
			\item \label{prop: neighbors of A in M in S basic} For every vertex $v\in A$, for every edge $(u,v)\in M$ incident to $v$, $u\in S$ must hold.
		\end{properties}
		\item For every vertex $u\in U$, a list $\Lambda(u)=N_G(u)\cap (B\cup U)$, whose length must be $O(z)$; and
		
		\item For every vertex $v\in A$, the list $L(v)=N_G(v)\cap U$.
\end{itemize}
	
	We denote the basic $z$-subgraph system by $\sset = \left(S,B,U,M,\Lambda,L\right )$, where $\Lambda=\set{\Lambda(u)}_{u\in U}$ and $L=\set{L(v)}_{v\in A}$.
\end{definition}

As noted already, the basic $z$-subgraph data structure is somewhat similar to the \emph{Edge-Degree Constrained Subgraph} (EDCS) data structure that was used in the algorithm of  \cite{BBKS25} for fully dynamic maximal matching, but it is designed specifically for the task of maintaining a maximal matching in $G$. While the EDCS data structure satisfies some analogues of Properties  \ref{prop: small degrees in M basic}--\ref{prop: degrees in U basic}, it may not satisfy Properties \ref{prop: U neighbors in B basic} and \ref{prop: neighbors of A in M in S basic}, which are crucially exploited in our algorithm. 
We provide a more detailed comparison of the two data structures in \Cref{subsec: techniques}.

Next, we provide two main algorithmic tools that are key ingredients of our algorithm. 
The first tool, summarized in the following theorem, is an efficient algorithm for constructing a basic $z$-subgraph system for a graph $G$.

\begin{theorem}\label{thm: constructing simple subgraph system}
	There is a deterministic algorithm, that, given an $n$-vertex and $m$-edge graph $G$, and an integral parameter $1\leq z\leq n$, constructs a basic $z$-subgraph system $\sset = \left(S,B,U,M,\Lambda,L\right )$ for $G$, in time $\tilde O(n+m)$.
\end{theorem}

The  proof of Theorem \ref{thm: constructing simple subgraph system} appears in \Cref{subsec: constructing basic}; a high-level overview of the proof can be found in \Cref{subsec: techniques}.

Our second main tool is an algorithm that, given a decremental graph $G$ and a basic $z$-subgraph system for the initial graph $G\attime[0]$, maintains a maximal matching in $G$.  The algorithm is summarized in the following theorem.

\begin{theorem}\label{thm: matching from basic subgraph system}
	There is a deterministic algorithm, whose input consists of an initial $n$-vertex graph $G$,  together with integral parameters  $1\leq z\leq n$ and $r>0$, and a basic $z$-subgraph system $\sset$ for $G$. The algorithm maintains a maximal matching in $G$, as it undergoes an online sequence of at most $r$ edge deletions, in total time $\tilde O\left(n^{1+o(1)}\cdot z+rz+\frac{(n+r)^2}{z}\right)$.
\end{theorem}

The proof of \Cref{thm: matching from basic subgraph system} is deferred to \Cref{sec: basic-subgraph-system-to-matching}; 
a high-level overview of the proof can be found in \Cref{subsec: techniques}.

\paragraph{Completing the algorithm.}
We are now ready to complete our algorithm for maximum matching with $\tilde O\left (n^{2/3}\right )$ amortized update time. We assume that we are given a decremental $n$-vertex graph $G$, that undergoes an online sequence of at least $n^{4/3}$ edge deletions. We let $r=\floor{n^{4/3}}$ and $z=\floor{n^{2/3}}$, and we partition timeline into \emph{phases}, each of which spans exactly $r$ edge deletions by the adversary, except for the last phase, that may contain fewer updates. 
Consider now some phase $\Phi$, and let $\tau$ be the time when the phase started. At the beginning of the phase $\Phi$, we use the algorithm from \Cref{thm: constructing simple subgraph system} in order to compute a basic $z$-subgraph system $\sset$ of $G\attime$. 
We then apply the algorithm from \Cref{thm: matching from basic subgraph system} to graph $G\attime$ and the $z$-subgraph system $\sset$ in order to maintain a maximal matching in $G$ over the course of the entire phase $\Phi$. Recall that the running time of the algorithm from \Cref{thm: constructing simple subgraph system}  is bounded by $\tilde O(n^2)$, while the running time of the algorithm from \Cref{thm: matching from basic subgraph system}  is bounded by:

\[ \tilde O\left(n^{1+o(1)}\cdot z+rz+\frac{(n+r)^2}{z}\right)\leq  \tilde O\left(n^{2}\right),  \]

since $r=\floor{n^{4/3}}$ and $z=\floor{n^{2/3}}$. Therefore, the total running time of our algorithm for a single phase is $\tilde O\left(n^{2}\right)$. Since every phase except for possibly the last one spans $\Omega(n^{4/3})$ updates to $G$, and since we assumed that the total length of the update sequence is  $\Omega(n^{4/3})$, we get that the amortized  update time of our algorithm is $\tilde O\left (n^{2/3}\right )$.

\subsection{Proof of \Cref{thm: constructing simple subgraph system}}
\label{subsec: constructing basic}

In this proof, we say that a $z$-subgraph system $\sset = \left(S,B,U,M,\Lambda,L\right )$ for $G$ is \emph{weak}, if it has all required properties, except for, possibly, Property \ref{prop: U neighbors in B basic}. A proper (non-weak) $z$-subgraph system is called \emph{strong} in this proof, to avoid confusion.

	Our algorithm consists of two steps. In the first step, we construct a {\bf weak} basic $z$-subgraph system  $\sset = \left(S,B,U,M,\Lambda,L\right )$. Then in the second step we then ``fix'' it so it becomes  a {\bf strong} subgraph system. We now describe each of the two steps in turn.
	
\subsection*{Step 1: a Weak Subgraph System}
	We start by constructing an initial maximal set $M'\subseteq E(G)$ of edges with the following property: every vertex of $G$ must be incident to at most $z$ edges of $M'$.
	In order to do so efficiently, we start by setting $M'=\emptyset$, and then gradually add edges to $M'$. We maintain, for every vertex $v\in V(G)$, a counter $m'(v)$, that counts the number of edges of $M'$ incident to $v$. All such counters are initialized to $0$ initially. We then consider every edge $e\in E(G)$ one by one. When $e$ is considered, if both $m'(v),m'(u)<z$, then we insert $e$ into $M'$ an increase each of the counters $m'(v),m'(u)$ by $1$. Clearly, this algorithm can be implemented in time $O(m+n)$.

	Next, we set $S=\set{v\in V(G)\mid m'(v)=z}$ and $U=V(G)\setminus S$, and we let $M$ be the set of edges obtained from $M'$, by deleting from it all edges with both endpoints in $U$. For every vertex $v\in V(G)$, we denote by $m(v)$ the number of edges of $M$ incident to $v$. Notice that, if $v\in S$, then $m(v)=m'(v)=z$, and so Property \ref{prop: edges inc to S basic} now holds. It is also immediate to verify that Properties \ref{prop: small degrees in M basic} and \ref{prop: no edge of M w both endpoints in U basic} hold as well. Consider now a vertex $u\in U$. We claim that  $|N_G(u)\cap U|\leq z$ must hold. Indeed, otherwise, since $m'(u)<z$ holds by the definition of the set $U$, there must be an edge $(u,u')\in E(G)\setminus M'$ with $u'\in U$; so in particular, $m'(u')<z$, contradicting the maximality of $M'$. Therefore, Property \ref{prop: degrees in U basic} holds as well.

We define the partition $(A,B)$ of $S$ as follows: set $A$ contains all vertices  $v\in S$, such that, for every edge $(u,v)\in M$ incident to $v$, $u\in S$ holds; we then set $B=S\setminus A$.  This ensures that Property	\ref{prop: neighbors of A in M in S basic} holds as well.

Lastly, we compute, for every vertex $u\in U$, a list $\Lambda(u)=N_G(u)\cap (B\cup U)$, and, for every vertex $v\in A$, the list $L(v)=N_G(v)\cap U$.
This completes the construction of the initial subgraph system. From our discussion, it has all required properties, except, possibly, for Property \ref{prop: U neighbors in B basic}.
It is immediate to verify that the running time of the algorithm so  far is $O(m+n)$.

\subsection*{Step 2: a Strong Subgraph System}	

So far we have computed a weak basic $z$-subgraph system $\sset = \left(S,B,U,M,\Lambda,L\right )$ for $G$; in other words, $\sset$ has all required properties, except for, possibly, Property \ref{prop: U neighbors in B basic}. In this step we ``fix'' the subgraph system, to ensure that it has all required properties.
Throughout the algorithm, we maintain, for every vertex $v\in V(G)$, the counter $m(v)$, that counts the number of edges of $M$ incident to $v$. All these counters can be initialized in time $O(m)$, and it is easy to maintain without increasing the asymptotic running time of the algorithm, so we do not further describe the process for maintaining them. 
For every vertex $v\in B$, we also maintain a counter $n(v)$, that counts the number of edges $(u,v)\in M$ with $u\in U$ and the set $Z(v)$ of all vertices $u\in U$ with $(u,v)\in M$. We initialize all these counters $n(v)$ and vertex sets $Z(v)$ at the beginning of Step 2 in time $O(|M|)\leq O(m)$, and then maintain them explicitly.

Our algorithm will gradually modify the subgraph system $\sset$, using the following actions, that we refer as \emph{allowed modifications}:

\begin{itemize}
	\item {\bf Edge Swap:} An edge $e=(v,u)\in M$ with $v\in B$ and $u\in U$ is deleted from $M$, and, instead, another edge $e'=(v,u')$ with $u'\in U$ is inserted into $M$. The edge swap may only be performed if $m(u')<z$; note that it does not affect $m(v)$.
		
	\item {\bf Promotion of vertices from $U$ to $S$:} If a vertex $u\in U$ has  exactly $z$ edges of $M$ incident to it, then we may move it from $U$ to $S$, where it joins the set $B$. If, additionally, for every edge $(u,v)\in M$ incident to $u$, the other endpoint $v$ lies in $S$, then we move $u$ to $A$ instead.
	
	\item {\bf Promotion of vertices from $B$ to $A$:} If, due to any of the modifications, for some vertex $v\in B$, $n(v)=0$ (or, equivalently, every vertex in $\set{v'\mid (v,v')\in M}$ lies in $S$), then we move $v$ to $A$.
\end{itemize}

We use the following immediate observation.

\begin{observation}\label{obs: properties preserved}
	Let $\sset = \left(S,B,U,M,\Lambda,L\right )$ be a weak $z$-subgraph system for a graph $G$, and let $\sset'$ be obtained from $\sset$ via one of the allowed modifications. Then $\sset'$ is a valid weak $z$-subgraph system; in other words, it has all properties of the $z$-subgraph system, except, possibly, for Property \ref{prop: U neighbors in B basic}. 
\end{observation}
\begin{proof}
	Observe first that the edge set $M$ may only be modified by the Edge Swap. The edge swap may only increase $m(u')$ value for a single vertex $u'\in U$ for which $m(u')<z$ held before this procedure, and it decreases $m(u)$ value for a single vertex $u\in U$. Additionally, we may only move vertices from $U$ to $S$ but not in the other direction. This ensures that Properties 
\ref{prop: small degrees in M basic}--\ref{prop: degrees in U basic}  hold in $\sset'$. Observe that the edge swap may not affect edges of $M$ incident to vertices of $A$, and that we may never remove vertices from $S$. Moreover, a vertex may only be added to the set $A$ when Property \ref{prop: neighbors of A in M in S basic} holds for it. Since vertices can be moved from $U$ to $S$ but not in the other direction, this ensures that Property \ref{prop: neighbors of A in M in S basic} continues to hold as well.
\end{proof}

Our algorithm does not maintain the vertex sets $\set{\Lambda(u)}_{u\in U}\cup \set{L(v)}_{v\in A}$ explicitly; instead it simply recomputes them at the end.
Lastly, our algorithm ensures that the following invariant always holds:

\begin{properties}{I}
	\item\label{inv: in B} If $v\in B$ then $n(v)>0$.
\end{properties}

We are now ready to describe our algorithm for Step 2. Our algorithm considers every vertex $u\in U$ one by one. When a vertex $u$ is considered, we inspect the set $N_G(u)$ to check whether $|N_G(u)\cap B|\leq z$ holds. If so, then no further processing of $U$ is necessary. Assume now that $|N_G(u)\cap B|>z$, and denote by $r=z-m(u)$. In this case, we call Procedure $\procprocessU(u)$, that is described in Figure \ref{fig:procprocesss2}. Note that there must be a collection 
$\set{v_1,\ldots,v_r}$ of $r$ vertices of $N_G(u)\cap B$ such that, for all $1\leq i\leq r$, $(v_i,u)\not\in M$. The procedure selects any such collection of vertices by inspecting $N_G(u)$. For each such vertex $v_i$, it then selects an arbitrary vertex $a_i\in Z(v_i)$ (so $a_i\in U$ and $(v_i,a_i)\in M$ must hold). It then performs an edge swap between the edge $(u,v_i)$, that is inserted into $M$, and $(v_i,a_i)$, that is deleted from $M$. After this step, we are  guaranteed that $m(v)=z$, and $v$ is promoted from $U$ to $S$; all data structures are then updated accordingly.

\program{Procedure \procprocessU}{fig:procprocesss2}
{
	{\bf Input:} a vertex $u\in U$ with $|N_G(u)\cap B|>z$.
	
	\begin{itemize}
		\item Denote: $r=z-m(u)$.
		
		\item Select an arbitrary set $\set{v_1,\ldots,v_r}$ of $r$ vertices of $N_G(u)\cap B$ such that, for all $1\leq i\leq r$, $(v_i,u)\not\in M$, by inspecting the vertex set $N_G(u)$;
		
			\item For all $1\leq i\le r$, select an arbitrary vertex $a_i\in Z(v_i)$ (recall that $a_i\in U$ and $(v_i,a_i)\in M$ must hold).
			\item For all $1\leq i\leq r$:
			\begin{itemize}
				\item insert the edge $(v_i,u)$ into $M$;
				\item delete the edge $(v_i,a_i)$ from $M$;
				\item delete $a_i$ from $Z(v_i)$;
			\end{itemize}
			\item Move $u$ from $U$ to $S$ (note that $m(u)=z$ now holds);
		\item If, for every edge $(u,v)\in M$ incident to $u$, $v\in S$ holds: add $u$ to $A$.
		
		\item Otherwise:
		add $u$ to $B$, and
		initialize $Z(u)=\emptyset$ and $n(u)=0$.

		\item For every vertex $v\in B$ with $u\in Z(v)$: (note that each such vertex $v$ lies in $N_G(u)$): delete $u$ from $Z(v)$ and decrease $n(v)$ by $1$; if it reaches $0$, move $v$ from $B$ to $A$.
		
			\item For all $1\leq i\leq r$, decrease $n(v_i)$ by $1$, and, if it reaches $0$, move $v_i$ from $B$ to $A$.
		\end{itemize}
}

It is easy to verify that Procedure $\procprocessU(u)$ can be implemented to run in time $\tilde O(\deg_G(u))$.
Note that, if we applied Procedure $\procprocessU$ to vertex $u$, then it is moved from $U$ to $S$, and it remains in $S$ until the end of the algorithm. Otherwise, at the time when $u$ was processed,  $|N_{G}(u)\cap B|\leq z$ held.
From Property \ref{prop: degrees in U basic}, $|N_G(u)\cap U|\leq z$ holds as well. Note that the only vertices that may join the vertex set $B$ over the course of Step 2 of the algorithm are  vertices that lie in $U$. Therefore, Property \ref{prop: U neighbors in B basic} must hold for every vertex $u\in U$ at the end of the algorithm.

Our last step is to compute, for every vertex $u\in U$, the list $\Lambda(u)=N_G(u)\cap (B\cup U)$, whose length now must be $O(z)$. We also compute, for every vertex $v\in A$, the list $L(v)=N_G(v)\cap U$. It is immediate to verify that the running time of Step 2 is $\tilde O(m+n)$, and so the running time of the entire algorithm is $\tilde O(m+n)$.

\subsection{From Basic Subgraph System to Fully Dynamic Maximal Matching: Proof of \Cref{thm: matching from basic subgraph system}}
\label{sec: basic-subgraph-system-to-matching}

In this section we prove \Cref{thm: matching from basic subgraph system}. We assume that we are given as input an initial $n$-vertex graph $G$, together with
 an integral parameter $1\leq z\leq n$, and a basic $z$-subgraph system  $\sset = \left(S,B,U,M,\Lambda,L\right )$ for $G$. Our goal is to maintain a maximal matching in $G$, as it undergoes an online sequence of $r$ edge deletions.

Throughout the algorithm, the vertex sets $S,B,U$ and $A$ remain unchanged. The edge set $M$ is decremental: whenever the adversary deletes an edge $e\in M$ from $G$, this edge is deleted from $M$ as well, but no new edges may be inserted into $M$.
The lists in $\Lambda$ and $L$ may be updated by our algorithm, as described below. Throughout the algorithm, we denote by $\rho=n+r$.

\subsubsection{Data Structures and Invariants}

Recall that, as part of the subgraph system $\sset$, we are given, for every vertex $u\in U$, a vertex list $\Lambda(u)=N_G(u)\cap (B\cup U)$, whose length is $O(z)$. Over the course of the algorithm, we will update the list $\Lambda(u)$ to ensure that $\Lambda(u)=N_{G}(u)\cap (B\cup U)$ always holds. Specifically, whenever the adversary deletes an edge $(u,v)$ from $G$ with $v\in \Lambda(u)$,  we delete $v$ from $\Lambda(u)$.
Similarly, as part of the subgraph system, we are given, for every vertex $v\in A$, the list $L(v)=N_G(v)\cap U$. 
Over the course of the algorithm, we will update the list $L(v)$ to ensure that $L(v)=N_{G}(v)\cap U$ always holds, as follows: whenever the adversary deletes an edge $(u,v)$ from $G$ with $u\in L(v)$, we delete $u$ from $L(v)$.

Let $\hat G$ be the subgraph of $G$ induced by the edges of $M$. Note that the maximum vertex degree in $\hat G$ is bounded by $z$. We start by computing an edge-coloring of $\hat G$ with $z+1$ colors using the algorithm from \Cref{thm: edge coloring}, in time $O\left (|E(\hat G)|\cdot n^{o(1)}\right )\leq O\left (n^{1+o(1)}\cdot z\right )$.
This coloring naturally defines a partition $(M_1,\ldots,M_{z+1})$ of $M$, where for all $1\leq i\leq z+1$, $M_i$ is a matching, consisting of all edges of $\hat G$ that were assigned the color $i$. For all $1\leq i\leq z+1$, we compute the number $m_i$ of vertices of $S$ that are not matched in $M_i$.  Over the course of the algorithm, when an edge of $M_i$ is deleted by the adversary, we delete it from $M_i$ and update $m_i$ accordingly.

Note that, from Property \ref{prop: edges inc to S basic} of the basic subgraph system,  every vertex $v\in S$ is incident to exactly $z$ edges of $M$. Therefore, there are at most one index $1\leq i\leq z+1$, for which $v$ is not matched in $M_i$.  Altogether, $\sum_{i=1}^{z+1} m_i \leq |S|\leq n$ must hold, and there must exist an index $1\leq i^*\leq z+1$ with $m_{i^*}\le \frac{n}{z+1}\leq \frac{2n}{z}$. We reindex the matchings so that $i^*=1$ holds.
We partition the algorithm into \emph{phases},  where each phase spans exactly $\ceil{ \frac{r}{z}}$ edge deletions by the adversary. Since the total number of edge deletions is bounded by $r$, we get that the number of phases is bounded by $2z$. Throughout the algorithm, we will maintain a maximal matching $M^*$ with $M_1\subseteq M^*$,  
while matchings $M_2,\dots,M_{z+1}$ are used in order to ``repair" $M_1$ when necessary. The algorithm may modify the matching $M_1$, but it always ensures that  $M_1\subseteq M$ holds.  Recall that we denoted $\rho=n+r$.
We ensure that, at the beginning of every phase, the following invariant holds:

\begin{properties}{R}
	\item \label{prop: matching M1 for basic} After the initialization of every phase, all but at most $\frac{32\rho}{z}$ vertices of $S$ are matched in $M_1$.
\end{properties}

Our algorithm also maintains a {\bf directed} graph $H$, that allows us to rematch the vertices of $U\cup B$ efficiently, and that we define next. 

\paragraph{Graph $H$.} Graph $H$ is a directed graph that is used in order to quickly rematch vertices of $B\cup U$. The vertex set of $H$ is $V(H)=B\cup U$. For every vertex $u\in U$ that is currently unmatched by $M^*$, there is a directed edge in $H$ from $u$ to every vertex in $N_{G}(u)\cap (B\cup U)$ (recall that all vertices in $N_{G}(u)\cap (B\cup U)$ lie in $\Lambda(u)$, and that $|\Lambda(u)|\leq O(z)$). Equivalently, for every vertex $v\in B\cup U$, the set of the in-neighbors of $v$ in $H$ is exactly the set of all unmatched vertices $u\in U$ with $(u,v)\in E(G)$.

\paragraph{Phase execution: a high-level overview.}
Over the course of a phase, if an edge $e\in M_1$ is deleted by the adversary, it is deleted from $M_1$ as well; recall that at most $\frac{r}{z}$ such edge deletions may occur over the course of a phase. Additionally, our algorithm may choose to delete up to $\frac{r}{z}$ edges from $M_1$ over the course of the phase. Each edge deletion from $M_1$ may cause at most two vertices of $S$ to become unmatched in $M^*$. By Property \ref{prop: matching M1 for basic}, immediately after the initialization of a phase, at most $\frac{32\rho}{z}$ vertices of $S$ are unmatched by $M_1$. Therefore, we will ensure that the following invariant must hold throughout the phase:

\begin{properties}[1]{R}
	\item \label{prop during subphase: matching M1 for basic} At all times, the number of vertices of $S$ that are not matched by $M_1$ is bounded by $\frac{64\rho}{z}$.
\end{properties}

Recall that, from Property 
\ref{prop: neighbors of A in M in S basic} of the basic subgraph system, for every vertex $v\in A$, for every edge $(u,v)\in M$ incident to $v$, $u\in S$ must hold. Since $M_1\subseteq M$, we conclude that $M_1$ may not match a vertex of $A$ to a vertex of $U$. Therefore, the only way that a vertex $v\in A$ is matched to a vertex of $U$ in $M^*$ is if $v$ is not matched by $M_1$. Combining this with Invariant \ref{prop during subphase: matching M1 for basic}, we get that the following invariant also holds	 throughout the phase:

\begin{properties}[2]{R}
	\item \label{prop during subphase: few vertices of S'' matched to U for basic} At all times, at most $\frac{64\rho}{z}$ vertices of $A$ may be matched to vertices of $U$ in $M^*$.
\end{properties}

Our algorithm will also explicitly maintain the set $\hat S\subseteq S$ of all vertices $v\in S$ that are currently not matched by $M^*$. We say that a vertex $v\in V(G)$ \emph{switches its status in $M^*$} whenever it switches from being matched to unmatched in $M^*$, or the other way around.
Throughout the algorithm, we say that a vertex $x\in V(G)$ is \emph{settled}, if it is settled with respect to $M^*$ (see Definition \ref{def: settled}).

\subsubsection{Update Procedures}

Next, we describe the main procedures that our algorithm uses in order to maintain all data structures correctly. 
We start with Procedure
\procupdate, that is called whenever a vertex of $G$ switches its status in $M^*$.
The purpose of the procedure is to ensure that all data structures are consistent with this change of status.

\subsection*{Procedure \procupdate}

The input to Procedure \procupdate is a vertex $v\in V(G)$. The procedure checks if $v$ is currently matched in $M^*$, and updates the vertex set $\hat S$ and the graph $H$ accordingly. We provide the description of \procupdate in \Cref{fig:procupdatesimple}.

\program{Procedure \procupdate}{fig:procupdatesimple}
{
	{\bf Input:} a vertex $v\in V(G)$.
	
	\begin{itemize}
		\item If $v$ is matched in $M^*$:
            \begin{itemize}
                \item if $v\in \hat S$, delete $v$ from $\hat S$.
                \item if $v\in U$, delete all edges of $H$ leaving $v$ from $H$.
            \end{itemize}
        \item Otherwise ($v$ is unmatched in $M^*$):
            \begin{itemize}
                \item if $v\in S$, insert $v$ into $\hat S$.
                \item if $v\in U$, for every vertex $v'\in \Lambda(v)$, insert the edge $(v,v')$ into $H$.
            \end{itemize}
	\end{itemize}
}

Procedure $\procupdate(v)$ is called whenever a vertex $v\in V(G)$ changes status in $M^*$, and it ensures that the following properties hold:

\begin{properties}{Q}
    \item \label{prop: set S basic} At all times, $\hat S$ is exactly the set of all vertices of $S$ that are not matched by $M^*$.
	\item \label{prop: graph H basic} At all times, for every vertex $v\in B\cup U$, the set of all in-neighbors of $v$ in $H$ is the set of all vertices of  $N_{G}(v)\cap U$ that are currently not matched by $M^*$.
\end{properties}

Since, for every vertex $v\in U$, $|\Lambda(v)|=O(z)$ holds, and since the set of all out-neighbors of $v$ in $H$ is contained in $\Lambda(v)$,  the following observation follows immediately from the description of the procedure.
\begin{observation} \label{obs: procupdate time basic}
The running time of $\procupdate(v)$ is $\tilde O(z)$.
\end{observation}

Next, we describe procedure $\procrematchBU$, whose purpose is to attempt to rematch a vertex of $B\cup U$ that became unmatched in $M^*$.

\subsection*{Procedure \procrematchBU}

The input to Procedure $\procrematchBU$ is a vertex $u\in B\cup U$ that is currently not matched by $M^*$. The procedure attempts to rematch $u$ to one of its neighbors in $G$, and then updates all data structures. We provide the description of \procrematchBU in \Cref{fig:procrematchBUbasic}. 

\program{Procedure \procrematchBU}{fig:procrematchBUbasic}
{
	{\bf Input:} a vertex $u\in B\cup U$ that is  not matched by $M^*$.

    \begin{enumerate}
        \item \label{step: H basic} If $H$ contains any incoming edge incident to $u$, let $(v,u)$ be any such edge:
            \begin{itemize}
                \item Insert $(u,v)$ into $M^*$, and call \procupdate for $u$ and $v$.
                \item Return SUCCESS.
            \end{itemize}
        \item  \label{step: check S basic} For every vertex $v\in \hat S$, if $(u,v)\in E(G)$:
            \begin{itemize}
                \item Insert $(u,v)$ into $M^*$, and call \procupdate for $u$ and $v$.
                \item Return SUCCESS.
            \end{itemize}

        \item Return FAILURE.
    \end{enumerate}

}

The following observation summarizes the properties of the procedure.

\begin{observation} \label{obs: procrematchBU basic}
The running time of Procedure $\procrematchBU(u)$ is $\tilde O\left(z+\frac{\rho}{z}\right)$. After the procedure terminates, vertex $u$ becomes settled.
Moreover, for every vertex $v\in V(G)$, if $v$ was settled before the call to the procedure, it remains so after the execution of the procedure.
\end{observation}
\begin{proof}
	We first analyze the running time of the procedure, excluding the time spent on calls to Procedure \procupdate.
	Observe that Step \ref{step: H basic} can be implemented in time $O(1)$, since it only needs to check whether $u$ has an incoming edge in $H$. Step \ref{step: check S basic} can be implemented in time $O(|\hat S|)\leq  O(\rho/z)$ (from  Property~\ref{prop during subphase: matching M1 for basic} and since $M_1\subseteq M^*$), since it only requires scanning the vertices of $\hat S$ and checking, for each such vertex, whether it is a neighbor of $u$ in $G$.  Additionally, the algorithm may call Procedure \procupdate at most twice; the running time for each such call is bounded by $\tilde O\left(z\right)$ from Observation \ref{obs: procupdate time basic}. Altogether, the running time of Procedure $\procrematchBU$ is bounded by $\tilde O\left(z+\frac{\rho}{z}\right)$.

Note that, if the procedure returned SUCCESS, then it inserted an edge incident to $u$ into $M^*$, so vertex $u$ becomes settled.
Assume now that the procedure returned FAILURE, so it did not insert an edge incident to $u$ into $M^*$. We claim that every vertex in $N_G(u)$ must then be matched by $M^*$. Indeed, assume for contradiction that this is not the case, so there exists a vertex $v\in N_G(u)$ that is currently unmatched by $M^*$. If $v\in U$, then by Property \ref{prop: graph H basic}, edge $(v,u)$ must lie in $H$, so $u$ should have been matched in Step \ref{step: H basic}. Otherwise, $v\in S$, and, by Property \ref{prop: set S basic}, $v\in \hat S$ must hold. Then $u$ should have been matched in Step \ref{step: check S basic},  a contradiction. We conclude that, at the end of the procedure, vertex $u$ becomes settled. Lastly, it is immediate to verify that, for every vertex $v\in V(G)$, if $v$ was settled before the call to the procedure, it remains so after the execution of the procedure, since the procedure never deleted any edges from $M^*$.
\end{proof}

Finally, we describe procedure \procrematchA, whose purpose is to attempt to rematch a vertex of $A$ that became unmatched in $M^*$.

\subsection*{Procedure \procrematchA}

The input to Procedure $\procrematchA$ is a vertex $v\in A$ that is currently not matched by $M^*$. The procedure first attempts to match $v$ with some vertex $u\in U$. If successful, it may delete a single edge $e=(u,u')$ from $M^*$ (and possibly from $M_1$), but it guarantees that $u'\in B\cup U$ holds; it then attempts to rematch $u'$ by calling Procedure $\procrematchBU(u')$. If this attempt to match $v$ fails, the procedure then tries to match $v$ with its other neighbors in $G$ that lie in $\hat S$. We provide the description of \procrematchA in Figure \ref{fig:procrematchAbasic}.

\program{Procedure \procrematchA}{fig:procrematchAbasic}
{
	{\bf Input:} a vertex $v\in A$ that is currently not matched by $M^*$.
	
	\begin{enumerate}
		\item \label{step: checking Lv basic} Process the first $\frac{64\rho}{z}+1$ vertices of $L(v)$. 
		 When a vertex $u\in L(v)$ is processed, if $u$ is currently not matched by $M^*$, or if it is matched by $M^*$ to a vertex of $B\cup U$:
		\begin{itemize}
			\item If some edge $e=(u,u')$ incident to $u$ lies in $M^*$:  (note that $u'\in B\cup U$ must hold)
			\begin{itemize}
				\item Delete $e$ from $M^*$. If $e\in M_1$, delete $e$ from $M_1$.
                \item Insert $(u,v)$ into $M^*$, and call \procupdate for $u'$ and $v$.
               
                \item Call $\procrematchBU(u')$.
			\end{itemize}
    		\item Otherwise (if $u$ is not currently matched by $M^*$): Insert $(u,v)$ into $M^*$, and call \procupdate for $u$ and $v$.
			\item Return SUCCESS.
		\end{itemize}
    \item \label{step: checking S for A basic} For every $v'\in \hat S$, if $(v,v')\in E(G)$:
            \begin{itemize}
                \item Insert $(v,v')$ into $M^*$, and call \procupdate for $v$ and $v'$.
                \item Return SUCCESS.
            \end{itemize}
	\item Return FAILURE
	\end{enumerate}
}

The following observation summarizes the properties of Procedure \procrematchA.

\begin{observation} \label{obs: procrematchA basic}
The running time of Procedure $\procrematchA(v)$ is $\tilde O\left(z+\frac{\rho}{z}\right)$. The procedure may delete a single edge $(u,u')$ from $M^*$ and from $M_1$, and in this case $u,u'\in U\cup B$ must hold, and vertices $u$ and $u'$ are settled at the end of the procedure. Moreover, at the end of the procedure, vertex $v$ is settled. Lastly, every vertex $x\in V(G)$ that was settled before the call to the procedure, remains so.
\end{observation}
\begin{proof}
	We first analyze the running time of Procedure $\procrematchA(v)$ excluding the calls to Procedures $\procrematchBU$ and $\procupdate$.
	Step \ref{step: checking Lv basic} of Procedure $\procrematchA(v)$ only needs to process the first $O\left(\frac{\rho}{z}\right )$ vertices of $L(v)$, and the time required for processing each such vertex (excluding the calls to Procedures $\procrematchBU$ and $\procupdate$) is $O(1)$. Therefore, the running time of this step is $O\left(\frac{\rho}{z}\right )$.
		Step \ref{step: checking S for A basic} can be implemented in time $O(|\hat S|)\leq  O\left (\frac {\rho} z\right )$ (from  Property~\ref{prop during subphase: matching M1 for basic} and since $M_1\subseteq M^*$), since it only requires scanning the vertices of $\hat S$ and checking, for each such vertex, whether it is a neighbor of $v$ in $G$.  
	Overall, the running time of Procedure $\procrematchA(v)$ excluding calls to Procedures $\procrematchBU$ and $\procupdate$, is bounded by $\tilde O\left(z+\frac{\rho}{z}\right)$. Note that the procedure may call Procedure \procupdate at most twice; the running time for each such call is bounded by $\tilde O\left(z\right)$ from Observation \ref{obs: procupdate time basic}. 
	Additionally, it may call Procedure \procrematchBU, whose running time is bounded by $\tilde O\left(z+\frac{\rho}{z}\right)$ from Observation \ref{obs: procrematchBU basic}, at most once. Altogether, the running time of Procedure $\procrematchA(v)$, including calls to Procedures $\procrematchBU$ and $\procupdate$, is $\tilde O\left(z+\frac{\rho}{z}\right)$.

Observe that Procedure $\procrematchA(v)$ may only delete an edge from $M^*$ or from $M_1$ in step \ref{step: checking Lv basic}, and, if an edge is deleted from $M_1$, it is also deleted from $M^*$. Therefore,  Procedure $\procrematchA(v)$ may delete at most one edge from $M^*$, and at most one edge from $M_1$. 
If the procedure deletes an edge $(u,u')$ from $M^*$, then it inserts the edge $(v,u)$ into $M^*$. Therefore, prior to the call to $\procrematchBU(u')$, vertices $v$ and $u$ are settled. From Observation \ref{obs: procrematchBU basic}, after the call to $\procrematchBU(u')$, vertices $u,v$ and $u'$ are all settled.

Assume now that Procedure $\procrematchA(v)$ did not delete edges from $M^*$. We claim that $v$ must be settled at the end of the procedure. It is enough to prove that, once the procedure terminates, either $v$ is matched by $M^*$, or every vertex in $N_G(v)$ is matched by $M^*$. Assume for contradiction that this is not the case; in other words, at the end of the procedure, $v$ is not matched by $M^*$, and there is a vertex $u\in N_{G}(v)$ that is not matched by $M^*$. Since we assumed that $v$ is not matched at the end of the procedure, no edges were inserted into $M^*$ by the procedure. Note that it is impossible that $u\in S$, since then, from Invariant \ref{prop: set S basic}, $u\in \hat S$ must hold, and so $v$ must have been matched in Step \ref{step: checking S for A basic}. Therefore, $u\in U$ must hold. Since $L(v)=N_{G}(v)\cap U$ holds throughout the algorithm, $u\in L(v)$ must hold. Since vertex $v$ was not matched in Step \ref{step: checking Lv basic}, we get that $u$ is not among the first $\frac{64\rho}{z}+1$ vertices of $L(v)$; moreover, every one of the first  $\frac{64\rho}{z}+1$ vertices of $L(v)$ must be matched to a vertex of $A$ by $M^*$. Since $L(v)\subseteq U$, this contradicts Invariant \ref{prop during subphase: few vertices of S'' matched to U for basic}. We conclude that $v$ must be settled at the end of the procedure.

 Consider now some vertex $x\in V(G)$ that was settled before the call to Procedure $\procrematchA(v)$. Assume first that $x$ was matched by $M^*$ before the call to the procedure. If an edge incident to $x$ was deleted from $M^*$ by the procedure, then we already showed that $x$ must be settled at the end of the procedure. Otherwise, $x$ remains settled at the end of the procedure. 
	
	Lastly, assume that $x$ was not matched by $M^*$ before the call to Procedure $\procrematchA(v)$, but, for every vertex $y\in N_G(x)$, $y$ was matched by $M^*$. From our discussion above, each such vertex $y$ must remain settled at the end of the procedure. But then, from 
	Observation \ref{obs: settled if neighbors settled}, $x$ is also settled at the end of the procedure.
\end{proof}

We are now ready to complete the description of our algorithm. We start by describing the initialization procedure, which is slightly different for the first phase and for the remaining phases. Then we describe the algorithm for handling edge  deletions.

\subsubsection{Initialization Algorithm for the First Phase}

We initialize the maximal matching $M^*$ and the graph $H$ as follows. 
Recall that,  at the beginning of the algorithm,
we have computed matchings $M_1,\ldots,M_{z+1}$, and we enured that the number of vertices of $S$ that are not matched in $M_1$ is at most $\frac{2n}{z}\leq \frac{2\rho}{z}$.
We start by setting $M^* = M_1$, and by computing the set $\hat S\subseteq S$ of vertices that are not matched in $M^*$, so that $|\hat S|\leq  \frac{2\rho}{z}$ and Property \ref{prop: set S basic} holds.
 We initialize the graph $H$ to contain all vertices of $B\cup U$; for every vertex $u\in U$ that is currently unmatched by $M^*$, for every vertex $v\in \Lambda(u)$, we insert a directed edge $(u,v)$ into $H$. This ensures that Property \ref{prop: graph H basic} holds as well. Since, for every vertex $u\in U$, $|\Lambda(u)|\leq O(z)$, the running time of the initialization algorithm so far is $O(nz)$.

We extend $M^*$ to a maximal matching in the remainder of the initialization procedure, as follows.
Let $X$ contain all vertices of $G$ that are currently not matched in $M^*$. We consider every vertex $x\in X$ in turn. When $x\in X$ is considered, if $x$ is still not matched by $M^*$ at this time: if $x\in B\cup U$, we invoke Procedure $\procrematchBU(x)$, and otherwise, we invoke Procedure $\procrematchA(x)$. 
From Observations \ref{obs: procrematchBU basic} and \ref{obs: procrematchA basic}, once we finish processing $x$, it becomes settled, and it remains so until the end of the initialization procedure.
Therefore, once every vertex in $X$ is processed, matching $M^*$ is guaranteed to be maximal. From Observations \ref{obs: procrematchBU basic} and \ref{obs: procrematchA basic}, the running time of this step is $\tilde O\left(n\left(z+\frac{\rho}{z}\right)\right)$.

Recall that initially,  at most $\tfrac{2 \rho}{z}$ vertices of $S$ were not matched by $M_1$, so $|X\cap S|\leq \tfrac{2 \rho}{z}$ held. Therefore, Procedure \procrematchA was called  at most $\frac{2 \rho}{z}$ times. From Observation \ref{obs: procrematchA basic}, each such such call to \procrematchA may lead to the deletion of at most one edge from $M_1$. Therefore, the total number of vertices of $S$ that are not matched by $M_1$ at the end of this step is bounded by $ \tfrac{2 \rho}{z} + 2\cdot \frac{2\rho}{z}\leq \frac{32\rho}{z}$; in particular, Property~\ref{prop: matching M1 for basic} is guaranteed to hold.
This completes the description of the initialization algorithm for the first phase. Its total running time, including the time required to compute the initial collection $\set{M_1,\ldots,M_{z+1}}$ of matchings, is bounded by $\tilde O\left(n^{1+o(1)}\cdot z+\frac{\rho\cdot n}{z}\right )$, and at the end of this procedure, Property \ref{prop: matching M1 for basic} holds.

\subsubsection{Initialization Algorithm for Subsequent Phases}

We now provide an algorithm for initializing a phase that is not the first phase. 
We let $\hat S'$ be the set of all vertices $v\in S$ that are currently not matched by $M_1$; note that this set may be different from the set $\hat S$ of all vertices that are currently not matched by $M^*$, and that $\hat S\subseteq \hat S'$ holds. Over the course of the initialization procedure, we may update the matching $M_1$, but the set $\hat S'$ of vertices remains fixed and contains all vertices of $S$ that were not matched by $M_1$ at the beginning of the procedure.
If, at the beginning of the phase, Property \ref{prop: matching M1 for basic}  holds, then no further initialization is needed. Therefore, we assume from now on that Property \ref{prop: matching M1 for basic} does not hold. Note that, from Property \ref{prop during subphase: matching M1 for basic}, all but at most $\frac{64\rho}{z}$ vertices of $S$ are currently matched by $M_1$. Therefore, we can assume from now on that $\frac{32\rho}{z}<|\hat S'|\leq \frac{64\rho}{z}$ holds.

Recall that at the beginning of the algorithm, we computed a collection $\set{M_1, \dots, M_{z+1}}$ of disjoint matchings, and, for all $1\leq i\le z+1$ we maintain the counter $m_i$ of the number of vertices of $S$ that are unmatched by $M_i$ throughout the algorithm. Over the course of the algorithm, matching $M_1$ may be modified by the algorithm, but for all $2\leq i\leq z+1$, the only modification that the matching $M_i$ undergoes is the deletion of edges that the adversary deletes from $G$. Let $\Pi$ be the collection of all pairs $(i,v)$, where $2\leq i\leq z+1$ and $v\in S$ is a vertex that is not matched by $M_i$. Clearly, $|\Pi|=\sum_{i=2}^{z+1}m_i$. At the beginning of the algorithm, from Property 
\ref{prop: edges inc to S basic} of the subgraph system, for every vertex $v\in S$, there is at most one index $2\leq i\leq z+1$ such that $v$ is not matched by $M_i$, so $|\Pi|\le |S|\leq n$ held at that time. Since there total number of adversarial edge deletions is bounded by $r$, and since each such edge deletion may affect at most two vertices of $S$, we get that, throughout the algorithm, $|\Pi|\le 2r+n < 4\rho$ always holds (recall that $\rho=n+r$).
Therefore, at the beginning of the current phase, there must exist an index $2\leq i\leq z+1$, with:

\[m_i\leq\frac{\sum_{j=2}^{z+1}m_j}{z}=\frac{|\Pi|}{z}\leq \frac{4\rho}{z}.\]

We fix such an index $2\leq i\leq z$ from now on; from our discussion, the number of vertices of $S$ that are not matched by $M_i$ is at most $\frac{4\rho}{z}$.

We construct the graph $G^*=M_1\cup M_i$ in time $O(n)$. Since $M_1$ and $M_i$ are both matchings, every vertex in $G^*$ has degree at most $2$, and so every connected component of $G^*$ is either a simple path or a cycle (we view isolated vertices as paths of length $0$). 
Consider a vertex $v\in \hat S'$; since $v$ is not matched by $M_1$, its degree in $G^*$ is either $0$ or $1$. Therefore, the connected component of $G^*$ containing $v$ must be a path, that we denote by $P(v)$. Let $v'$ be the other endpoint of the path $P(v)$. We say that vertex $v$ is \emph{problematic} if (i) $v=v'$; or (ii) $v'\in S$ and the last edge on $P(v)$ lies in $M_1$. Notice that, if $v$ is problematic, then $v'$ is not matched by $M_i$. 
Since the number of vertices of $S$ that are not matched by $M_i$ is bounded by $\frac{4\rho}{z}$, at most $\frac{4\rho}{z}$ vertices of $\hat S'$ may be problematic. Next, we process every non-problematic vertex of $\hat S'$ one by one. Over the course of this step, we may delete some edges from $M^*$. We maintain the set $X\subseteq V(G)$ containing every vertex $v\in V(G)$ that was matched by $M^*$ at the beginning of the initialization procedure but now becomes unmatched. We initialize $X=\emptyset$.

\paragraph{Processing a non-problematic vertex of $\hat S'$.}
Let $v\in \hat S'$ be a vertex of $\hat S'$ that is not problematic. We augment the matching $M_1$ along the path $P(v)$ as follows. For every edge $e\in E(P)$, if $e\in M_1$ then we delete $e$ from $M_1$, and otherwise we insert it into $M_1$. It is easy to verify that after this augmentation, $M_1$ remains a valid matching, and vertex $v$ is now matched by $M_1$. The only vertex that was previously matched by $M_1$ and may become now unmatched is $v'$ (if the last edge on $P(v)$ lies in $M_1$), and in this case, $v'\not\in S$ holds. 
Notice that it is also possible that $v'\in \hat S'$ (in which case the last edge on $P(v)$ must lie in $M_i$); in this case, both $v$ and $v'$ become matched in $M_1$, and there is no need to process $v'$ separately later. 
Next, we modify the matching $M^*$ in order to ensure that $M_1\subseteq M^*$ and that $M^*$ remains a valid matching; we also update the vertex set $X$. Specifically, for every edge $e\in E(P)$, if $e$ was deleted from $M_1$ then we delete it from $M^*$ as well, and if it was inserted into $M_1$ then we insert it into $M^*$ as well. Note that every inner vertex $u\in P(v)$ was initially matched in $M_1$ and  in $M^*$, and it remains so; only the endpoints $v,v'$ of $P(v)$ may have changed their status from matched to unmatched with respect to $M_1$, or the other way around. 
For every endpoint $x$ of $P(v)$, if $x$ was matched by $M_1$ but now becomes unmatched (notice that $x=v'$ must hold in this case, and $x$ is now unmatched by $M^*$ as well), we add $x$ to $X$.
For every endpoint $x$ of $P(v)$, if $x$ was unmatched by $M_1$ but now becomes matched by it, if $M^*$ now contains two edges incident to $x$, we delete the edge that is incident to $x$ and does not lie in $M_1$ from $M^*$, and we add the other endpoint of this edge into $X$.

Once very non-problematic vertex of $\hat S'$ is processed, the only vertices of $S$ that may remain unmatched by $M_1$ are the problematic vertices of $\hat S'$, and their number is bounded by $\frac{4\rho}{z}$. Vertex set $X$ now contains every verex $v$ that was matched by $M^*$ at the beginning of the procedure but is currently unmatched by it. 
It is immediate to verify that every non-problematic vertex $v\in\hat S'$ may contribute at most $O(1)$ vertices to set $X$, so $|X|\leq O\left (\frac{\rho}{z}\right )$ must hold.
Lastly, we invoke Procedure $\procupdate(u)$ for every vertex $u$ whose status in $M^*$ changed between matched and unmatched. Since, for every vertex $v\in \hat S'$, at most $O(1)$ vertices  change status while processing $P(v)$, from Observation \ref{obs: procupdate time basic}, the running time of this step is bounded by $\tilde O(|\hat S'|\cdot z)=\tilde O(\rho)$. Overall, the running time of the initialization procedure so far is bounded by $O(n)+\tilde O\left(\frac{\rho}{z}\cdot \left(z+\frac{\rho}{z}\right)\right)\leq \tilde O\left(\rho+\frac{\rho^2}{z^2}\right)$.
From our discussion, all but at most $\frac{4\rho}{z}$ vertices of $S$ are now matched by $M_1$. In the remainder of the initialization procedure, we extend $M^*$ to a maximal matching. This part of the algorithm is similar to that of the initialization procedure for the first phase.

\paragraph{Extending $M^*$ to a maximal matching.} We consider every vertex $x\in X$ in turn. When $x\in X$ is considered, if $x$ is still not matched by $M^*$ at this time: if $x\in B\cup U$, we invoke Procedure $\procrematchBU(x)$, and otherwise, we invoke Procedure $\procrematchA(x)$. 
From Observations \ref{obs: procrematchBU basic} and \ref{obs: procrematchA basic}, the time required to process a single vertex of $X$ is 
 $O\left(z+\frac{\rho}{z}\right)$, and the total running time of this step is bounded by:

\[ \tilde O\left(|X|\cdot \left(z+\frac{\rho}{z}\right)\right)\leq \tilde O\left(\rho+\frac{\rho^2}{z^2}\right).\]
In the following observation we prove that, once all vertices in $X$ are processed, $M^*$ becomes a maximal matching.

\begin{observation}\label{ob: maximal matching at the end}
	After all vertices in $X$ are processed, the matching $M^*$ becomes maximal.
	\end{observation}
\begin{proof}
	Let $E'$ be the set of all edges that lied in $M^*$ before the start of the initialization procedure, but do not lie in the final matching $M^*$, and let $V'$ be the set of all vertices that serve as endpoints of the edges in $E'$. From Observation \ref{obs: maximal matching transformation}, it is enough to show that, at the end of the initialization algorithm, every vertex in $V'$ is settled. Consider now a vertex $v\in V'$. If $v\in X$ then,  from Observations \ref{obs: procrematchBU basic} and \ref{obs: procrematchA basic}, when $v$ was processed, it became settled, and it remains so until the end of the algorithm. Otherwise, the only way that $v$ may lie in $V'$ is if the unique edge of $M^*$ that was incident to $v$ at the start of the initialization procedure was deleted by Procedure $\procrematchA$. However, from Observation \ref{obs: procrematchA basic}, once that procedure terminated, $v$ became settled, and, from  Observations \ref{obs: procrematchBU basic} and \ref{obs: procrematchA basic}, it remains so until the end of the algorithm.
\end{proof}

Recall that, after the first step of the initialization procedure,  at most $\tfrac{4 \rho}{z}$ vertices of $S$ were not matched by $M_1$, so $|X\cap S|\leq \tfrac{4 \rho}{z}$ held. Therefore, Procedure \procrematchA was called at most $\frac{4 \rho}{z}$ times. From Observation \ref{obs: procrematchA basic}, each such call to \procrematchA may lead to the deletion of at most one additional edge from $M_1$. Therefore, the total number of vertices of $S$ that are not matched by $M_1$ at the end of this step is bounded by $ \tfrac{4 \rho}{z} + 2\cdot \frac{4\rho}{z}\leq \frac{32\rho}{z}$; in particular, Property~\ref{prop: matching M1 for basic} is guaranteed to hold.

This completes the description of the initialization algorithm for a phase. Its total running time is $\tilde O\left(\rho+\frac{\rho^2}{z^2}\right)$, and, at the end of this procedure, Property \ref{prop: matching M1} holds.

Recall that the number of phases in the algorithm is bounded by $O(z)$. Therefore, the total time that the algorithm spends on initialization of all phases, including the first one, is bounded by: 

\[
  \tilde O\left(n^{1+o(1)}\cdot  z+\frac{\rho n}{z}\right )+ \tilde O\left(z\cdot \left(\rho+\frac{\rho^2}{z^2}\right )\right ) \leq  \tilde O\left(n^{1+o(1)}\cdot z+\rho z+\frac{\rho^2}{z}\right)
\]

(since $\rho\geq n$).

\subsubsection{Processing an Edge Deletion}

We now provide an algorithm for processing a single edge deletion by the adversary.

Let $e=(u,v)$ be an edge that the adversary deleted. 
If $u\in U$ and $v\in \Lambda(u)$, then we delete $v$ from $\Lambda(u)$. If $u\in A$ and $v\in L(u)$, then we delete $v$ from $L(u)$. Similarly, if $u\in \Lambda(v)$ or $u\in L(v)$, we delete $u$ from the corresponding lists.
If $e$ lies in some matching $M_i$, for $2\leq i\leq z+1$, then we delete $e$ from $M_i$ and we update the counter $m_i$ as needed. Finally, if either of the edges $(u,v)$ or $(v,u)$ lies in $H$, we delete it from $H$.

If $e\not\in M^*$, then no further updates are needed. Therefore, we assume that $e\in M^*$ from now on. We delete $e$ from $M^*$, and call \procupdate for $u$ and for $v$. By Observation \ref{obs: procupdate time basic}, the running time of the procedure is $\tilde O\left(z\right)$. If $e$ lies in $M_1$, then it is also deleted from $M_1$. For every vertex $x\in \set{u,v}$, if $x\in B\cup U$, then we invoke Procedure $\procrematchBU(x)$, and otherwise we invoke Procedure $\procrematchA(x)$. From Observations \ref{obs: procrematchBU basic} and \ref{obs: procrematchA basic}, the time we spend on processing each of the vertices $u$, $v$ is $\tilde O\left(z+\frac{\rho}{z}\right)$, and so the total running time of the entire update algorithm is $\tilde O\left(z+\frac{\rho}{z}\right)$.

We claim that the resulting matching $M^*$ is maximal. In order to prove this, it is enough to show that every vertex in $G$ is settled with respect to this final matching. Indeed, consider any vertex $a\in V(G)$. If $a\in \set{u,v}$, then, from Observations \ref{obs: procrematchBU basic} and \ref{obs: procrematchA basic}, once we finish processing $a$, it becomes settled, and it remains so until the end of the algorithm. Assume now that $a\not\in \set{u,v}$. Assume first that the original matching $M^*$ contained an edge $e'$ incident to $a$. Then either $e'$ lies in the final matching $M^*$, or it was deleted by Procedure $\procrematchA$. In the latter case, from Observation \ref{obs: procrematchA basic}, once that procedure terminated, $a$ became settled, and, from Observations \ref{obs: procrematchBU basic} and \ref{obs: procrematchA basic}, it remains so until the end of the algorithm. Therefore, if the initial matching $M^*$ contained an edge incident to $a$, then $a$ is settled with respect to the final matching $M^*$. Lastly, assume that the original matching $M^*$ did not contain an edge incident to $a$. Since $M^*$ was a maximal matching, every vertex in $N_G(a)$ was matched by $M^*$. From the above discussion, every vertex in $N_G(a)$ must be settled in the final matching $M^*$. From Observation \ref{obs: settled if neighbors settled}, it then follows that $a$ is settled in the final matching $M^*$.

Note that if $e\in M_1$, then $e$ is deleted from $M_1$. Moreover, the update procedure calls Procedure \procrematchA at most twice, and each such call may lead to the deletion of at most one additional edge from $M_1$. Recall that by Property \ref{prop: matching M1 for basic}, after the initialization procedure of the current phase, at most $\frac{32\rho}{z}$ vertices of $S$ are  not matched by $M_1$. Since the total number of edge deletions by the adversary in a phase is bounded by $\frac{\rho}{z}$, and since each such edge deletion may lead to up to $2$ additional deletions of edges from $M_1$, we then get that, throughout the phase, the total number of vertices of $S$ that are not matched by $M_1$ remains bounded by:

\[\frac{32\rho}{z}+6\cdot \frac{\rho}{z}\leq \frac{64\rho}{z}.\] 

Therefore, Properties \ref{prop during subphase: matching M1 for basic} and \ref{prop during subphase: few vertices of S'' matched to U for basic} hold throughout the phase.

\paragraph{Bounding the Running Time.}
The total time the algorithm spends on initializing the data structures at the beginning of every phase, including on computing the initial matchings, is bounded by:

$$\tilde O\left(n^{1+o(1)}\cdot z+\rho z+\frac{\rho^2}{z}\right).$$

 The time required to process each adversarial edge deletion is bounded by $ O\left(z+\frac{\rho}{z}\right)$. Since the total number of such edge deletions over the course of the algorithm is bounded by $r\leq \rho$, we get that the total running time of the algorithm is bounded by:

\[\tilde O\left(n^{1+o(1)}\cdot z+\rho z+\frac{\rho^2}{z}\right)+\tilde O\left(\rho z+\frac{\rho^2}{z}\right)\leq \tilde O\left(n^{1+o(1)}\cdot z+rz+\frac{(n+r)^2}{z}\right),\]

since $\rho=n+r$.

\section{Main Algorithm: Proof of \Cref{thm: main}}
\label{sec: main alg}

In this section we provide our main result -- the proof of  \Cref{thm: main}, with some technical details deferred to the following sections. We start by describing the main new technical tool that we use -- the multi-level subgraph system, and the two key subroutines used in our algorithm. The first subroutine allows us to efficiently compute a multi-level subgraph system with appropriate parameters. The second subroutine allows us to maintain a maximal matching in a fully dynamic graph $G$, over a limited number of updates, given such a system. 
We only state the theorems summarizing these two subroutines in this section, and we defer their proofs to the following sections.
We then complete the proof of \Cref{thm: main} using these tools.

\subsection{A Multi-Level Subgraph System}

The main new tool that we use is the \emph{multi-level subgraph system}, that we define next.

\begin{definition}[Multi-Level Subgraph System]\label{def: subgraph structure}
	Let $G$ be an $n$-vertex graph, and let $k>0$ and $1\leq z\leq n$ be integral parameters. A \emph{$k$-level $z$-subgraph system for $G$} consists of the following ingredients:
	\begin{itemize}
		\item A subset $M\subseteq E(G)$ of edges. For every vertex $v\in V(G)$, we denote by $m(v)$ the number of edges of $M$ that are incident to $v$, and we require that the following property holds:
		
		\begin{properties}{P}
			\item \label{prop: small degrees in M} For every vertex $v\in V(G)$, $m(v)\leq z$;
		\end{properties}

        \item A partition $(S,U)$ of $V(G)$, for which the following properties hold:
            \begin{properties}[1]{P}
            		\item \label{prop: no edge of M w both endpoints in U} No edge of $M$ has both endpoints in $U$;
            		
                \item \label{prop: edges inc to S} For every vertex $v\in S$, $m(v)\geq z-k+1$; and
                \item \label{prop: degrees in U} For every vertex $u\in U$, $|N_G(u)\cap U|\leq z$.
            \end{properties}
		\item A partition $(A_1,\dots,A_k,B)$ of $S$ that satisfies the following property:
		\begin{properties}[4]{P}
            \item \label{prop: U neighbors in Ak} For every vertex $u\in U$, $|N_G(u)\cap B|\leq 2z$.
		\end{properties}
        \item For every vertex $u\in U$, a list $\Lambda(u)=N_G(u)\cap (B\cup U)$, whose length must be $O(z)$;

        For all $1\leq i\leq k$, we denote by $A^{\leq i}=A_1\cup\dots\cup A_{i}$ and by $A^{\geq i}=A_i\cup \dots\cup A_k$. For consistency, we also denote $A^{\leq 0}=A^{\geq k+1}=\emptyset$.
        \item For all $1\leq i\leq k$, a vertex set $N_i\subseteq A^{\geq i+1}\cup B$, and the vertex set $R_i=(A^{\geq i+1}\cup B\cup U)\setminus N_i$ such that the following properties hold:
        \begin{properties}[5]{P}
            \item \label{prop: neighbors of Ai in M in Ni} For every vertex $v\in A_i$, for every edge $(u,v)\in M$ incident to $v$, $u\in N_i\cup A^{\le i}$; and
            \item \label{prop: containment of Rs}  $U= R_k\subseteq R_{k-1} \subseteq \dots \subseteq R_1$.
		\end{properties}
        \item For all $1\le i\leq k$, for every vertex $v\in A_i$, the list $L(v)=N_G(v)\cap R_i$.
		\end{itemize}
We denote the $k$-level $z$-subgraph system by $\sset = \left(S,B,U,M,\set{A_i,N_i,R_i}_{i=1}^k,\Lambda,L\right )$, where $\Lambda=\set{\Lambda(u)}_{u\in U}$ and $L=\set{L(v)}_{v\in A}$.
\end{definition}

Note that, from Property \ref{prop: containment of Rs}, $R_k=U$ must hold, and so $N_k=B$, and $A^{\leq k}\cup B=S$  must hold. 
Requirement \ref{prop: neighbors of Ai in M in Ni} for the vertices of $A_k$ is then equivalent to the requirement that, for every vertex $v\in A_k$,  for every edge $(u,v)\in M$ incident to $v$, $u\in S$ must hold.
It is then immediate to verify that the definition of the $1$-level $z$-subgraph system is identical to the definition of basic $z$-subgraph system used in \Cref{sec: simple alg} (see Definition \ref{def: basic subgraph structure}).

For intuition for the definition of the multi-level subgraph system, consider the sets $A_1,\ldots,A_k,B,U$ as being ``stacked'' on top of each other (see Figure \ref{fig: stacking}), so for $1\leq i\leq k$, we view the sets $A_1,\ldots,A_{i-1}$ as lying ``above'' $A_i$, and $A_{i+1}\cdots,A_k,B,U$ as lying ``below'' it. The set of all vertices that lie below $A_i$ is partitioned into two subsets $N_i$ and $R_i$, with $U\subseteq R_i$. The requirement is that, if an edge $(u,v)\in M$ connects a vertex $u\in A_i$ to a vertex $v$ that lies below $A_i$, then $v\in N_i$ must hold; so in particular $v\not\in U$ (see Figure \ref{fig: edges}). Moreover,
from Property \ref{prop: containment of Rs}, $M$ may not contain edges connecting vertices that lie above $A_i$ to vertices of $R_i$.
On the other hand, every vertex $v\in A_i$ must store a list $L(v)=N_G(v)\cap R_i$; these lists may be arbitrarily long. Lastly, the set $B$ and $U$ must have the property that, for every vertex $u\in U$, the number of its neighbors in $G$ that lie in $B\cup U$ is small (at most $O(z)$), and all such neighbors must be stored in list $\Lambda(u)$; note that $u$ may have arbitrarily many neighbors in $S\setminus B$. We note that we could have defined, for all $1\leq i\leq k$, $N_i=A^{\geq i+1}\cup B$ and $R_i=U$. The main problem with this definition is that, since the lists $L(v)$ for $v\in S\setminus U$ may be arbitrarily long, it is challenging to construct them efficiently enough, when building a $k$-level subgraph system from a $(k-1)$-level subgraph system. With the current definition, we can allow the sets $N_i$ and the lists $L(v)$ for $v\in A_i$, for all $1\leq i\leq k-1$, to remain unchanged when constructing a $k$-level subgraph system from a $(k-1)$-level one, resulting in a much more efficient algorithm.

\begin{figure}[t]
    \centering
    \begin{subfigure}[t]{0.45\textwidth}
        \centering
        \begin{tikzpicture}[x=1cm,y=1cm, font=\small]
  \def\i{3}
  \def\k{6}
  \def\w{5.0}
  \def\h{0.7}
  \def\hell{1.0}
  \def\hU{2.0}
  \def\labeldx{0.35}

  \newcommand{\toprect}[5]{
    \path (#2,#3) coordinate (#1NW);
    \path (#2+#4,#3-#5) coordinate (#1SE);
    \draw (#1NW) rectangle (#1SE);
  }
  \newcommand{\labelright}[3]{
    \node[#2] at ($ (#1NW)!0.5!(#1SE) $) {#3};
  }

  \coordinate (cur) at (0,0);

  \draw[thick, black] (cur) rectangle ++(\w,-\h);
  \path (cur) ++(0,-\h -0.01) coordinate (cur);
  \path (0,0) coordinate (AoneNW);
  \path (\w,-\h) coordinate (AoneSE);
  \labelright{Aone}{black}{$A_1$}
  
  \draw[thick, black] (cur) -- ++(0,-\hell) -- ++(\w,0) -- ++(0,\hell);
  \node[black] at ($ (cur) + (\w/2,-\hell*0.4) $) {\Large \vdots};
  \path (cur) ++(0,-\hell) coordinate (cur);

  \draw[thick, black] (cur) -- ++(0,-\h) -- ++(\w,0) -- ++(0,\h);
  \path (cur) coordinate (AiNW);
  \path ($ (cur) + (\w,-\h) $) coordinate (AiSE);
  \labelright{Ai}{black}{$A_i$}
  \path (cur) ++(0,-\h) coordinate (cur);

  \draw[thick, black] (cur) -- ++(0,-\hell) -- ++(\w,0) -- ++(0,\hell); 
  \path (cur) coordinate (NiTopLeft);
  \node[black] at ($ (cur) + (\w/2,-\hell*0.4) $) {\Large \vdots};
  \path (cur) ++(0,-\hell) coordinate (cur);

  \draw[thick, black] (cur) -- ++(0,-\h) -- ++(\w,0) -- ++(0,\h);
  \path (cur) coordinate (AkNW);
  \path ($ (cur) + (\w,-\h) $) coordinate (AkSE);
  \labelright{Ak}{black}{$A_k$}
  \path (cur) ++(0,-\h-0.02) coordinate (cur);

  \draw[thick, blue] (cur) rectangle ++(\w,-\h);
  \path (cur) coordinate (BNW);
  \path ($ (cur) + (\w*0.4,-\h)$) coordinate (NiBotRight);
  \path ($ (cur) + (\w,-\h) $) coordinate (BSE);
  \labelright{B}{blue}{$B$}
  \path (cur) ++(0,-\h-0.02) coordinate (cur);

  \draw[thick, red] (cur) rectangle ++(\w,-\hU);
  \path (cur) coordinate (UNW);
  \path ($ (cur) + (\w,-\hU) $) coordinate (USE);
  \labelright{U}{red}{$U$}

\end{tikzpicture}
        \caption{Stacking the sets $A_1,\ldots,A_k,B$ and $U$. Every vertex $u\in U$ satisfies $|N_G(u)\cap (U\cup B)|\leq O(z)$.}\label{fig: stacking}
    \end{subfigure}
    ~
    \begin{subfigure}[t]{0.45\textwidth}
        \centering
        \begin{tikzpicture}[x=1cm,y=1cm, font=\small]
    \def\i{3}
    \def\k{6}
    \def\w{5.0}
    \def\nw{2.1}
    \def\h{0.7}
    \def\hell{1.0}
    \def\hU{2.0}
    \newcommand{\labelright}[3]{
        \node[#2] at ($ (#1NW)!0.5!(#1SE) $) {#3};
    }

    \coordinate (cur) at (0,0);

    \fill[blue!25] (cur) rectangle ++(\w,-\h);
    \draw[thick, black] (cur) rectangle ++(\w,-\h);
    \path (cur) ++(0,-\h -0.01) coordinate (cur);
    \path (0,0) coordinate (AoneNW);
    \path (\w,-\h) coordinate (AoneSE);
    \labelright{Aone}{black}{$A_1$}

    \fill[blue!25] (cur) rectangle ++(\w,-\hell);
    \draw[thick, black] (cur) -- ++(0,-\hell) -- ++(\w,0) -- ++(0,\hell);
    \node[black] at ($ (cur) + (\w/2,-\hell*0.4) $) {\Large \vdots};
    \path (cur) ++(0,-\hell) coordinate (cur);

    \fill[blue!25] (cur) rectangle ++(\w,-\h);
    \draw[thick, black] (cur) rectangle ++(\w,-\h);
    \path (cur) coordinate (AiNW);
    \path ($ (cur) + (\w,-\h) $) coordinate (AiSE);
    \labelright{Ai}{black}{$A_i$}
    \path (cur) ++(0,-\h) coordinate (cur);

    \draw[thick, black] (cur) -- ++(0,-\hell) -- ++(\w,0) -- ++(0,\hell); 
    \path (cur) coordinate (NiTopLeft);
    \node[black] at ($ (cur) + (\w/2,-\hell*0.4) $) {\Large \vdots};
    \path (cur) ++(0,-\hell) coordinate (cur);

    \draw[thick, black] (cur) -- ++(0,-\h) -- ++(\w,0) -- ++(0,\h);
    \path (cur) coordinate (AkNW);
    \path ($ (cur) + (\w,-\h) $) coordinate (AkSE);
    \labelright{Ak}{black}{$A_k$}
    \path (cur) ++(0,-\h) coordinate (cur);

    \draw[thick, black] (cur) rectangle ++(\w,-\h);
    \path (cur) coordinate (BNW);
    \path ($ (cur) + (\nw,-\h)$) coordinate (NiBotRight);
    \path ($ (cur) + (\w,-\h) $) coordinate (BSE);
    \labelright{B}{black}{$B$}
    \path (cur) ++(0,-\h) coordinate (cur);

    \draw[thick, black] (cur) rectangle ++(\w,-\hU);
    \path (cur) coordinate (UNW);
    \path ($ (cur) + (\w,-\hU) $) coordinate (USE);
    \labelright{U}{black}{$U$}

    \draw[fill = red, fill opacity=0.3, thick, red] ($(NiTopLeft)+(\nw,0)$) -- ($(NiTopLeft)+(\w,0)$) -- (USE) -- ($(cur) + (0,-\hU)$) -- (UNW) -- (NiBotRight) -- ($(NiTopLeft)+(\nw,0)$);
    \node[red] at ($(NiBotRight)+(\w*0.15, -\hU*0.25)$) { \Large $R_i$};
  
    \draw[fill=green!40, fill opacity=0.5, draw=green!60!black, thick, rounded corners=2pt]
    (NiTopLeft) rectangle (NiBotRight);
    \node[green!40!black]
    at ($(NiTopLeft)!0.5!(NiBotRight) $) {\large $N_i$};

    \draw ($(AiNW) + (\w*0.25,-\h*0.4)$) node(1)[circle, draw, fill=black!50,
	inner sep=0pt, minimum width=4pt, label = 180:{$v$}] {};
  
    \draw ($(AiNW) + (\w*0.18,-\h - \hell *0.6)$) node(2)[circle, draw, fill=black!50,
	inner sep=0pt, minimum width=4pt] {};

    \draw ($(BNW) + (\w*0.35,-\h *0.5)$) node(3)[circle, draw, fill=black!50,
	inner sep=0pt, minimum width=4pt] {};
  
    \draw ($(AiNW) + (\w*0.8,-\h*0.7)$) node(4)[circle, draw, fill=black!50,
	inner sep=0pt, minimum width=4pt] {};
    \draw ($(AoneNW) + (\w*0.7,-\h-\hell*0.3)$) node(5)[circle, draw, fill=black!50,
	inner sep=0pt, minimum width=4pt] {};
    \draw ($(AoneNW) + (\w*0.6,-\h*0.5)$) node(6)[circle, draw, fill=black!50,
	inner sep=0pt, minimum width=4pt] {};

    \draw [semithick] (1) to (2);
    \draw [semithick] (1) to (3);
    \draw [semithick] (1) to (4);
    \draw [semithick] (1) to (5);
    \draw [semithick] (1) to (6);

\end{tikzpicture}
        \caption{$N_i$ is the green region and $R_i$ is the red region. For every edge $e=(u,v)\in M$ with $v\in A_i$, its other endpoint $u$ must lie in the green or the blue region. No edges of $M$ may connect vertices in the blue region to vertices in the red region}\label{fig: edges}
    \end{subfigure}
	\caption{Illustration of the subgraph system}\label{fig: illustration}
\end{figure}
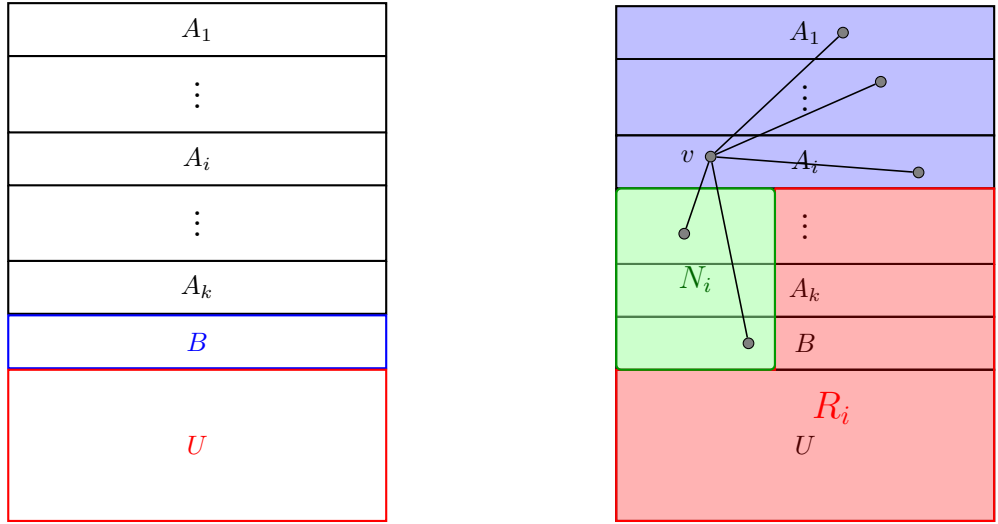

\subsection{From Subgraph System to Fully Dynamic Maximal Matching}
A central ingredient of our approach is an algorithm that, given a fully dynamic graph $G$ and a $k$-level $z$-subgraph system $\sset$ for an initial graph $G$ (that is {\bf not} assumed to be empty), maintains a maximal matching in $G$, as it undergoes an online sequence of (a limited number) of edge insertions and deletions. The algorithm is summarized in the following theorem, whose proof appears in Section \ref{sec: subgraph-system-to-matching}.

\begin{theorem}\label{thm: matching from subgraph system}
There is a deterministic algorithm, whose input consists of an initial $n$-vertex graph $G$ given in the adjacency-list representation,  together with integral parameters $0<k\leq \log n$ and $1\leq z\leq n$, and a  $k$-level $z$-subgraph system $\sset$ for $G$. The algorithm maintains a maximal matching in $G$, as it undergoes an online sequence of at most $n$ edge insertions and deletions, in total time $O\left(n^{1+o(1)}\cdot \left(z+\tfrac{n}{z}\right)\right)$.
\end{theorem}

\subsection{Algorithms for Constructing a Subgraph System}

We start with the following theorem, that provides an efficient algorithm for computing a $1$-level $z$-subgraph system for a graph $G$, for any given integral parameter $1\le z\leq n$. Since, as discussed already, a $1$-level $z$-subgraph system is equivalent to the basic subgraph system introduced in \Cref{sec: simple alg}, the theorem follows immediately from \Cref{thm: constructing simple subgraph system}.

\begin{theorem}\label{thm: level-1 constr of subgraph system}
	There is a deterministic algorithm, that, given an $n$-vertex and $m$-edge graph $G$ and a parameter $z\in [n]$, constructs a 1-level $z$-subgraph system $\sset = \left(S,B,U,M,A_1,N_1,R_1,\Lambda,L\right )$ for $G$, in time $\tilde O(n+m)$.
\end{theorem}

Note that the running time of the algorithm from \Cref{thm: level-1 constr of subgraph system} may be as high as $\Omega(n^2)$, if the input graph $G$ is dense. Therefore, we cannot afford to recompute the subgraph system too often. For example, if our goal is to obtain an algorithm for maximal matching whose amortized update time is $n^{1/2+o(1)}$, then we can only afford to recompute it every time the adversary makes roughly $n^{3/2}$ edge updates to the graph $G$. On the other hand, the running time of the algorithm from \Cref{thm: matching from subgraph system}, that maintains a maximal matching in $G$ given an initial subgraph system, is $O\left(n^{1+o(1)}\left(z+\tfrac{n}{z}\right)\right)$ for $n$ updates; therefore, if our goal is to achieve amortized update time $n^{1/2+o(1)}$, we should choose $z$ to be roughly $\sqrt{n}$. Lastly, if the parameter $z$ is not too large, then the initial subgraph system may become significantly damaged after a relatively small number of edge deletions by the adversary, so, for example, after $\omega(n)$ edge updates to $G$, the subgraph system computed for the initial graph $G$ may no longer be possible to use by the algorithm from  \Cref{thm: matching from subgraph system} in order to continue and maintain a maximal matching in $G$.

 In order to overcome these difficulties, we select a parameter $k=\Theta(\log n)$, and a sequence $z_1\geq z_2\geq\cdots\geq z_k$ of integral parameters, where $z_1$ is close to $n$ and $z_k$ is close to $\sqrt{n}$, while $\frac{z_{i-1}}{z_{i}}=2$ for all $1<i\leq k$. Our algorithm will construct a $1$-level $z_1$-subgraph system every once in a while (after a sufficiently large number of edge updates by the adversary). We will not use this subgraph system in order to maintain a maximal matching in $G$ directly; instead, we will use it in order to compute a $2$-level $z_2$-subgraph system. This subgraph system, in turn, will be used in order to compute a $3$-level $z_3$-subgraph system and so on. 

The key tool that we use is an algorithm that,  given an $h$-level $z$-subgraph system for a graph $G$, efficiently constructs an $(h+1)$-level $z'$-subgraph system for $G$, for a given parameter $z'<z$.  The main idea is that the $h$-level $z$-subgraph system can be exploited in order to speed up the construction of the $z'$-subgraph system (but unfortunately this speedup comes at the cost of increasing the number of levels in the subgraph system); the running time of the algorithm is bounded by $O\left(n^{1+o(1)}\cdot z\right )$, which may be significantly smaller than $|E(G)|$ if $G$ is sufficiently dense and the parameter $z$ is sufficiently small. This speedup allows us to compute the subgraph systems for higher levels $i$ more often, and this is what allows us to set the parameters $z_i$ to be progressively smaller (recall that the smaller the parameter $z$ of the subgraph system, the faster the adversary may destroy the subgraph system via edge deletions, so a subgraph system with a lower parameter $z$ needs to be recomputed more often; therefore, we need to ensure that its computation is more efficient).

For technical reasons, in our algorithm that computes an $(h+1)$-level $z'$-subgraph system from an $h$-level $z$-subgraph system, we assume that we are given, as part of input, a subset $E_D\subseteq E(G)$ of ``deleted'' edges, and we would like to ensure that the new $(h+1)$-level subgraph system only contains a limited number of such edges. 
Intuitively, these are edges that were deleted by the adversary since the $h$-level subgraph system was constructed, but we temporarily keep them in $G$ because their removal may have a large impact on the $h$-level subgraph system. Instead, we will ensure that the final $k$-level subgraph system only contains few such edges. We also assume that we are given a set $E_I$ of edges that do not lie in the input graph; these are edges that the adversary has inserted into $G$ since the $h$-level subgraph system was constructed but they were not yet incorporated into the $h$-level subgraph system (so for example there may be many such edges connecting a vertex $u\in U$ to other vertices in $U$). We require that these edges are incorporated into the new $(h+1)$-level subgraph system.
The following theorem, whose proof appears in \Cref{sec: recursive subgraph system}, provides an algorithm for constructing an $(h+1)$-level $z'$-subgraph system from an $h$-level $z$-subgraph system.

\begin{theorem}\label{thm: inductive constr of subgraph system}
There is a deterministic algorithm, whose input consists of:

\begin{itemize}
	\item an $n$-vertex graph $G$ given in the adjacency-list representation;
	\item an integral parameter $h>0$ and a parameter $1\leq z\leq n$ that is an integral power of $2$;
	\item an $h$-level $z$-subgraph system $\sset = \left(S,B,U,M,\set{A_i,N_i,R_i}_{i=1}^h,\Lambda,L\right )$ for $G$; 
	\item a subset $E_D\subseteq E(G)$ of edges; 
	\item a set $E_I\subseteq V(G)\times V(G)$ of edges that do not lie in $G$; and
	\item a parameter $1\leq z'<z$ that is an integral power of $2$.
\end{itemize}
The algorithm constructs a subset $E_D'\subseteq E_D$ of at most $\frac{|E_D|\cdot z'}{z}$ edges and an $(h+1)$-level $z'$-subgraph system $\sset' =\left(S',B',U',M',\set{A_i',N_i',R_i'}_{i=1}^{h+1},\Lambda',L'\right )$ for the graph $G'=(G\cup E_I)\setminus (E_D\setminus E_D')$. The running time of the algorithm is $\tilde O(n^{1+o(1)}\cdot z+|E_D|+|E_I|)$. The algorithm may modify the input adjacency-list representation of $G$ and the input representation of the $h$-level $z$-subgraph system $\sset$ for $G$.
\end{theorem}

We note that generally, the size of the $(h+1)$-level $z'$-subgraph system $\sset'$ that the algorithm from \Cref{thm: inductive constr of subgraph system} constructs may be quite large (specifically, the lists in $L'(v)$ for $v\in A'$ may be very long), and so the algorithm may not be able to output it explicitly within its required running time. Instead, it modifies the input $h$-level $z$-subgraph system $\sset$ in order to transform it into the $(h+1)$-level $z'$-subgraph system $\sset'$ for $G'$.

\subsection{Completing the Proof of  \Cref{thm: main}}
\label{sec: high-level}

We are now ready to complete the proof of \Cref{thm: main}. We denote by $r$ the total number of adversarial updates to the input graph $G$, that may not be known to the algorithm in advance. It will be convenient for us to assume that $n=|V(G)|$ is an integral power of $2$. If this is not the case, then we add dummy vertices into $G$ until $n=|V(G)|$ becomes an integral power of $2$. Note that this may only increase $n$ by at most factor $2$.
Our algorithm works in phases. We partition the time horizon into phases $\Phi_0,\Phi_1,\Phi_2,\ldots$ as follows. The first phase, $\Phi_0$, starts immediately after the first edge is inserted into the initially empty graph $G$, and lasts for $r_0=\min\set{n,r}$ updates, after which Phase $\Phi_1$ starts.
For all $i\ge 1$, we denote by $m_i$ the number of edges in graph $G$ at the beginning of Phase $\Phi_i$. If $m_i\leq n^{3/2}$, then the $i$th phase lasts for exactly $r_i=n$ updates. In this case, we say that Phase $\Phi_i$ is of \emph{type 1}. Otherwise, the $i$th phase lasts for $r_i=m_i$ updates, and we say that it is of \emph{type 2}. For all $i\geq 0$, if the update sequence terminates before $r_i$ updates are completed since the beginning of Phase $\Phi_i$, it becomes the last phase of the algorithm.

The first phase $\Phi_0$ is special, and is implemented by a simple algorithm summarized in the following lemma; the proof is deferred to Section \ref{subsubsec: type 0 phase}. Intuitively, this algorithm is needed for the case where the length of the entire update sequence is relatively short, for example, below $n$.

\begin{lemma} \label{lem: small number of updates}
	There is a deterministic algorithm, that, given as input an $n$-vertex graph $G$ that initially contains no edges, and undergoes an online sequence of $r>0$ edge deletions and insertions, together with an integral parameter $t>0$, maintains a maximal matching in $G$. The running time of the algorithm is $\tilde O\left(r\left(t+\frac{r}{t}\right)\right )$.
\end{lemma}

We will employ the algorithm from Lemma \ref{lem: small number of updates} for Phase $\Phi_0$ with parameter $t=n^{1/2}$. 
The algorithm for implementing type-1 phases is summarized in the following theorem, whose proof is deferred to \Cref{subsubsec: type 1 phase}.

\begin{theorem}\label{thm: type 1 phase}
	There is a deterministic algorithm, that, given an $n$-vertex and $m$-edge graph $G$ with $m\leq n^{3/2}$, that undergoes an online sequence of $r\leq n$ edge insertions and deletions, maintains a maximal matching in $G$ with total running time $ O\left(n^{3/2+o(1)}\right )$.
\end{theorem}

The following theorem summarizes our algorithm for implementing type-2 phases; the proof is deferred to \Cref{subsubsec: type 2 phase}.

\begin{theorem}\label{thm: type 2 phase}
There is a deterministic algorithm, that, given an $n$-vertex and $m$-edge graph $G$ with $m> n^{3/2}$ and $n$ an integral power of $2$, that undergoes an online sequence of $r\leq m$ edge insertions and deletions, maintains a maximal matching in $G$ with total running time $O\left(m\cdot n^{1/2+o(1)}\right )$.
\end{theorem}

We are now ready to complete the proof of \Cref{thm: main}. 
We start by employing the algorithm from Lemma \ref{lem: small number of updates} with parameter $t=\ceil{\sqrt{n}}$ over the course of Phase $\Phi_0$. Recall that the running time of the algorithm from Lemma \ref{lem: small number of updates} is bounded by:

\[
\tilde O\left(r_0\cdot \left(t+\frac{r_0}{t}\right)\right)\leq \tilde O\left(r_0\cdot \left(n^{1/2}+\frac{n}{n^{1/2}}\right)\right) \leq \tilde O\left(r_0\cdot n^{1/2}\right).
\]

If $\Phi_0$ is the only phase (that is, the number of adversarial updates is $r<n$), then the amortized update time of the algorithm is bounded by $\tilde O(n^{1/2})$, as required. Assume now that $r>n$. Let $q$ denote the index of the last phase.
For all $1\leq i\leq q$, let $\tau_i$ be the time when Phase $\Phi_i$ starts, and recall that $m_i$ is the number of edges in graph $G$ at time $\tau_i$. Recall also that $r_i$ is the number of updates that $G$ undergoes during Phase $\Phi_i$. If $m_i\leq n^{3/2}$ (that is, $\Phi_i$ is a type-1 phase), then we use the algorithm from \Cref{thm: type 1 phase} with the initial graph $G\attime[\tau_i]$, and the updates that $G$ undergoes during Phase $\Phi_i$,  to maintain a maximal matching in $G$ during the phase $\Phi_i$. Recall that, if $i<q$, then $r_i=n$, and the running time of the algorithm from \Cref{thm: type 1 phase} is $ O\left(n^{3/2+o(1)}\right )\leq  O\left(r_i\cdot n^{1/2+o(1)}\right ) $. If $m_i>n^{3/2}$ (so $\Phi_i$ is a type-2 phase), 
then we use the algorithm from \Cref{thm: type 2 phase} with the initial graph $G\attime[\tau_i]$, and the updates that $G$ undergoes during Phase $\Phi_i$, to maintain a maximal matching in $G$ during the phase $\Phi_i$. Recall that, if $i<q$, then $r_i=m_i$, and the running time of the algorithm from \Cref{thm: type 2 phase} is $ O\left(m_i\cdot n^{1/2+o(1)}\right )\leq  O\left(r_i\cdot n^{1/2+o(1)}\right ) $. If the last phase is of type 1, then its running time is $ O\left(n^{3/2+o(1)}\right )\leq  O\left(r\cdot n^{1/2+o(1)}\right ) $, since we assumed that $r\geq n$. Otherwise, the running time of the last phase is $O\left(m_q\cdot n^{1/2+o(1)}\right )\leq O\left(r\cdot n^{1/2+o(1)}\right ) $, since $r\geq m_q$ must hold, as the initial graph $G$ contained no edges. Overall, the total running time of the algorithm is bounded by:

\[ O\left(r\cdot n^{1/2+o(1)}\right )+\sum_{i=0}^{q-1}O\left(r_i\cdot n^{1/2+o(1)}\right )\leq  O\left(r\cdot n^{1/2+o(1)}\right ).\]

In order to complete the proof of \Cref{thm: main}, it is now enough to prove Lemma \ref{lem: small number of updates} and Theorems \ref{thm: type 1 phase} and \ref{thm: type 2 phase}, which we do in the remainder of the section.

\subsubsection{Implementing Phase $\Phi_0$: Proof of Lemma \ref{lem: small number of updates}}
\label{subsubsec: type 0 phase}
Recall that we are given as input an $n$-vertex graph $G$ that initially contains no edges, and then undergoes an online sequence of $r$ edge insertions and deletions. Throughout the algorithm, we say that a vertex $v\in V(G)$ is \emph{bad} if the number of edges \emph{inserted} by the adversary over the course of the algorithm that are incident to $v$ is at least $t$,
and we say that $v$ is \emph{good} otherwise. Note that the same edge may
be inserted and deleted numerous times, but once a vertex $v$ becomes bad it remains so until the end of the algorithm. Since the total number of updates is bounded by $r$, the total number of bad vertices is bounded by $\frac{2r}{t}$ at all times.

Our algorithm maintains a maximal matching $M^*$ in $G$, where initially $M^*=\emptyset$. Additionally, our algorithm maintains a set $B\subseteq V(G)$ of all bad vertices, and, for every bad vertex $b\in B$, a list  $F(b)$ of all vertices in $N_G(b)$  that are currently not matched by $M^*$, as a binary search tree data structure. At the beginning of the algorithm, $B=\emptyset$. We now provide algorithms for processing edge insertions and deletions. The algorithms consist of two parts. In the first part, we update the matching $M^*$; this part is different for the two update types. Then in the second part, we update the data structure $\set{F(b)\mid b\in B}$; this part is similar for the two update types.

\subsubsection*{Part 1 of Update Procedure: Processing the Update}

\paragraph{Processing an edge insertion.} Let $e=(u,v)$ be an edge that is inserted into $G$ by the adversary. We check whether $u$ and $v$ are both currently not matched in $M^*$; if this is the case, then we insert $e$ into $M^*$, and otherwise $M^*$ remains unchanged. 
Next, we process the vertex $u$ as follows. If $v\in B$ and $u$ is currently unmatched, then we insert $u$ into $F(v)$. Additionally, if 
 $u\not\in B$, but the number of edges inserted into $G$ by the adversary reaches $t+1$, then we add $u$ to $B$, and  initialize the list $F(u)$: for every vertex $y\in N_G(u)$, if $y$ is unmatched by $M^*$, we insert $y$ into $F(u)$. Since $\deg_G(u)\leq t+1$ at this time, this step takes $\tilde O(t)$ time. 
 We process $v$ similarly.
 Overall, the total time for processing an edge insertion is bounded by $\tilde O\left (t\right )$.

\paragraph{Processing an edge deletion.} Let $e=(u,v)$ be an edge deleted by the adversary from $G$. 
If $u\in B$ and $v\in F(u)$, then we delete $v$ from $F(u)$. Similarly, if $v\in B$ and $u\in F(v)$, then we delete $u$ from $F(v)$.
If $e\in M^*$, then we delete $e$ from $M^*$, and then try to rematch $u$ and $v$ by processing every vertex $x\in \{u,v\}$, as follows. 

Assume first that $x$ is a good vertex, so $\deg_G(x)\leq t$. Then we scan the set $N_G(x)$ of vertices in time $O(t)$, until we identify a vertex $y\in N_G(x)$ that is currently not matched (if such a vertex exists). If we find such a vertex $y$, then we insert the edge $(x,y)$ into $M^*$. Assume now that $x$ is a bad vertex. If $F(x)\neq \emptyset$, then we select the first vertex $y\in F(x)$, and insert the edge $(x,y)$ into $M^*$, in time $O(1)$. 

This finishes the description of Part 1 of the update procedure. It is easy to verify that its running time is $\tilde O(t)$, and that, if $M^*$ was a maximal matching before the update, it remains so after Part 1 of the udpate procedure terminates.

\subsubsection*{Part 2 of Update Procedure: Updating the Data Structure $\set{F(b)\mid b\in B}$}

Observe that there are at most $O(1)$ vertices $x$, whose status may have changed from matched to unmatched or the other way around in $M^*$. We process each such vertex $x$ as follows. 
Assume first that the status of $x$ changed from matched to unmatched.
Then we inspect every vertex $y\in B$, and, if $y\in N_G(x)$, we insert $x$ into $F(y)$. Since $|B|\leq O(r/t)$, this can be done in time $\tilde O(r/t)$.
Assume now that the status of $x$ changed from unmatched to matched. Then we inspect every vertex $y\in B$, and, if $y\in N_G(x)$, we delete $x$ from $F(y)$. As before, this can be done in time $\tilde O(r/t)$.

Overall, from the above discussion, processing each update by the adversary requires $\tilde O\left (t+\frac r t\right )$ time. Since the total number of updates is bounded by $r$, the total running time of the algorithm is bounded by $\tilde O\left(r\cdot \left(t+\frac{r}{t}\right)\right )$.

\subsubsection{Implementing a Type-1 Phase: Proof of \Cref{thm: type 1 phase}}
\label{subsubsec: type 1 phase}

We assume that we are given as input an $n$-vertex and $m$-edge graph $G$ with $m\leq n^{3/2}$, that undergoes an online sequence of $r\leq n$ edge insertions and deletions. We provide an algorithm for maintaining a maximal matching in $G$ with total running time $O\left(n^{3/2+o(1)}\right )$.

We use the parameter $z=\ceil{\sqrt{n}}$. We invoke the algorithm from \Cref{thm: level-1 constr of subgraph system} to construct, in time $ \tilde O(n+m)\leq \tilde O(n^{3/2})$, a 1-level $z$-subgraph system $\sset$ for $G$. Then we apply the algorithm from \Cref{thm: matching from subgraph system} to the graph $G$, and the 1-level $z$-subgraph system $\sset$ for $G$ that we have computed. As the graph $G$ undergoes an online sequence of $r\leq n$ updates, the algorithm from \Cref{thm: matching from subgraph system} is guaranteed to maintain a maximal matching in $G$, and its total running time is $ O\left(n^{1+o(1)}\cdot \left(z+\frac{n}{z}\right)\right)\le O\left(n^{3/2+o(1)}\right)$. Overall, the total running time of the algorithm is $ O\left(n^{3/2+o(1)}\right )$, as required.

\subsubsection{Implementing a Type-2 Phase: Proof of \Cref{thm: type 2 phase}}
\label{subsubsec: type 2 phase}
We start with a high-level intuitive overview of the proof.
Let $d=\frac{2m}{n}$ denote the average vertex degree in $G$; since $m>n^{3/2}$, $d\ge 2\sqrt n$ must hold.
For this dense regime, the algorithm similar to that used for type-1 phases is no longer efficient: constructing a $1$-level subgraph system from scratch in every phase requires $\Omega(m)$ time, that needs to amortized over a long enough sequence of updates. That would force us to set the length of every phase to contain $r\gg n$ updates, and to set the parameter $z$ to a value that is sufficiently high, leading to a high amortized update time in the algorithm for maintaining a maximal matching (from \Cref{thm: matching from subgraph system}). Instead, we use a multi-level hierarchical construction, that spreads the rebuild work across a hierarchy of increasingly finer
phases. 

We first fix a base scale $d\leq z_1< 2d$ that is a power of $2$, and define a sequence of parameters
$z_i=\frac{z_1}{2^{i-1}}$ for $i>1$. We also choose a coarse time parameter $\eta\approx \sqrt n$, that is a power of $2$.
The update sequence is partitioned into a hierarchy of nested phases so that a level-$i$ phase spans $z_i\cdot\eta$ adversarial updates, and is partitioned into two level-$(i+1)$ phases of equal length.
In particular, a level-$1$ phase spans $z_1\cdot \eta$ updates by the adversary. Since $z_1=\Theta(d)$, the number of level-$1$ phases is bounded by $\frac{m}{z_1\eta}\leq O\left(\frac{n}{\eta}\right )=O(\eta)$. For all $i>1$, the number of level-$i$ phases is $\Theta(2^i\cdot \eta)$. The last level of the hierarchy $k$ is chosen so that $z_k=\tilde \Theta(\eta)$, and so the number of level-$k$ phases is $\tilde O\left(\frac{m}{z_k\cdot \eta}\right) \leq \tilde O\left (\frac{m}{n}\right)\leq \tilde O(d)$.

Consider now some level $i> 1$.
At the beginning of every level-$i$ phase $\Phi^i_j$, the level-$i$ data structure must produce:
(i) an \emph{auxiliary graph} $G^i_j$ together with a small set $E^i_j$ of \emph{deferred deletions}, such that, if we denote by $\tau^i_j$ the start time of phase $\Phi^i_j$, then
$G^i_j\setminus E^i_j=G\attime[\tau^i_j]$; and
(ii) an $i$-level $z_i$-subgraph system $\sset^i_j$ for the graph $G^i_j$. 
Let $\Phi^{i-1}_{j'}$ be the level-$(i-1)$ phase that contains $\Phi^i_j$.
In order to produce the graph $G^i_j$ and the subgraph system $\sset^i_j$, we  apply the algorithm from \Cref{thm: inductive constr of subgraph system} to the graph $G^{i-1}_{j'}$ and the $(i-1)$-level subgraph system $\sset^{i-1}_{j'}$ that the level-$(i-1)$ data structure computed at the beginning of Phase $\Phi^{i-1}_{j'}$.
Intuitively, edge set $E^i_j$ contains some of the edges that the adversary deleted from the graph $G$, but that were not yet incorporated into the subgraph system $\sset^i_j$. Deferring such edge deletions for later allows us to avoid recomputing the subgraph system after every edge deletion.

Our analysis relies on maintaining two invariants. First, we ensure that, for every level $i\geq 1$, at the beginning of every level-$i$ phase $\Phi^i_j$,  the deferred set $E^i_j$ of edge deletions satisfies $|E^i_j|\le (i -1)\cdot z_i\cdot \eta$. Second, since $z_i=z_1/2^{\,i-1}$ and there are $O\left (2^{i}\cdot \eta\right )$ level-$i$
phases, the total rebuild cost across all phases at level $i$ is
$ O(n^{1+o(1)}\cdot z_{i-1}\cdot 2^{i}\cdot \eta)\leq O(z_1\cdot n^{1+o(1)}\cdot \eta)\leq O(n^{3/2+o(1)}\cdot z_1)$; summing over all levels $i$ adds only logarithmic factors.
Since $z_1=\Theta(d)$, the total time spent on rebuild computations across all levels is bounded by $O(n^{3/2+o(1)}\cdot d)\leq O(m\cdot n^{1/2+o(1)})$.

At the last level $k$ (where $z_k=\tilde\Theta(\sqrt n)$), each level-$k$ phase $\Phi^k_j$ has length $z_k\cdot \eta$. 
We employ the algorithm from \Cref{thm: matching from subgraph system} in order to maintain a maximal matching in $G$ over the course of the phase. We provide the algorithm from \Cref{thm: matching from subgraph system} with the graph $G^k_j$ and the $k$-level subgraph system $\sset^k_j$ that the level-$k$ data structure computed for it, as input. We also supply the algorithm with a sequence of updates, that starts with the set $E_j^k$ of at most $(k - 1)\cdot z_k\cdot \eta$ deferred deletions, and then follows the update sequence of the adversary for Phase $\Phi^k_j$. Thus the resulting update sequence on the
auxiliary graph $G^k_j$ has length $|E^k_j|+z_k\cdot \eta\le k\cdot z_k\cdot \eta\le n$. The running time of the algorithm from
\Cref{thm: matching from subgraph system} for this update sequence is bounded by
$ O\left(n^{1+o(1)}\cdot \left(z_k+\frac{n}{z_k}\right)\right)\leq O(n^{3/2+o(1)})$ per phase. Since the total number of level-$k$ phases is bounded by $\tilde O\left(d\right )$ as observed above, the total time required to maintain the maximal matching across all level-$k$ phases is bounded by $ O(m\cdot n^{1/2+o(1)})$.
Overall, combining this with the  $ O(m\cdot n^{1/2+o(1)})$ time bound for the rebuild computations across all levels yiels the claimed total
$ O\left (m\cdot n^{1/2+o(1)}\right )$ running time bound for the entire algorithm. 
This concludes the high level overview of the proof. We now provide a formal proof.

We assume that we are given as input an $n$-vertex and $m$-edge graph $G$ with $m>n^{3/2}$.
We denote by $d=\frac{2m}{n}$ the average vertex degree in $G$, so $d\geq 2\sqrt n$. Recall that $G$ undergoes an online sequence of at most $ m$ edge insertions and deletions. Our goal is to design a deterministic algorithm that maintains a maximal matching in $G$ and has total running time $O\left( m\cdot n^{1/2+o(1)} \right)$.

\paragraph{Parameter setting.}

We let $z_1$ be the unique integral power of $2$ with $d\leq z_1<2d$, so in particular $z_1\geq 2\sqrt{n}$, and,  since $n$ is an integral power of $2$, $z_1\leq n$. For $i\geq 2$, we set $z_i=\frac{z_1}{2^{i-1}}$, and we let $k$ to be the largest integer for which $z_k\geq \frac{\sqrt{n}}{4\log n}$. Clearly, $k\leq \log n$ and $z_k\leq \frac{\sqrt{n}}{2\log n}$.
We also let $\eta$ be the unique integral power of $2$ with $\sqrt{n}\leq \eta<2\sqrt{n}$.

Recall that $G$ undergoes an online sequence of $r\leq m$ updates. 
Since $z_1 \cdot \eta^2 \ge d\cdot(\sqrt{n})^2 \geq \frac{nd}{2} = m$, we get that $r \le z_1 \cdot \eta^2$.
It will be convenient for us to assume that the number of updates is equal to $r=z_1\cdot \eta^2$, so, if needed, we can just add some dummy updates at the end of the update sequence, that repeatedly insert and delete the same edge. Note that the number of edges in $G$ remains $O(m)$ throughout the algorithm.

\paragraph{Hierarchical partition of the timeline into phases.}
We now define a hierarchical partition of the timeline into phases. There are $q_1=\eta$ level-1 phases, $\Phi_1^1,\ldots,\Phi_{q_1}^{1}$, each of which spans exactly $z_1\cdot \eta$ updates by the adversary.

For $2\le i\le k$, we partition the timeline into $q_i=2^{i-1}\eta$ level-$i$ phases,
each of which spans exactly $\frac{z_1\cdot \eta^2}{2^{i-1}\cdot \eta}=z_i\cdot \eta$ updates.
Note that for all $1\leq i<k$, every level-$i$ phase is partitioned into exactly $2$ level-$(i+1)$ phases.
 
For all $1\leq i\leq k$ and $1\leq j\leq q_i$, we denote by $\tau^i_j$ the start time of Phase $\Phi^i_j$.

For all $1\leq i \leq k$, we now define a level-$i$ data structure, whose purpose is the following:

\begin{properties}{R}
\item\label{requirement from level i}	For all $1\leq j\leq q_i$, at time $\tau^i_j$, the level-$i$ data structure must produce an $n$-vertex graph $G^i_j$ in the adjacency-list representation, together with a set $E^i_j\subseteq E(G^i_j)$ of at most $(i-1)\cdot z_i\cdot \eta$ edges, such that $G^i_j\setminus E^i_j= G\attime[\tau^i_j]$. Additionally, it must produce, at time $\tau^i_j$, an $i$-level $z_i$-subgraph system $\sset^i_j$ for graph $G^i_j$.
\end{properties}

Intuitively, set $E_j^i$ contains some of the edges that the adversary deleted from $G$ over the course of the algorithm, but that we temporarily keep in the graph $G^i_j$ so as not to destroy the $(i-1)$-level $z_{i-1}$-subgraph system that was constructed by the level $i-1$ data structure. 

\paragraph{Level-$1$ Data Structure.}
Recall that there are $q_1=\eta$ level-$1$ phases, each of which spans exactly $z_1\cdot \eta$ updates by the adversary. For all $1\leq j\leq q_1$, at the beginning of the $j$th level-$1$ phase $\Phi^1_j$, we apply the algorithm from \Cref{thm: level-1 constr of subgraph system} to the current graph $G$ with parameter $z=z_1$, to obtain a 1-level $z_1$-subgraph system $\sset^1_j$ for the current graph $G=G\attime[\tau^1_j]$, in time $ \tilde O(n+m)$. We then set $G^1_j=G$ and $E^1_j=\emptyset$. Since the number of level-$1$ phases is $q_1=\eta = O(n^{1/2})$, the total running time of the level-$1$ data structure is bounded by $\tilde O\left(m\cdot n^{1/2}\right )$.

\paragraph{Level-$i$ Data Structure for $i\geq 2$.}
Consider now an integer $2\leq i\leq k$, and some level-$i$ phase $\Phi^i_j$. Let $\Phi^{i-1}_{j'}$ be the level-$(i-1)$ phase that contains $\Phi^i_j$. Recall that, at the beginning of Phase $\Phi^{i-1}_{j'}$, the level-$(i-1)$ data structure produced a graph $G^{i-1}_{j'}$, and a subset $E^{i-1}_{j'}\subseteq E(G^{i-1}_{j'})$ of at most $(i-2)\cdot z_{i-1}\cdot \eta$ edges, that we denote by $E_D$, such that $G^{i-1}_{j'}\setminus E_D=G\attime[\tau^{i-1}_{j'}]$. Additionally, it produced an $(i-1)$-level $z_{i-1}$-subgraph system $\sset^{i-1}_{j'}$ for graph $G^{i-1}_{j'}$. 
Since a level-$(i-1)$ phase spans $z_{i-1}\cdot \eta$ updates by the adversary, the total number of edges that were deleted from $G$ since Phase $\Phi^{i-1}_{j'}$ started, and until time $\tau^i_j$, is bounded by $z_{i-1}\cdot \eta$. 
We insert all these edges into set $E_D$ (but we exclude the edges that did not lie in $G\attime[\tau^{i-1}_j]$, so they were first inserted and then deleted by the adversary); note that $|E_D|\leq (i-1)\cdot z_{i-1}\cdot \eta$ now holds. 
 We let $E_I$ be the set of all edges that have been inserted into $G$ since the beginning of Phase $\Phi^{i-1}_{j'}$, and until time $\tau^i_j$ by the adversary and that remain in $G\attime[\tau^i_j]$; notice that $|E_I|\leq z_{i-1}\cdot \eta$ must hold, and that $(G^{i-1}_{j'}\setminus E_D)\cup E_I$ is now identical to the graph $G$ at time $\tau^{i}_j$.
We apply the algorithm from \Cref{thm: inductive constr of subgraph system} to graph $G^{i-1}_{j'}$, the sets $E_D$ and $E_I$ of edges, the $(i-1)$-level $z_{i-1}$-subgraph system $\sset^{i-1}_{j'}$ for $G^{i-1}_{j'}$, parameters $h=i-1,z=z_{i-1}$ and $z'=z_i$. Recall that both $z_{i-1}$ and $z_i$ are integral powers of $2$, and that $z_{i-1}= 2z_i$. Therefore, we obtain a valid input to \Cref{thm: inductive constr of subgraph system}. Recall that the algorithm from \Cref{thm: inductive constr of subgraph system} returns a subset $E_D'\subseteq E_D$ of edges, with
\[
|E_D'|\leq \frac{|E_D|\cdot z_i}{z_{i-1}}\leq \frac{(i-1)\cdot z_{i-1}\cdot \eta\cdot z_i}{z_{i-1}} = (i-1)\cdot z_i \cdot \eta.
\]
It also returns an $i$-level $z_i$-subgraph system $\sset'$ for the graph $G'=(G^{i-1}_{j'}\cup E_I)\setminus (E_D\setminus E_D')$. 
Recall that $(G^{i-1}_{j'}\cup E_I)\setminus E_D=G\attime[\tau^i_j]$. Therefore, $G'\setminus E_D'=G\attime[\tau^i_j]$. We set $G^i_j=G'$,  $E^i_j=E_D'$, and $\sset^i_j=\sset'$. Note that $(G^i_j,E^i_j,\sset^i_j)$ satisfy Requirement
\ref{requirement from level i}.

The running time of the algorithm from \Cref{thm: inductive constr of subgraph system} is $ O(n^{1+o(1)}\cdot z_{i-1}+|E_D|+|E_I|)\leq  O(n^{1+o(1)}\cdot z_{i-1})$, since $|E_D|\leq (i-1)\cdot z_{i-1}\cdot \eta\leq n\cdot z_{i-1}$ and $|E_I|\leq z_{i-1}\cdot \eta\leq n\cdot z_{i-1}$. Recall that the algorithm may modify the adjacency-list representation of the graph $G^{i-1}_{j'}$ and the $(i-1)$-level subgraph system $\sset^{i-1}_{j'}$ (which may, in turn, then be modified by algorithms from next levels). Just before we apply the algorithm from \Cref{thm: inductive constr of subgraph system}  at the beginning of the next level-$i$ phase, we undo all these changes, to restore the adjacency-list representation of graph $G^{i-1}_{j'}$, and the original $(i-1)$-level subgraph system $\sset^{i-1}_{j'}$.

Since the number of level-$i$ phases is $q_i=O(2^{i-1}\cdot \eta)$, we get that the total running time of the data structure from level $i$ is $ O\left(n^{1+o(1)}\cdot z_{i-1}\cdot 2^{i-1}\cdot \eta\right)\leq  O\left(z_1\cdot n^{1+o(1)}\cdot\eta\right )$ since $z_{i-1}=\frac{z_1}{2^{i-2}}$.

The total running time of our algorithms for all levels $1,\ldots,k$ is then bounded by $O\left(k\cdot z_1\cdot n^{1+o(1)}\cdot\eta\right ) \leq  O(m\cdot n^{1/2+o(1)})$, since $k=O(\log n),z_1=O(d)=O(m/n)$, and $\eta=O(n^{1/2})$.

Lastly, recall that every level-$k$ phase $\Phi^k_j$ spans $z_k\cdot \eta$ updates by the adversary. At the beginning of Phase $\Phi^k_j$, the level-$k$ data structure computes an $n$-vertex graph $G^k_j$, together with a set $E^k_j\subseteq E(G^k_j)$ of at most $(k-1)\cdot z_k\cdot \eta$ edges, such that $G^k_j\setminus E^k_j= G\attime[\tau^k_j]$. It also computes a $k$-level $z_k$-subgraph system $\sset^k_j$ for graph $G^k_j$. We set $z=z_k$, and apply the algorithm from \Cref{thm: matching from subgraph system} to graph $G^k_j$ and the $k$-level $z_k$-subgraph system $\sset^k_j$ for $G^k_j$. The adversarial update sequence to graph $G^k_j$ starts by first deleting the edges of $E^k_j$ from it, and then follows the $z_k\cdot \eta$ updates that graph $G$ undergoes during Phase $\Phi^k_j$. Therefore, the total length of the update sequence to $G^k_j$ is bounded by $|E^k_j|+z_k\cdot \eta\leq k\cdot z_k\cdot \eta \leq n$, since $k\leq \log n,z_k\leq \frac{\sqrt{n}}{2\log n}$ and $\eta\leq 2\sqrt{n}$, and the algorithm from \Cref{thm: matching from subgraph system} can process all these updates.

Note that, from the description of the update sequence, at all times $\tau$ during Phase $\Phi^k_j$, the dynamic graph that the algorithm from \Cref{thm: matching from subgraph system} maintains (obtained from $G^k_j$ by deleting the edges of $E^k_j$ from it and then applying the first $\tau$ updates that $G$ undergoes during Phase $\Phi^k_j$) is exactly $G\attime[\tau^k_j+\tau]$. The algorithm from \Cref{thm: matching from subgraph system} maintains a maximal matching $M'$ in the resulting graph, which must also be a maximal matching in $G$.

The running time of the algorithm from \Cref{thm: matching from subgraph system} is $ O\left(n^{1+o(1)}\cdot \left(z_k+\frac{n}{z_k}\right)\right)$.
Since the number of level-$k$ phases is $q_k=2^{k-1}\cdot \eta$, we get that the total running time of the algorithm from \Cref{thm: matching from subgraph system}  over all level-$k$ phases is:

\[\begin{split}
 O\left(n^{1+o(1)}\cdot \left(z_k+\frac{n}{z_k}\right)\cdot 2^{k-1}\cdot \eta\right)&\leq  O\left(\left(n^{1+o(1)}+\frac{n^{2+o(1)}}{z_k^2}\right)\cdot z_1\cdot \eta\right ) \\& 
 \leq O\left(n^{1+o(1)}\cdot z_1\cdot \eta\right )\\&\leq
  O\left (m\cdot n^{1/2+o(1)}\right )
  \end{split}
\]

(in the first inequality, we used the fact that $z_k=\frac{z_1}{2^{k-1}}$; in the second inequality we used the fact that $z_k\geq\frac{\sqrt{n}}{4\log n}$; and in the last inequality we used the facts that $z_1=O(d)=O(m/n)$ and $\eta=O(n^{1/2})$. We conclude that the total running time of the entire algorithm is $ O\left(m^{1+o(1)}\cdot n^{1/2}\right )$, as required.

\section{From Multi-Level Subgraph System to Fully Dynamic Maximal Matching: Proof of \Cref{thm: matching from subgraph system}}
\label{sec: subgraph-system-to-matching}

In this section we  prove \Cref{thm: matching from subgraph system}. 
The proof of the theorem is very similar to the proof of \Cref{thm: matching from basic subgraph system} and follows the same high-level structure. The main difference, in addition to a slight change in the parameters, is a more involved algorithm for rematching vertices of $A$, and the extension of the data structures so that they can handle edge insertions in addition to edge deletions.

We assume that we are given as input an initial $n$-vertex graph $G$, together with
 integral parameters $0<k\leq \log n$ and $1\leq z\leq n$ and a $k$-level $z$-subgraph system  $\sset=\left(S,B,U,M,\set{A_i,N_i,R_i}_{i=1}^k,\Lambda,L\right )$ for $G$. Our goal is to maintain a maximal matching in $G$, as it undergoes an online sequence of $r\leq n$ edge deletions and insertions.

Throughout the algorithm, the vertex sets $S,B,U$ and $\set{A_i,N_i,R_i}_{i=1}^k$ remain unchanged. The edge set $M$ is decremental: whenever the adversary deletes an edge $e\in M$ from $G$, this edge is deleted from $M$ as well, but no new edges may be inserted into $M$.
The lists in $\Lambda$ and $L$ may be updated by our algorithm, as described below.

\subsection{Data Structures and Invariants}

We maintain a graph $G'$, that can be thought of as the decremental version of $G$. In other words, at the beginning of the algorithm, $G'=G$, and, whenever an edge is deleted from $G$, we delete it from $G'$ as well, but we ignore edge insertions into $G$. We use the input adjacency-list representation of the initial graph $G\attime[0]$ in order to maintain the adjacency-list representation of $G'$, initializing $G'=G$ at the beginning of the algorithm.

We also maintain the set $E_I\subseteq E(G)$ of all edges inserted by the adversary into $G$, that currently lie in $G$. Whenever the adversary inserts an edge $e$ into $G$, we add it to set $E_I$. When the adversary deletes an edge $e$ from $G$, we first check if $e\in E_I$, and, if so, we remove it from $E_I$. Otherwise, $e\in E(G')$ must hold, and we delete $e$ from $G'$.
We note that we do not maintain the graph $G$ explicitly, and instead we only maintain $G'$ (using the input adjacency-list representation of $G\attime[0]$), and the edge set $E_I$; clearly, at all times, $G=G'\cup E_I$ holds.

Recall that, as part of the subgraph system $\sset$, we are given, for every vertex $u\in U$, a vertex list $\Lambda(u)=N_G(u)\cap (B\cup U)=N_{G'}(u)\cap (B\cup U)$, whose length is $O(z)$. Over the course of the algorithm, we will update the list $\Lambda(u)$ to ensure that $\Lambda(u)=N_{G'}(u)\cap (B\cup U)$ always holds. Specifically, whenever the adversary deletes an edge $(u,v)$ with $v\in \Lambda(u)$ from $G$, we delete $v$ from $\Lambda(u)$.
Similarly, as part of the subgraph system, we are given, for all $1\leq i\leq k$, for every vertex $v\in A_i$, the list $L(v)=N_G(v)\cap R_i=N_{G'}(v)\cap R_i$. 
Over the course of the algorithm, we will update the list $L(v)$ to ensure that $L(v)=N_{G'}(v)\cap R_i$ always holds. Specifically, whenever the adversary deletes an edge $(u,v)$ from $G$ with $u\in L(v)$, we delete $u$ from $L(v)$.

Let $\hat G$ be the subgraph of $G$ induced by the edges of $M$. Note that the maximum vertex degree in $\hat G$ is bounded by $z$. We start by computing an edge-coloring of $\hat G$ with $z+1$ colors using the algorithm from \Cref{thm: edge coloring}, in time $O\left (|E(\hat G)|\cdot n^{o(1)}\right )= O\left (n^{1+o(1)}\cdot z\right )$.
This coloring naturally defines a partition $(M_1,\ldots,M_{z+1})$ of $M$, where for all $1\leq i\leq z+1$, $M_i$ is a matching (consisting of all edges of $\hat G$ that were assigned the color $i$). For all $1\leq i\leq z+1$, we compute the number $m_i$ of vertices of $S$ that are unmatched by $M_i$.  Over the course of the algorithm, when an edge of $M_i$ is deleted by the adversary, we delete it from $M_i$ and update $m_i$ accordingly. We use the following simple claim.

\begin{claim}\label{claim:good-color}
At the beginning of the algorithm, there is an index $i^\star\in[z{+}1]$, such that the number of vertices of $S$ that are not matched by $M_{i^\star}$ is bounded by $\frac{|S|\cdot (\log n+1)}{z+1}\le \frac{2n\log n}{z}$.
\end{claim}
\begin{proof}
	Recall that, from Property \ref{prop: edges inc to S} of the $k$-level subgraph system,  every vertex $v\in S$ is incident to at least $z-k+1\geq z-\log n$ edges in $M$. Therefore, there are at most $\log n+1$ indices $1\leq i\leq z+1$, for which $v$ is not matched in $M_i$, and so $\sum_{i=1}^{z+1} m_i \leq |S|\cdot (\log n+1)$. We conclude that there must be an index $1\leq i^*\leq z+1$ with $m_{i^*}\le \frac{|S|\cdot (\log n+1)}{z+1}\leq \frac{2|S|\cdot \log n}{z}\leq \frac{2n\log n}{z}$. 
\end{proof}

By Claim \ref{claim:good-color}, there is an index $1\leq i^*\leq z+1$ such that at most $\frac{2n\log n}{z}$ vertices of $S$ are not matched in $M_{i^*}$; we reindex the matchings so that $i^*=1$ holds.
We partition the algorithm into \emph{phases},  where each phase spans exactly $\lfloor \frac{n}{z} \rfloor$ updates by the adversary. Since the total number of updates is bounded by $r\leq n$, we get that the number of phases is bounded by $2z$. Throughout the algorithm, we will maintain a maximal matching $M^*$ with $M_1\subseteq M^*$,  
while matchings $M_2,\dots,M_{z+1}$ are used in order to ``repair" $M_1$ when necessary. The algorithm may modify the matching $M_1$, but it always ensures that  $M_1\subseteq M$ holds. We ensure that, at the beginning of every phase, the following invariant holds:

\begin{properties}{R}
	\item \label{prop: matching M1} After the initialization of every phase, all but at most $\frac{12n\log^2 n}{z}$ vertices of $S$ are matched by the edges in $M_1$.
\end{properties}

Over the course of each phase, we maintain a dynamic matching $M^*$ with $M_1\subseteq M^*$, that is a maximal matching in the current graph $G$. 
Throughout the algorithm, we say that a vertex $x\in V(G)$ is \emph{settled}, if it is settled with respect to $M^*$ (see Definition \ref{def: settled}).
Our algorithm also maintains two additional data structures: {\bf directed} graphs $H$ and $\tilde H$, that we define next. 

\paragraph{Graph $H$.} Graph $H$ is a directed graph that is used in order to quickly rematch vertices of $B\cup U$. The vertex set of $H$ is $V(H)=B\cup U$. For every vertex $u\in U$ that is currently unmatched by $M^*$, there is a directed edge in $H$ from $u$ to every vertex in $N_{G'}(u)\cap (B\cup U)$ (recall that all vertices in $N_{G'}(u)\cap (B\cup U)$ lie in $\Lambda(u)$, and that $|\Lambda(u)|\leq O(z)$). Equivalently, for every vertex $v\in B\cup U$, the set of the in-neighbors of $v$ in $H$ is exactly the set of all unmatched vertices $u\in U$ with $(u,v)\in E(G')$.

\paragraph{Good and Bad Vertices, and Graph $\tilde H$.}
The second data structure is a \emph{directed} graph $\tilde H$, that is designed to take care of the edges of $E_I$ -- the edges of $G$ that were inserted by the adversary since the beginning of the algorithm; recall that $|E_I|\leq r\leq n$ must hold. Before defining the graph $\tilde H$, we need to introduce the notion of good and bad vertices.
Throughout the algorithm, we say that a vertex $v\in V(G)$ is \emph{bad} if the number of edges inserted into $G$ by the adversary over the course of the algorithm that are incident to $v$ is at least $z$, and it is \emph{good} otherwise. Note that the same edge may be inserted and deleted numerous times, but once a vertex $v$ becomes bad it remains so until the end of the algorithm. Clearly, the
total number of bad vertices is always bounded by $\frac{2r}{z} \le \frac{2n}{z}$. The vertex set of graph $\tilde H$ is $V(G)$. For every vertex $v\in V(G)$ that is currently unmatched by $M^*$, and every edge $(v,u)\in E_I$ where $u$ is bad, there is a directed edge in $\tilde H$ from $v$ to $u$. Equivalently, for every bad vertex $v\in V(G)$, the set of all in-neighbors of $v$ in $\tilde H$ is exactly the set of all unmatched vertices $u\in V(G)$ with $(u,v)\in E_I$.

\paragraph{Phase execution: high-level overview.}
Over the course of a phase, if an edge $e\in M_1$ is deleted by the adversary, it is deleted from $M_1$ as well; recall that at most $\frac{n}{z}$ such edge deletions may occur over the course of a phase. Additionally, our algorithm may choose to delete up to $\frac{2n\log n}{z}$ edges from $M_1$ over the course of the phase. Each edge deletion from $M_1$ may cause at most two vertices of $S$ to become unmatched in $M_1$. By Property \ref{prop: matching M1}, after the initialization of a phase, at most $\frac{12n\log^2 n}{z}$ vertices of $S$ are unmatched by $M_1$. Therefore, we ensure that the following invariant holds throughout the phase:

\begin{properties}[1]{R}
	\item \label{prop during subphase: matching M1} At all times, the number of vertices of $S$ that are not matched by $M_1$ is bounded by $\frac{18 n\log^2 n}{z}$.
\end{properties}

Recall that, from Property \ref{prop: neighbors of Ai in M in Ni} of the multi-level subgraph system (see Definition \ref{def: subgraph structure}), for all $1\leq i\leq k$, for every vertex $v\in A_i$, for every edge $e=(v,u)\in M$ incident to $v$, $u\not\in R_i$ must hold. Moreover,
from Property \ref{prop: containment of Rs}, $U= R_k\subseteq R_{k-1} \subseteq \dots \subseteq R_1$ must hold. Therefore, for all $1\leq i\leq k$, for every vertex $v\in A^{\leq i}$, for every edge $e=(v,u)\in M$ incident to $v$, $u\not\in R_i$ must hold. Since $M_1\subseteq M$, we conclude that $M_1$ may not match a vertex $v\in A^{\leq i}$ to a vertex of $R_i$. Therefore, the only way that a vertex $v\in A^{\leq i}$ is matched to a vertex of $R_i$ in $M^*$ is if $v$ is not matched by $M_1$. Combining this with Invariant \ref{prop during subphase: matching M1}, we get that the following invariant must also hold throughout the phase:

\begin{properties}[2]{R}
	\item \label{prop during subphase: few vertices of S'' matched to U} At all times, for all $1\leq i\leq k$, at most $\frac{18 n\log^2 n}{z}$ vertices of $A^{\leq i}$ may be matched to vertices of $R_i$ in $M^*$.
\end{properties}

Our algorithm will also explicitly maintain the set $\hat S\subseteq S$ of all vertices $v\in S$ that are currently unmatched by $M^*$. We say that a vertex $v\in V(G)$ \emph{switches its status in $M^*$} whenever it switches from being matched to unmatched in $M^*$, or the other way around.

\subsection{Update Procedures}

Next, we describe the main procedures that our algorithm uses in order to maintain all data structures correctly.
We start with Procedure
\procupdate, that is called whenever a vertex of $G$ switches its status in $M^*$.
The purpose of the procedure is to ensure that all data structures are consistent with this change of status.

\subsection*{Procedure \procupdate}

The input to Procedure \procupdate is a vertex $v\in V(G)$. The procedure checks if $v$ is currently matched in $M^*$, and updates the vertex set $\hat S$, and the graphs $H$ and $\tilde H$ accordingly. We provide the description of \procupdate in \Cref{fig:procupdate}.

\program{Procedure \procupdate}{fig:procupdate}
{
	{\bf Input:} a vertex $v\in V(G)$.
	
	\begin{itemize}
		\item If $v$ is matched in $M^*$:
            \begin{itemize}
                \item if $v\in \hat S$, delete $v$ from $\hat S$.
                \item if $v\in U$, delete all edges of $H$ leaving $v$ from $H$.
                \item delete all edges of $\tilde H$ leaving $v$ from $\tilde H$.
            \end{itemize}
        \item Otherwise ($v$ is unmatched in $M^*$):
            \begin{itemize}
                \item if $v\in S$, insert $v$ into $\hat S$.
                \item if $v\in U$, for every vertex $v'\in \Lambda(v)$, insert the edge $(v,v')$ into $H$.
                \item for every bad vertex $v'$, if $(v,v')\in E_I$, insert  the edge $(v,v')$ into $\tilde H$.
            \end{itemize}
	\end{itemize}
}

Procedure $\procupdate(v)$ is called whenever a vertex $v\in V(G)$ changes status in $M^*$, and it ensures that the following properties hold:

\begin{properties}{Q}
    \item \label{prop: set S} At all times, $\hat S$ is exactly the set of all vertices of $S$ that are unmatched by $M^*$.
	\item \label{prop: graph H} At all times, for every vertex $v\in B\cup U$, the set of all in-neighbors of $v$ in $H$ is the set of all vertices of  $N_{G'}(v)\cap U$ that are currently unmatched by $M^*$.
    \item \label{prop: graph tilde H} At all times, for every bad vertex $v\in V(G)$, the set of all in-neighbors of $v$ in $\tilde H$ is the set of all vertices of $\set{v'\in V(G)\mid (v,v')\in E_I}$, that are currently unmatched by $M^*$.
\end{properties}

Since, for every vertex $v\in U$, $|\Lambda(v)|=O(z)$ holds, and since there are at most $O(\frac{n}{z})$ bad vertices, the following observation follows immediately from the description of the procedure.
\begin{observation} \label{obs: procupdate time}
The running time of $\procupdate(v)$ is $\tilde O\left(z+\frac{n}{z}\right)$.
\end{observation}

Next, we describe procedure $\procrematchBU$, whose purpose is to attempt to rematch a vertex of $B\cup U$ that became unmatched in $M^*$.

\subsection*{Procedure \procrematchBU}

The input to Procedure $\procrematchBU$ is a vertex $u\in B\cup U$ that is currently not matched by $M^*$. The procedure attempts to rematch $u$ to one of its neighbors in $G$. We provide the description of \procrematchBU in \Cref{fig:procrematchBU}.

\program{Procedure \procrematchBU}{fig:procrematchBU}
{
	{\bf Input:} a vertex $u\in B\cup U$ that is  not matched by $M^*$.

    \begin{enumerate}
        \item \label{step: H} If $H$ contains any incoming edge incident to $u$, let $(v,u)$ be any such edge:
            \begin{itemize}
                \item Insert $(u,v)$ into $M^*$, and call \procupdate for $u$ and $v$.
                \item Return SUCCESS.
            \end{itemize}
        \item  \label{step: check S} For every vertex $v\in \hat S$, if $(u,v)\in E(G')$:
            \begin{itemize}
                \item Insert $(u,v)$ into $M^*$, and call \procupdate for $u$ and $v$.
                \item Return SUCCESS.
            \end{itemize}
        \item \label{step: good} If $u$ is a good vertex:
            \begin{itemize}
                \item For every edge $(u,v)\in E_I$, if $v$ is unmatched by $M^*$:
                \begin{itemize}
                    \item Insert $(u,v)$ into $M^*$, and call \procupdate for $u$ and $v$.
                    \item Return SUCCESS.
                \end{itemize}
            \end{itemize}
        \item \label{step: bad} Otherwise, $u$ is a bad vertex:
            \begin{itemize}
                \item  If $\tilde H$ contains any incoming edge incident to $u$, let $(v,u)$ be any such edge:
                \begin{itemize}
                    \item Insert $(u,v)$ into $M^*$, and call \procupdate for $u$ and $v$.
                    \item Return SUCCESS.
                \end{itemize}
            \end{itemize}
        \item Return FAILURE.
    \end{enumerate}

}

The following observation summarizes the properties of the procedure. 

\begin{observation} \label{obs: procrematchBU}
The running time of Procedure $\procrematchBU(u)$ is $\tilde O\left(z+\frac{n}{z}\right)$. If the procedure returns FAILURE, then every vertex of $N_G(u)$ is currently matched by $M^*$.
\end{observation}
\begin{proof}
	We first analyze the running time of the procedure, excluding the time spent on calls to Procedure \procupdate.
	Observe that Step \ref{step: H} can be implemented in time $O(1)$, since it only needs to check whether $u$ has an incoming edge in $H$. Step \ref{step: check S} can be implemented in time $O(|\hat S|)\leq \tilde O(n/z)$ (from  Property~\ref{prop during subphase: matching M1}), since it only requires scanning the vertices of $\hat S$ and checking, for each such vertex, whether it is a neighbor of $u$ in $G'$. If $u$ is a good vertex, then at most $z$ edges of $E_I$ are incident to $u$, so Step \ref{step: good} takes $O(z)$ time. If $u$ is a bad vertex, then Step \ref{step: bad} only needs to check whether $\tilde H$ contains an edge entering $u$, which can be done in $O(1)$ time. Additionally, the algorithm may call Procedure \procupdate at most twice; the running time for each such call is bounded by $\tilde O\left(z+\frac{n}{z}\right)$ from Observation \ref{obs: procupdate time}. Altogether, the running time of Procedure $\procrematchBU$ is bounded by $\tilde O\left(z+\frac{n}{z}\right)$.

Assume now that the procedure returns FAILURE, and assume for contradiction that there exists a vertex $v\in N_G(u)$ that is currently unmatched by $M^*$. Assume first that $v\in N_{G'}(u)$. If $v\in U$, then by Property \ref{prop: graph H}, edge $(v,u)$ must lie in $H$, so $u$ should have been matched in Step \ref{step: H}. Otherwise, $v\in S$, and, by Property \ref{prop: set S}, $v\in \hat S$ must hold. Then $u$ should have been matched in Step \ref{step: check S}. Lastly, assume that $v\not\in N_{G'}(u)$, so $(v,u)\in E_I$ must hold. If $u$ is a good vertex, then it should have been matched in Step \ref{step: good}. Otherwise,  $u$ is a bad vertex, and then by Property \ref{prop: graph tilde H}, edge $(v,u)$ must lie in $\tilde H$. Therefore, $u$  should have been matched in Step  \ref{step: bad}.
Since $V(G)=S\cup U$ and $E(G)=E(G')\cup E_I$, we reach a contradiction.
\end{proof}

We obtain the following immediate corollary of Observation \ref{obs: procrematchBU}, that follows from the fact that Procedure \procrematchBU does not delete any edges from $M^*$.

\begin{corollary}\label{cor: settled after rematchBU}
	After Procedure $\procrematchBU(u)$ is executed, vertex $u$ becomes settled. Moreover, for every vertex $v\in V(G)$, if $v$ was settled before the call to the procedure, it remains so after the execution of the procedure.
\end{corollary}

Finally, we describe procedure \procrematchA, whose purpose is to attempt to rematch a vertex of $A$ that became unmatched in $M^*$.

\subsection*{Procedure \procrematchA}

The input to Procedure $\procrematchA$ is a vertex $v\in A_i$ for some $1\le i\leq k$ that is currently unmatched by $M^*$. The procedure first attempts to match $v$ with some vertex $u\in R_i$. If successful, it may delete a single edge $e=(u,u')$ from $M^*$ (and possibly from $M_1$), but it guarantees that $u'\in A^{\geq i+1}\cup B\cup U$ holds; it then attempts to rematch $u'$ by either calling Procedure $\procrematchA$ for $u'$ recursively (if $u'\in A$), or Procedure $\procrematchBU(u')$ (otherwise). If this attempt to match $v$ fails, the procedure then tries to match $v$ with its other neighbors in $G'$ that lie in $\hat S$, and finally it tries to match it via the edges of $E_I$. We provide the description of \procrematchA in Figure \ref{fig:procrematchA}.

\program{Procedure \procrematchA}{fig:procrematchA}
{
	{\bf Input:} an index $1\leq i\leq k$ and a vertex $v\in  A_i$ that is currently not matched by $M^*$.
	
	\begin{enumerate}
		\item \label{step: checking Lv}Process the first $\frac{18n\log^2 n}{z}+1$ vertices of $L(v)$. 
		 When a vertex $u\in L(v)$ is processed, if $u$ is currently not matched by $M^*$, or if it is matched by $M^*$ to a vertex of $V(G)\setminus A^{\leq i}$:
		\begin{itemize}
			\item If some edge $e=(u,u')$ incident to $u$ lies in $M^*$:  (note that $u'\in A^{\geq i+1}\cup B\cup U$ must hold)
			\begin{itemize}
				\item Delete $e$ from $M^*$. If $e\in M_1$, delete $e$ from $M_1$.
                \item Insert $(u,v)$ into $M^*$, and call \procupdate for $u'$ and $v$.
                \item If $u'\in A^{\geq i+1}$, then call $\procrematchA(i',u')$, where $i'>i$ is the index with $u'\in A_{i'}$.
                \item Otherwise ($u'\in B\cup U$ must hold), call $\procrematchBU(u')$.
			\end{itemize}
    		\item Otherwise (if $u$ is not currently matched by $M^*$): Insert $(u,v)$ into $M^*$, and call \procupdate for $u$ and $v$.
			\item Return SUCCESS.
		\end{itemize}
    \item \label{step: checking S for A} For every $v'\in \hat S$, if $(v,v')\in E(G')$:
            \begin{itemize}
                \item Insert $(v,v')$ into $M^*$, and call \procupdate for $v$ and $v'$.
                \item Return SUCCESS.
            \end{itemize}
    \item \label{step: checking good vertex for A} If $v$ is a good vertex:
        \begin{itemize}
            \item For every edge $(v,v')\in E_I$, if $v'$ is unmatched by $M^*$:
            \begin{itemize}
                \item Insert $(v,v')$ into $M^*$, and call \procupdate for $v$ and $v'$.
                \item Return SUCCESS.
            \end{itemize}
        \end{itemize}
    \item \label{step: checking bad vertex for A} Otherwise, $v$ is a bad vertex:
        \begin{itemize}
            \item  If $\tilde H$ contains any incoming edge incident to $v$, let $(v',v)$ be any such edge:
            \begin{itemize}
                \item Insert $(v,v')$ into $M^*$, and call \procupdate for $v$ and $v'$.
                \item Return SUCCESS.
            \end{itemize}
        \end{itemize}
	\item Return FAILURE
	\end{enumerate}
}

The following observation summarizes the properties of Procedure \procrematchA.

\begin{observation} \label{obs: procrematchA}
	There is a constant $c\geq 1$, such that the running time of $\procrematchA(i,v)$, including recursive calls to itself and calls to Procedures \procrematchBU and \procupdate, is bounded by $(k-i+1)\cdot \left(z+\frac{n}{z}\right)\cdot (\log n)^c\leq \tilde O\left(z+\frac{n}{z}\right)$. The procedure may call to itself recursively at most once. Moreover, if $(i_1,v_1),\ldots,(i_q,v_q)$ are the pairs to which the recursive calls of $\procrematchA$ are applied in this order, then $i<i_1<i_2<\dots < i_q\leq k$ must hold, and the algorithm may delete at most $q+1\leq k$ edges from $M^*$ and from $M_1$ (where every edge that is deleted from $M_1$ is also deleted from $M^*$). Morever, if $D$ denotes the set of all vertices of $G$ that serve as endpoints of all edges that are deleted from $M^*$ by the procedure (including all its recursive calls), then, at the end of the procedure, for every vertex $x\in D\cup \set{v}$, either $x$ is matched by $M^*$, or every vertex in $N_G(x)$ is matched by $M^*$; in other words, $x$ is settled.
\end{observation}
\begin{proof}
	We first analyze the running time of Procedure $\procrematchA(i,v)$ excluding recursive calls to itself, and calls to Procedures $\procrematchBU$ and $\procupdate$.
	Step \ref{step: checking Lv} of Procedure $\procrematchA(i,v)$ only needs to process the first $O\left(\frac{n\log^2 n}{z}\right )$ vertices of $L(v)$, and the time required for processing each such vertex (excluding the time for recursive calls to $\procrematchA$ and  calls to Procedures $\procrematchBU$ and $\procupdate$) is $O(1)$. Therefore, the running time of this step is $O\left(\frac{n\log^2 n}{z}\right )$.
		Step \ref{step: checking S for A} can be implemented in time $O(|\hat S|)\leq \tilde O\left (\frac n z\right )$ (from  Property~\ref{prop during subphase: matching M1}), since it only requires scanning the vertices of $\hat S$ and checking, for each such vertex, whether it is a neighbor of $v$ in $G'$. If $v$ is a good vertex, then at most $z$ edges of $E_I$ are incident to $v$, so Step \ref{step: checking good vertex for A} takes $O(z)$ time. If $v$ is a bad vertex, then Step \ref{step: bad} only needs to check whether $\tilde H$ contains an edge entering $v$, which can be done in $O(1)$ time. 
	Overall, the running time of Procedure $\procrematchA(i,v)$ excluding recursive calls to itself, and calls to Procedures $\procrematchBU$ and $\procupdate$, is bounded by $\tilde O\left(z+\frac{n}{z}\right)$. Note that the procedure may call Procedure \procupdate at most twice; the running time for each such call is bounded by $\tilde O\left(z+\frac{n}{z}\right)$ from Observation \ref{obs: procupdate time}. 
	Additionally, it may call Procedure \procrematchBU, whose running time is bounded by $\tilde O\left(z+\frac{n}{z}\right)$ from Observation \ref{obs: procrematchBU} at most once. Altogether, we get that there is a constant $c>1$, such that the running time of Procedure $\procrematchA(i,v)$, excluding the recursive calls to itself, is bounded by $\left(z+\frac{n}{z}\right)\cdot (\log n)^c$. Since the procedure may call to itself recursively at most once, and the depth of the recursion is bounded by $k-i+1\leq k\leq \log n$, the total running time of the procedure is bounded by $(k-i+1)\cdot \left(z+\frac{n}{z}\right)\cdot (\log n)^c\leq \tilde O\left(z+\frac{n}{z}\right)$.

Observe that Procedure $\procrematchA(i,v)$ may only delete an edge from $M^*$ or from $M_1$ in step \ref{step: checking Lv}, and, if an edge is deleted from $M_1$, it is also deleted from $M^*$. Therefore, excluding the recursive calls to itself,  Procedure $\procrematchA(i,v)$ may delete at most one edge from $M^*$, and at most one edge from $M_1$. Additionally, the procedure may only  call to itself recursively in Step \ref{step: checking Lv}
 and, if the recursive call is applied to a pair $(i',u')$, then $i'>i$ must hold. It is now immediate to verify that, if $(i_1,v_1),\ldots,(i_q,v_q)$ are the pairs to which the recursive calls of $\procrematchA$ are applied in this order, then $i<i_1<i_2<\dots<i_q\leq k$ must hold, and the algorithm may delete at most $q+1\leq k$ edges from $M^*$ and from $M_1$; every edge that is deleted from $M_1$ is also deleted from $M^*$. Lastly, we denote by
$D(i,v)$ the set of all vertices of $G$ that serve as endpoints of all edges that are deleted from $M^*$ by Procedure $\procrematchA(i,v)$ (including all its recursive calls), and we denote by $D'(i,v)=D(i,v)\cup \set{v}$. It is now enough to prove that, at the end of the procedure, for every vertex $x\in D'(i,v)$, either $x$ is matched by $M^*$, or every vertex in $N_G(x)$ is matched by $M^*$; in other words, $x$ is settled.
The proof is by induction on the depth of the recursion. 

\paragraph{Induction Base.} For the base case, we assume that Procedure $\procrematchA(i,v)$ does not perform any recursive calls to itself. 
Assume first that it deleted an edge $(u,u')$ from $M^*$. Then it must be the case that $u'\in B\cup U$, so after the call to $\procrematchBU(u')$, this vertex becomes settled. The procedure then inserted the edge $(v,u)$ into $M^*$, so both $v$ and $u$ become settled. Moreover, it is easy to verify that  $D'(i,v)=\set{v,u,u'}$ must hold in this case, so every vertex in $D'(i,v)$ is settled.

Assume now that the procedure does not delete any edges from $M^*$. Then, $D'(i,v)=\set{v}$, and it is enough to prove that, once the procedure terminates, either $v$ is matched by $M^*$, or every vertex in $N_G(v)$ is matched by $M^*$. Assume for contradiction that this is not the case; in other words, at the end of the procedure, $v$ is not matched by $M^*$, and there is a vertex $u\in N_{G}(v)$ that is not matched by $M^*$. 

Assume first that $u\in N_{G'}(v)$. Since we assumed that $v$ is not matched at the end of the procedure, no edges were inserted into $M^*$ by the procedure. Note that it is impossible that $u\in S$, since then from Invariant \ref{prop: set S}, $u\in \hat S$ must hold, and so $v$ must have been matched in Step \ref{step: checking S for A}. Therefore, $u\in U$ must hold. Since, from Property \ref{prop: containment of Rs} of the subgraph system, $U\subseteq R_i$, and since $L(v)=N_{G'}(v)\cap R_i$ holds throughout the algorithm, $u\in L(v)$ must hold. Since vertex $v$ was not matched in Step \ref{step: checking Lv}, we get that $u$ is not among the first $\frac{18n\log^2 n}{z}+1$ vertices of $L(v)$; moreover, every one of the first  $\frac{18n\log^2 n}{z}+1$ vertices of $L(v)$ must be matched to a vertex of $A^{\leq i}$ by $M^*$. Since $L(v)\subseteq R_i$, this contradicts Invariant \ref{prop during subphase: few vertices of S'' matched to U}.

Finally, assume that $u\not\in N_{G'}(v)$, so $(u,v)\in E_I$ must hold. If $v$ is a good vertex, then it should have been matched in Step \ref{step: checking good vertex for A} of the procedure. Otherwise,  $v$ is a bad vertex, and then by Property \ref{prop: graph tilde H}, edge $(u,v)$ must lie in $\tilde H$. Therefore, $v$  should have been matched in Step  \ref{step: checking bad vertex for A} of the procedure, a contradiction.

\paragraph{Induction Step.} Assume now that Procedure $\procrematchA(i,v)$ calls itself recursively, for a pair $(i',u')$, where $i'>i$. Note that the recursive call may only occur at Step \ref{step: checking Lv} of the procedure, after an edge that is incident to $v$ is inserted into $M^*$. Note that $D'(i,v)=D'(i',u')\cup \set{v}$. By the induction hypothesis, at the end of the recursive call to Procedure  $\procrematchA(i',u')$, for every vertex $x\in D'(i',u')$, either $x$ is matched by $M^*$, or every vertex in $N_G(x)$ is matched by $M^*$. Since vertex $v$ was matched by $M^*$ prior to the call to Procedure  $\procrematchA(i',u')$, it is immediate to verify that, at the end of Procedure $\procrematchA(i,v)$, for every vertex $x'\in D'(i',u')$, either $x'$ is matched by $M^*$, or every vertex in $N_G(x')$ is matched by $M^*$.
\end{proof}

We obtain the following easy corollary of Observation \ref{obs: procrematchA}.

\begin{corollary}\label{cor: settled after rematchA}
	After Procedure $\procrematchA(i,v)$ is executed, vertex $v$ becomes settled. Moreover, for every vertex $u\in V(G)$, if $u$ was settled before the call to the procedure, then it remains so after the execution of the procedure.
\end{corollary}
\begin{proof}
	It is immediate to verify from the statement of Observation \ref{obs: procrematchA} that, at the end of the execution of Procedure $\procrematchA(i,v)$, vertex $v$ becomes settled. Consider now some vertex $u\in V(G)$ that was settled before the call to the procedure. Assume first that $u$ was matched by $M^*$ before the call to the procedure. If no edge incident to $u$ was deleted from $M^*$, then $u$ remains matched by $M^*$. Otherwise, $u\in D$ must hold, and,  from Observation \ref{obs: procrematchA}, $u$ must be settled at the end of the procedure.
	
	Lastly, assume that $u$ was not matched by $M^*$ before the call to the procedure, but, for every vertex $y\in N_G(u)$, $y$ was matched by $M^*$. From our discussion above, each such vertex $y$ must remain settled at the end of the procedure. But then, from 
	Observation \ref{obs: settled if neighbors settled}, $u$ is also settled at the end of the procedure.
\end{proof}

We are now ready to complete the description of our algorithm. We start by describing the initialization procedure, which is slightly different for the first phase and for the remaining phases. Then we describe the algorithms for handling edge insertions and deletions.

\subsection{Initialization Algorithm for the First Phase}

We initialize the maximal matching $M^*$ and the data structures $H$ and $\tilde H$ as follows. 
Recall that,  at the beginning of the algorithm,
we have computed matchings $M_1,\ldots,M_{z+1}$, and we ensured that the number of vertices of $S$ that are not matched in $M_1$ is at most $\frac{2n\log n}{z}$.
We start by setting $M^* = M_1$, and by computing the set $\hat S\subseteq S$ of vertices that are not matched in $M^*$, so that $|\hat S|\leq  \frac{2n\log n}{z}$ and Property \ref{prop: set S} holds.
We initialize the graph $\tilde H$ to contain all vertices of $G$ and no edges (recall that $E_I=\emptyset$ at time $0$, so all vertices of $G$ are good and Property \ref{prop: graph tilde H} now holds).  We also initialize the graph $H$ to contain all vertices of $B\cup U$; for every vertex $u\in U$ that is currently unmatched by $M^*$, for every vertex $v\in \Lambda(u)$, we insert a directed edge $(u,v)$ into $H$. This ensures that Property \ref{prop: graph H} holds as well. Since, for every vertex $u\in U$, $|\Lambda(u)|\leq O(z)$, the running time of the initialization algorithm so far is $O(nz)$.

We extend $M^*$ to a maximal matching in the remainder of the initialization procedure, as follows.
Let $X$ contain all vertices of $G$ that are currently not matched in $M^*$. We consider every vertex $x\in X$ in turn. When $x\in X$ is considered, if $x$ is still not matched by $M^*$ at this time: if $x\in B\cup U$, we invoke Procedure $\procrematchBU(x)$, and otherwise, we invoke Procedure $\procrematchA(i,x)$, where $1\leq i\leq k$ is the index with $x\in A_i$. 
From Corollaries \ref{cor: settled after rematchBU} and \ref{cor: settled after rematchA}, once we finish processing $x$, it becomes settled, and it remains so until the end of the initialization procedure.
Therefore, once every vertex in $X$ is processed, matching $M^*$ is guaranteed to be maximal. From Observations \ref{obs: procrematchBU} and \ref{obs: procrematchA}, the running time of this step is $\tilde O\left(n\left(z+\frac{n}{z}\right)\right)$.

Recall that initially,  at most $\tfrac{2 n\log n}{z}$ vertices of $S$ were not matched by $M_1$, so $|X\cap S|\leq \tfrac{2 n\log n}{z}$ held. Therefore, Procedure \procrematchA was called  directly at most $\frac{2 n\log n}{z}$ times (excluding the recursive calls). From Observation \ref{obs: procrematchA}, each such direct call to \procrematchA may lead to the deletion of at most $k\leq \log n$ additional edges from $M_1$. Therefore, the total number of vertices of $S$ that are not matched by $M_1$ at the end of this step is bounded by $ \tfrac{2 n\log n}{z} + 2\cdot \frac{2n\log ^2n}{z}\leq \frac{6n\log^2 n}{z}$; in particular, Property~\ref{prop: matching M1} is guaranteed to hold.
This completes the description of the initialization algorithm for the first phase. Its total running time is $\tilde O\left(n \left(z+\frac{n}{z}\right)\right )$, and at the end of this procedure, Property \ref{prop: matching M1} holds.

\subsection{Initialization Algorithm for Subsequent Phases}

We now provide an algorithm for initializing a phase that is not the first phase. 
We let $\hat S'$ be the set of all vertices $v\in S$ that are currently not matched by $M_1$; note that this set may be different from the set $\hat S$ of all vertices that are currently not matched by $M^*$, and that $\hat S\subseteq \hat S'$ holds. Over the course of the initialization procedure, we may update the matching $M_1$, but the set $\hat S'$ of vertices remains fixed and contains all vertices of $S$ that were not matched by $M_1$ at the beginning of the procedure.
If, at the beginning of the phase, Property \ref{prop: matching M1}  holds, then no further initialization is needed. Therefore, we assume from now on that Property \ref{prop: matching M1} does not hold. Note that, from Property \ref{prop during subphase: matching M1}, all but at most $\frac{18n\log^2 n}{z}$ vertices of $S$ are currently matched by $M_1$. Therefore, we can assume from now on that $\frac{12n\log^2 n}{z}<|\hat S'|\leq \frac{18n\log^2 n}{z}$ holds.

Recall that at the beginning of the algorithm, we computed a collection $\set{M_1, \dots, M_{z+1}}$ of disjoint matchings, and, for all $1\leq i\le z+1$ we maintain the counter $m_i$ of the number of vertices of $S$ that are unmatched by $M_i$ throughout the algorithm. Over the course of the algorithm, matching $M_1$ may be modified by the algorithm, but for all $2\leq i\leq z+1$, the only modification that the matching $M_i$ undergoes is the deletion of edges that the adversary deletes from $G$. Let $\Pi$ be the collection of all pairs $(i,v)$, where $2\leq i\leq z+1$ and $v\in S$ is a vertex that is not matched by $M_i$. Clearly, $|\Pi|=\sum_{i=2}^{z+1}m_i$. At the beginning of the algorithm, from Property 
\ref{prop: edges inc to S} of the subgraph system, for every vertex $v\in S$, there are at most $k+1\leq \log n+1$ indices $2\leq i\leq z+1$ such that $v$ is not matched by $M_i$, so $|\Pi|\le (\log n+1)\cdot |S|\leq  (\log n+1)\cdot n$ holds at that time. Since there total number of adversarial edge deletions is bounded by $n$, and since each such edge deletion may affect at most two vertices of $S$, we get that, throughout the algorithm, $|\Pi|\le (\log n+1)\cdot n+2n = 4n\log n$ always holds.
Therefore, at the beginning of the current phase, there must exist an index $2\leq i\leq z+1$, with:

\[m_i\leq\frac{\sum_{j=2}^{z+1}m_j}{z}=\frac{|\Pi|}{z}\leq \frac{4n\log n}{z}.\]

We fix such an index $2\leq i\leq z$ from now on, so the number of vertices of $S$ that are not matched by $M_i$ is at most $\frac{4n\log n}{z}$.

We construct the graph $G^*=M_1\cup M_i$ in time $O(n)$. Since $M_1$ and $M_i$ are both matchings, every vertex in $G^*$ has degree at most $2$, and so every connected component of $G^*$ is either a simple path or a cycle (we view isolated vertices as paths of length $0$). 
Consider a vertex $v\in \hat S'$; since $v$ is not matched by $M_1$, its degree in $G^*$ is either $0$ or $1$. Therefore, the connected component of $G^*$ containing $v$ must be a path, that we denote by $P(v)$. Let $v'$ be the other endpoint of the path $P(v)$. We say that vertex $v$ is \emph{problematic} if either (i) $v'=v$; or (ii) $v'\in S$ and the last edge on $P(v)$ is in $M_1$. Notice that, if $v$ is problematic, then $v'$ is not matched by $M_i$. 
Since the number of vertices of $S$ that are not matched by $M_i$ is bounded by $\frac{4n\log n}{z}$, at most $\frac{4n\log n}{z}$ vertices of $\hat S'$ may be problematic. Next, we process every non-problematic vertex of $\hat S'$ one by one. Over the course of this step, we may delete some edges from $M^*$. We maintain the set $X\subseteq V(G)$ containing every vertex $v\in V(G)$ that was matched by $M^*$ at the beginning of the initialization procedure but now becomes unmatched. We initialize $X=\emptyset$.

\paragraph{Processing a non-problematic vertex of $\hat S'$.}
Let $v\in \hat S'$ be a vertex of $\hat S'$ that is not problematic. We augment the matching $M_1$ along the path $P(v)$ as follows. For every edge $e\in E(P)$, if $e\in M_1$ then we delete $e$ from $M_1$, and otherwise we insert it into $M_1$. It is easy to verify that after this augmentation, $M_1$ remains a valid matching, and vertex $v$ is now matched by $M_1$. The only vertex that was previously matched by $M_1$ and may become now unmatched is $v'$ (if the last edge on $P(v)$ lies in $M_1$), and in this case, $v'\not\in S$ holds. 
Notice that it is also possible that $v'\in \hat S'$ (in which case the last edge on $P(v)$ must lie in $M_i$). In this case, both $v$ and $v'$ become matched in $M_1$, and there is no need to process $v'$ separately. 
Next, we modify the matching $M^*$ in order to ensure that $M_1\subseteq M^*$ and that $M^*$ remains a valid matching; we also update the vertex set $X$. Specifically, for every edge $e\in E(P)$, if $e$ was deleted from $M_1$ then we delete it from $M^*$ as well, and if it was inserted into $M_1$ then we insert it into $M^*$ as well. Note that every inner vertex $u\in P(v)$ was initially matched in $M_1$ and  in $M^*$, and it remains so; only the endpoints $v,v'$ of $P(v)$ may have changed their status from matched to unmatched with respect to $M_1$, or the other way around. 
For every endpoint $x$ of $P(v)$, if $x$ was matched by $M_1$ but now becomes unmatched (notice that $x=v'$ must hold in this case, and $x$ is now unmatched by $M^*$ as well), we add $x$ to $X$.
For every endpoint $x$ of $P(v)$, if $x$ was unmatched by $M_1$ but now becomes matched by it, if $M^*$ now contains two edges incident to $x$, we delete the edge that is incident to $x$ and does not lie in $M_1$ from $M^*$, and we add the other endpoint of this edge into $X$.

Once very non-problematic vertex of $\hat S'$ is processed, the only vertices of $S$ that may remain unmatched by $M_1$ are the problematic vertices of $\hat S'$, and their number is bounded by $\frac{4n\log n}{z}$. Vertex set $X$ now contains every verex $v$ that was matched by $M^*$ at the beginning of the procedure but is currently unmatched by it.
It is immediate to verify that every problematic vertex $v\in\hat S'$ may contribute at most $O(1)$ vertices to set $X$, so $|X|\leq O\left (\frac{n\log n}{z}\right )$ must hold.
Lastly, we invoke Procedure $\procupdate(u)$ for every vertex $u$ whose status in $M^*$ changed between matched and unmatched. Since, for every problematic vertex $v\in \hat S'$, at most $O(1)$ vertices  change status while processing $P(v)$, from Observation \ref{obs: procupdate time}, the running time of this step is bounded by $\tilde O\left(|\hat S'|\cdot \left(z+\frac{n}{z}\right)\right)=\tilde O\left(\frac{n}{z}\cdot \left(z+\frac{n}{z}\right)\right)$. Overall, the running time of the initialization procedure so far is bounded by $O(n)+\tilde O\left(\frac{n}{z}\cdot \left(z+\frac{n}{z}\right)\right)\leq  \tilde O\left(\frac{n}{z}\cdot \left(z+\frac{n}{z}\right)\right)$.
From our discussion, all but at most $\frac{4n\log n}{z}$ vertices of $S$ are now matched by $M_1$. In the remainder of the initialization procedure, we extend $M^*$ to a maximal matching. This part of the algorithm is very similar to that of the initialization procedure for the first phase.

\paragraph{Extending $M^*$ to a maximal matching.} We consider every vertex $x\in X$ in turn. When $x\in X$ is considered, if $x$ is still not matched by $M^*$ at this time: if $x\in B\cup U$, we invoke Procedure $\procrematchBU(x)$, and otherwise, we invoke Procedure $\procrematchA(i,x)$, where $1\leq i\leq k$ is the index with $x\in A_i$. 
From Observations \ref{obs: procrematchBU} and \ref{obs: procrematchA}, the time required to process a single vertex of $X$ is 
 $\tilde O\left(z+\frac{n}{z}\right)$, and the total running time of this step is bounded by:

\[\tilde O\left(|X|\cdot \left(z+\frac{n}{z}\right)\right)\leq \tilde O\left(n+\frac{n^2}{z^2}\right).\]
In the following observation we prove that, once all vertices in $X$ are processed, $M^*$ becomes a maximal matching.

\begin{observation}\label{ob: maximal matching at the end}
	After all vertices in $X$ are processed, the matching $M^*$ becomes maximal.
	\end{observation}
\begin{proof}
	Let $E'$ be the set of all edges that lied in $M^*$ before the start of the initialization procedure, but do not lie in the final matching $M^*$, and let $V'$ be the set of all vertices that serve as endpoints of the edges in $E'$. From Observation \ref{obs: maximal matching transformation}, it is enough to show that, at the end of the initialization algorithm, every vertex in $V'$ is settled. Consider now a vertex $v\in V'$. If $v\in X$ then,  from Corollaries  \ref{cor: settled after rematchBU} and \ref{cor: settled after rematchA}, when $v$ was processed it became settled, and it remains so until the end of the algorithm. 
	If an edge incident to $v$ lies in $M^*$ at the end of the algorithm, then $v$ is settled.
	Otherwise, the only way that $v$ may lie in $V'$ is if the unique edge of $M^*$ that was incident to $v$ after the algorithm finished processing the non-problematic vertices of $\hat S'$ was deleted by Procedure $\procrematchA$. However, from Observation \ref{obs: procrematchA}, once that procedure terminated, $v$ became settled, and, from  Corollaries  \ref{cor: settled after rematchBU} and \ref{cor: settled after rematchA}, it remains so until the end of the algorithm.
\end{proof}

Recall that, after the first step of the initialization procedure,  at most $\tfrac{4 n\log n}{z}$ vertices of $S$ were not matched by $M_1$, so $|X\cap S|\leq \tfrac{4 n\log n}{z}$ held. Therefore, Procedure \procrematchA was called  directly at most $\frac{4 n\log n}{z}$ times (excluding the recursive calls). From Observation \ref{obs: procrematchA}, each such direct call to \procrematchA may lead to the deletion of at most $k\leq \log n$ additional edges from $M_1$. Therefore, the total number of vertices of $S$ that are not matched by $M_1$ at the end of this step is bounded by $ \tfrac{4 n\log n}{z} + 2\cdot \frac{4n\log ^2n}{z}\leq \frac{12n\log^2 n}{z}$; in particular, Property~\ref{prop: matching M1} is guaranteed to hold.

This completes the description of the initialization algorithm for a phase. Its total running time is $\tilde O\left(n+\frac{n^2}{z^2}\right)$, and, at the end of this procedure, Property \ref{prop: matching M1} holds.

Recall that the number of phases in the algorithm is bounded by $O(z)$. Therefore, the total time that the algorithm spends on initialization of all phases is bounded by: 

\[
\tilde O\left(z\cdot \left(n+\frac{n^2}{z^2}\right )\right ) = \tilde O\left(n\cdot \left(z+\frac{n}{z}\right)\right)
\]

\subsection{Processing an Edge Deletion}

We now provide an algorithm for processing a single edge deletion by the adversary.

Let $e=(u,v)$ be an edge that the adversary deleted. If $e\in E_I$, we delete it from $E_I$. We also delete the edges $(u,v)$ and $(v,u)$ from the graph $\tilde H$ if they lie there.
Otherwise, $e\in E(G')$ must hold and we delete $e$ from $G'$.
If $u\in U$ and $v\in \Lambda(u)$, then we delete $v$ from $\Lambda(u)$. If $u\in S\setminus B$ and $v\in L(u)$, then we delete $v$ from $L(u)$. Similarly, if $u\in \Lambda(v)$ or $u\in L(v)$, we delete $u$ from the corresponding lists.
We also delete the edges $(u,v)$ and $(v,u)$ from graph $H$ if they lie there.
If $e$ lies in some matching $M_i$, for $2\leq i\leq z+1$, then we delete $e$ from $M_i$ and we update the counter $m_i$ as needed.

If $e\not\in M^*$, then no further updates are needed. Therefore, we assume that $e\in M^*$ from now on.  We delete $e$ from $M^*$, and call \procupdate for $u$ and for $v$. By Observation \ref{obs: procupdate time}, the running time of the procedure is $\tilde O\left(z+\frac{n}{z}\right)$. If $e$ lies in $M_1$, it is also deleted from $M_1$. For every vertex $x\in \set{u,v}$, if $x\in B\cup U$, then we invoke Procedure $\procrematchBU(x)$, and otherwise we invoke Procedure $\procrematchA(x)$. From Observations \ref{obs: procrematchBU} and \ref{obs: procrematchA}, the time we spend on processing each of the vertices $u$, $v$ is $\tilde O\left(z+\frac{n}{z}\right)$, and so the total running time of the entire update algorithm is $\tilde O\left(z+\frac{n}{z}\right)$.

We claim that the resulting matching $M^*$ is maximal. In order to prove this, it is enough to show that every vertex in $G$ is settled with respect to this final matching. Indeed, consider any vertex $a\in G$. If $a\in \set{u,v}$, then, from Corollaries \ref{cor: settled after rematchBU} and \ref{cor: settled after rematchA}, once we finish processing $a$, it becomes settled, and it remains so until the end of the algorithm. Assume now that $a\not\in \set{u,v}$. Assume first that the original matching $M^*$ contained an edge $e'$ incident to $a$. Then either $e'$ lies in the final matching $M^*$, or it was deleted by Procedure $\procrematchA$. In the latter case, from Observation \ref{obs: procrematchA}, once that procedure terminated, $a$ became settled, and, from  Corollaries  \ref{cor: settled after rematchBU} and \ref{cor: settled after rematchA}, it remains so until the end of the algorithm. Therefore, if the initial matching $M^*$ contained an edge incident to $a$, then $a$ is settled with respect to the final matching $M^*$. Lastly, assume that the original matching $M^*$ did not contain an edge incident to $a$. Since $M^*$ was a maximal matching, every vertex in $N_G(a)$ was matched by $M^*$. From the above discussion, every vertex in $N_G(a)$ must be settled in the final matching $M^*$. From Observation \ref{obs: settled if neighbors settled}, it then follows that $a$ is settled in the final matching $M^*$.

Note that if $e\in M_1$, then $e$ is deleted from $M_1$. Moreover, the update procedure calls Procedure \procrematchA at most twice, and each such call may lead to the deletion of at most $k\leq \log n$ additional edges from $M_1$, from Observation  \ref{obs: procrematchA}. Recall that by Property \ref{prop: matching M1}, after the initialization procedure of the current phase, at most $\frac{12n\log^2 n}{z}$ vertices of $S$ are  not matched by $M_1$. Since the total number of edge deletions by the adversary in a phase is bounded by $\frac{n}{z}$, and since each such edge deletion may lead to up to $2\log n+1\le 3\log n$ deletions of edges from $M_1$, we then get that, throughout the phase, the total number of vertices of $S$ that are not matched by $M_1$ remains bounded by:

\[\frac{12n\log^2 n}{z}+2\cdot \frac{n}{z}\cdot 3\log n\leq \frac{18n\log^2 n}{z}.\] 

Therefore, Properties \ref{prop during subphase: matching M1} and \ref{prop during subphase: few vertices of S'' matched to U} hold throughout the phase.

\subsection{Processing an Edge Insertion}

We now describe an algorithm for processing the insertion of an edge $e=(u,v)$ into the graph $G$. 

We start by inserting the edge $e$ into the set $E_I$ of inserted edges. If $u$ is unmatched by $M^*$ and $v$ is a bad vertex, we insert $(u,v)$ into the graph $\tilde H$. Similarly, if $v$ is unmatched by $M^*$ and $u$ is a bad vertex, we insert $(v,u)$ into the graph $\tilde H$.

Next, we check whether the vertices $u$ and $v$ are currently matched by $M^*$. If both $u$ and $v$ are unmatched by $M^*$, we insert $(u,v)$ into $M^*$, and call \procupdate for $u$ and $v$. By Observation \ref{obs: procupdate time}, the running time of the update procedure so far is $\tilde O\left(z+\frac{n}{z}\right)$.

Finally, for every vertex $x\in \{u,v\}$, if $x$ is a good vertex, and the number of edges of $E_I$ incident to $x$ reaches $z+1$, then $x$ becomes a bad vertex. In this case, for every edge $(x,y)\in E_I$, if $y$ is currently unmatched by $M^*$, we insert the edge $(y,x)$ into the graph $\tilde H$. This step takes $\tilde O(z)$ time.

\subsection{Bounding the Running Time}
From our discussion, the algorithm spends initially $O\left ( n^{1+o(1)}\cdot z\right )$ time in order to compute matchings $M_1,\ldots,M_{z+1}$. Additional time required to initialize all data structures for the first phase is bounded by $\tilde O\left(n\left(z+\frac{n}{z}\right)\right)$. The total time spent on the initialization of all subsequent phases is bounded by  $\tilde O\left(n\left(z+\frac{n}{z}\right)\right)$. The time required to process every edge insertion and edge deletion is bounded by $\tilde O\left(z+\frac{n}{z}\right)$. Since the total number of such updates is bounded by $n$, the  total running time of the algorithm is
$\tilde O\left(n^{1+o(1)}\left(z+\frac{n}{z}\right)\right)$.

\section{From $h$-Level to $(h+1)$-Level Subgraph System: Proof of \Cref{thm: inductive constr of subgraph system}}
\label{sec: recursive subgraph system}

In this section we prove \Cref{thm: inductive constr of subgraph system}. 
Recall that we are given as input an $n$-vertex graph $G$ in the adjacency-list representation. We are also given as input an integral parameter $h\geq 1$ and a parameter $1\leq z\leq n$ that is an integral power of $2$. Lastly, we are given an $h$-level $z$-subgraph system $\sset =\left(S,B,U,M,\set{A_i,N_i,R_i}_{i=1}^h,\Lambda,L\right )$ for the graph $G$, a subset $E_D\subseteq E(G)$ of edges, and another set $E_I\subseteq V(G)\times V(G)$ of edges that do not lie in $G$.

Our algorithm is also given as input an integral parameter $1\leq z'<z$ that is an integral power of $2$. Our goal is to construct a subset $E_D'\subseteq E_D$ of at most $\frac{|E_D|\cdot z'}{z}$ edges, and an $(h+1)$-level $z'$-subgraph system  for the graph $G'=(G\cup E_I)\setminus (E_D\setminus E_D')$. 

Our algorithm consists of two steps. In the first step, construct the set $E'_D\subseteq E_D$ of edges, and an initial $(h+1)$-level $z'$-subgraph system $\hat \sset$  for the corresponding graph  $G'=(G\cup E_I)\setminus (E_D\setminus E_D')$ that does not quite have all required properties. Then in the second step we fix $\hat\sset$ to obtain the final proper $(h+1)$-level $z'$-subgraph system $\sset'$ for $G'$.

\subsection{Step 1: An Initial Subgraph System}

In this step we construct the set $E'_D\subseteq E_D$ of edges, and an initial $(h+1)$-level $z'$-subgraph system $\hat \sset= \left(\hat S,\hat B,\hat U,\hat M,\set{\hat A_i,\hat N_i,\hat R_i}_{i=1}^{h+1},\hat\Lambda,\hat L\right )$ for the corresponding graph  $G'=(G\cup E_I)\setminus (E_D\setminus E_D')$, that has all required properties except for, possibly,
\ref{prop: degrees in U} and \ref{prop: U neighbors in Ak}. We will then ``fix'' this subgraph system in Step 2.

We start by applying the algorithm from \Cref{thm: edge coloring}  to the subgraph of $G$ induced by the set $M$ of edges, to obtain a $(z+1)$-coloring of the edges in $M$; equivalently, we obtain a partition $(M_1,\ldots,M_{z+1})$ of $M$, where, for all $1\leq i\leq z+1$, $M_i$ is a matching. Recall that the running time of the algorithm from \Cref{thm: edge coloring}   is $O\left (|M|\cdot n^{o(1)}\right )\leq  O\left (n^{1+o(1)}\cdot z\right )$. We also store, for each edge $e\in M$, its color index $c(e)\in\{1,\ldots,z{+}1\}$ to allow an $O(1)$ time lookup.

Next, for all $1\leq i\leq z+1$, we compute the value $w_i=|E_D\cap M_i|$ by scanning the edge set $E_D$, and looking up the color $c(e)$ of each scanned edge $e$. We reindex the matchings $M_1,\ldots,M_{z+1}$ so that $w_1\leq w_2\leq\cdots\leq w_{z+1}$ holds.
We then construct the edge sets $\hat M=M_1\cup\dots\cup M_{z'}$ and $E_D'=E_D\cap \hat M$. Clearly, $|E_D'|\leq\frac{|E_D|\cdot z'}{z}$ must hold. We define the graph $G'=(G\cup E_I)\setminus (E_D\setminus E_D')$, and we modify the adjacency-list representation of $G$  in time $\tilde O(|E_I|+|E_D|)$, in order to turn it into the adjacency-list representation of $G'$.
Notice that the running time of the algorithm so far is bounded by $\tilde O\left (n^{1+o(1)}\cdot z+|E_D|+|E_I|\right )$.

We now proceed to construct the $(h+1)$-level $z'$-subgraph system  $\hat \sset= \left(\hat S,\hat B,\hat U,\hat M,\set{\hat A_i,\hat N_i,\hat R_i}_{i=1}^{h+1},\hat\Lambda,\hat L\right )$ for $G'$. 
We start by setting $\hat S=S$, $\hat U=U$, and, for all $1\leq i\leq h$, we set $ \hat A_i=A_i$. We now define the sets $\hat A_{h+1}$ and $\hat B$, by appropriately partitioning the vertex set $B$. 
In order to do so, we inspect every vertex $v\in B$ one by one. For each such vertex $v\in B$, we inspect all edges of $\hat M$ that are incident to $v$. If, for every such edge $(v,u)\in \hat M$, $u\in \hat S$ holds, then we add $v$ to $\hat A_{h+1}$, and otherwise we add it to $\hat B$. We set $\hat N_{h+1}=\hat B$ and $\hat R_{h+1}=\hat U$, as required. Note that Property \ref{prop: neighbors of Ai in M in Ni} is now guaranteed to hold for $\hat A_{h+1}$.
As before, for all $1\leq i\leq  h+1$, we denote by $\hat A^{\leq i}=\hat A_1\cup\cdots\cup \hat A_i$ and by $\hat A^{\geq i}=\hat A_i\cup\cdots\cup \hat A_{h+1}$.

Consider now an index $1\leq i\leq h$. Recall that, in the $h$-level subgraph system, $N_i\subseteq A^{\geq i+1}\cup B$ held, and that $R_i=\left(A^{\geq i+1}\cup B\cup U\right )\setminus N_i$. Since $\hat A^{\geq i+1}\cup \hat B= A^{\geq i+1}\cup B$ (as $\hat A_{h+1}\cup \hat B=B$), we simply set $\hat N_i=N_i$, which ensures that $\hat N_i\subseteq \hat A^{\geq i+1}\cup \hat B$, and that Property \ref{prop: neighbors of Ai in M in Ni} holds for $\hat A_i$. We then get that $\hat R_i=\left(\hat A^{\geq i+1}\cup \hat B\cup \hat U\right )\setminus \hat N_i=\left(A^{\geq i+1}\cup B\cup U\right )\setminus N_i=R_i$. Recall that the $h$-level subgraph system $\sset$ contained, for every vertex $v\in A_i$, the vertex list $L(v)=N_G(v)\cap R_i=N_G(v)\cap \hat R_i$. For every vertex $v\in A_i$, we initially set $\hat L(v)=L(v)$ (note that we cannot afford to copy these lists since they may be too long, so we just use a pointer to $L(v)$ in order to initialize $\hat L(v)$). We then delete from $\hat L(v)$ all vertices $u$ with $(v,u)\in E_D\setminus E_D'$, and insert into it all vertices $u'$ with $(v,u')\in E_I$ and $u'\in \hat R_i$. The time required to initialize and update the lists $\hat L(v)$ for all vertices $v\in \hat A^{\leq h}$ is $\tilde O(|E_D|+|E_I|+n)$

Recall that, from Property \ref{prop: containment of Rs} of the $h$-level subgraph system, $U= R_h\subseteq R_{h-1}\subseteq \cdots\subseteq R_1$ holds. Since $\hat U= U$, and since we have set, for all  
$1\leq i\leq h$, $\hat R_i=R_i$, and $\hat R_{h+1}=\hat U=U$, we get that $\hat U= \hat R_{h+1}= \hat R_h\subseteq\cdots\subseteq \hat R_1$ holds, establishing Property  \ref{prop: containment of Rs} for the new subgraph system $\hat \sset$.

Next, we construct, for every vertex $v\in \hat A_{h+1}$ the list $\hat L(v)=N_{G'}(v)\cap \hat R_{h+1}=N_{G'}(v)\cap \hat U$. We also construct, for every vertex $v\in \hat B$, the list $\hat L(v)=N_{G'}(v)\cap \hat U$ (even though this list is not required from the definition of the subgraph system, it will be useful for us in Step 2). Lastly, recall that, from the definition of the subgraph system, for every vertex $u\in \hat U$, we need to provide the vertex list $\hat \Lambda(u)=N_{G'}(u)\cap (\hat B\cup \hat U)$. Instead, we will construct the vertex list $\hat \Lambda(u)=N_{G'}(u)\cap (\hat A_{h+1}\cup \hat B\cup \hat U)$, since this will be more convenient for Step 2 of the algorithm; we will fix all these lists in Step 2.
 
We start by setting, for every vertex $v\in \hat B\cup \hat A_{h+1}$, $\hat L(v)=\emptyset$. 
Recall that the $h$-level subgraph system contained, for every vertex $u\in U$, a list $\Lambda(u)=N_G(u)\cap (B\cup U)=N_G(u)\cap (\hat A_{h+1}\cup \hat B\cup \hat U)$, whose length is at most $O(z)$. 
We initially set $\hat \Lambda(u)=\Lambda(u)$, and then inspect every vertex $v\in \hat \Lambda(u)$. If $v\in \hat A_{h+1}\cup \hat B$, then we add $u$ to the list $\hat L(v)$. At the end of this process, we are guaranteed that, for every vertex $v\in \hat A_{h+1}\cup \hat B$, $\hat L(v)=N_{G}(v)\cap \hat U$, and, for all $u\in \hat U$, $\hat \Lambda(u)=N_G(u)\cap \left( \hat A_{h+1}\cup  \hat B\cup \hat U\right )$. Since, for all $u\in U$, $|\Lambda(u)|\leq O(z)$, the running time of this process so far is $\tilde O(nz)$. Next, we inspect every edge of $E_D\setminus E_D'$. For each such edge $e$, if one of the endpoints of $e$ (say $u$) lies in $\hat U$, then we check whether the other endpoint lies in $\hat \Lambda(u)$, and, if so, delete it from $\hat \Lambda(u)$. Similarly, if one of the endpoints of $e$ (say $v$) lies in $\hat A_{h+1}\cup \hat B$, then we check whether the other endpoint lies in $\hat L(v)$, and, if so, we delete it from $\hat L(v)$. The other endpoint of $e$ is processed similarly. We also inspect every edge of $E_I$. 
For each such edge $e'$, if one of the endpoints of $e'$ (say $u$) lies in $\hat U$, and the other endpoint (say $u'$) lies in $\hat A_{h+1}\cup \hat B\cup \hat U$, then we insert $u'$ into $\hat \Lambda(u)$. Similarly, if one of the endpoints of $e'$ (say $v$) lies in $\hat A_{h+1}\cup \hat B$, and the other endpoint (say $v'$) lies in $\hat U$, then we insert $v'$ into $\hat L(v)$. We process both endpoints of $e'$ in this manner. At the end of this procedure we are guaranteed that, for every vertex $u\in \hat U$, $\hat \Lambda(u)=N_{G'}(u)\cap (\hat A_{h+1}\cup \hat B\cup \hat U)$, and, for every vertex $v\in \hat A_{h+1}\cup \hat B$, $\hat L(v)=N_{G'}(v)\cap \hat U$.
Clearly, the time required to construct the lists $\set{\hat L(v)\mid v\in \hat S}$ and $\set{\hat \Lambda(u)\mid u\in \hat U}$ is bounded by $\tilde O(nz+|E_D|+|E_I|)$.

This completes the construction of the initial $(h+1)$-level $z'$-subgraph system $\hat \sset$ for $G'$. 
We have already established that Properties  \ref{prop: neighbors of Ai in M in Ni} and \ref{prop: containment of Rs} hold for $\hat \sset$, and it is immediate to verify that  Properties \ref{prop: small degrees in M} and \ref{prop: no edge of M w both endpoints in U} hold as well.

For every vertex $v\in V(G)$, we denote by $m(v)$ the number of edges of $M$ incident to $v$, and by $\hat m(v)$ the number of edges of $\hat M$ incident to it.
Recall that, from Property \ref{prop: edges inc to S}, for every vertex $v\in S$,  $m(v)\geq z-h+1$ holds. Therefore, there are at most $h$ matchings $M_i$ among $M_1,\cdots,M_{z+1}$, such that $v$ is unmatched in $M_i$. Since $\hat M=M_1\cup\dots\cup M_{z'}$, we get that
$\hat m(v)\geq z'-h$ must hold, and, since $\hat S=S$, we get that Property \ref{prop: edges inc to S} holds in $\hat \sset$. We note however Properties \ref{prop: degrees in U} and \ref{prop: U neighbors in Ak} may not hold, both because the new subgraph system is defined with respect to a parameter $z'<z$, and due to the insertion of edges $E_I$ into $G'$. In addition, for every vertex $u\in \hat U$, the list $\hat \Lambda(u)$ currently contains all vertices of $N_{G'}(u)\cap (\hat A_{h+1}\cup \hat B\cup \hat U)$ (instead of $N_{G'}(u)\cap (\hat B\cup \hat U)$ -- this is on purpose and will be used in Step 2).

From our discussion, once the set $\hat M$ of edges is constructed, the time required to construct the initial $(h+1)$-level $z'$-subgraph system $\hat \sset$ by modifying the input $h$-level $z$-subgraph system $\sset$, and to transform the adjacency-list representation of $G$ into that of $G'$, is bounded by: 

$$\tilde O\left (n^{1+o(1)}\cdot z+|M|+|E_D|+|E_I|\right )\leq \tilde O\left (n^{1+o(1)}\cdot z+|E_D|+|E_I|\right ).$$ Altogether, the running time of the algorithm for Step 1 is $\tilde O\left (n^{1+o(1)}\cdot z+|E_D|+|E_I|\right )$.

\subsection{Step 2: Final Construction}

In this step we carefully modify the $(h+1)$-level $z'$-subgraph system $$\hat \sset= \left(\hat S,\hat B,\hat U,\hat M,\set{\hat A_i,\hat N_i,\hat R_i}_{i=1}^{h+1},\hat \Lambda,\hat L\right )$$
 constructed in the first step in order to turn it into a valid $(h+1)$-level $z'$-subgraph system for the graph $G'$. Our main goal is to ensure that Properties \ref{prop: degrees in U} and \ref{prop: U neighbors in Ak} hold. 
 We note that this step is very similar to Step 2 in the proof of \Cref{thm: constructing simple subgraph system}, but is somewhat more involved. Throughout the algorithm we maintain, for every vertex $v\in V(G)$, the counter $\hat m(v)$, that counts the number of edges of $\hat M$ that are incident to $v$. These counters can be initialized in time $O(nz)$, and are trivial to maintain without increasing the asymptotic running time of the algorithm, so we do not further discuss their maintenance.
 Our algorithm will only perform the following modifications to the subgraph system:

\begin{itemize}
	\item {\bf Edge Swap:} An edge $e=(v,u)\in \hat M$ with $v\in \hat B$ and  $u\in \hat U$  is deleted from $\hat M$, and, instead, another edge $e'=(v,u')$ with $u'\in \hat U$ is inserted into $\hat M$. The edge swap may only be performed if $\hat m(u')<z'$.
	
	\item {\bf Edge Insertion and Vertex Promotion:} for a vertex $u\in \hat U$, a set $E_u\subseteq \delta_{G'}(u)$ of  $z'-\hat m(u)$ edges connecting it to other vertices of $\hat U$ is inserted into $\hat M$, and $u$ is moved from 
	$\hat U$ to $\hat B$; we will ensure that, for every inserted edge $(u,u')$, prior to the insertion, $\hat m(u')<z'$ held.
	
	\item {\bf Promotion of vertices from $\hat U$ to $\hat A_{h+1}\cup \hat B$:} If a vertex $u\in \hat U$ has at least $z'-h$ edges of $\hat M$ incident to it, then we may move it from $\hat U$ to $\hat S$, where it joins the set $\hat B$. If, additionally, for every edge $(u,v)\in \hat M$ incident to $u$, the other endpoint $v$ lies in $\hat S$, then we move $u$ to $\hat A_{h+1}$ instead.
	
	\item {\bf Promotion of vertices from $\hat B$ to $\hat A_{h+1}$:} If, due to any of the modifications, for some vertex $v\in \hat B$, every vertex in $\set{v'\mid (v,v')\in \hat M}$ lies in $\hat S$, then we move $v$ to $\hat A_{h+1}$.
\end{itemize}

Our algorithm also guarantees that the following invariant always holds:

\begin{properties}{I}
	\item \label{inv: in B connection} For every vertex $v\in \hat B$, there is an edge $(v,u)\in \hat M$ with $u\in \hat U$.
\end{properties}

It is immediate to verify that this property holds at the beginning of Step 2, from the construction of the vertex set $\hat B$.

Throughout the algorithm, we set $\hat N_{h+1}=\hat B$ and $\hat R_{h+1}=\hat U$; as vertices are moved between the sets $\hat U,\hat A_{h+1}$ and $\hat B$, the sets $\hat N_{h+1}$ and $\hat R_{h+1}$ are updated accordingly.
Observe that, over the course of the algorithm, vertices may move from $\hat U$ to $\hat B\cup \hat A_{h+1}$ but not in the other direction, and vertices may move from $\hat B$ to $\hat A_{h+1}$ but not in the other direction. Observe also that our algorithm may not modify an edge of $\hat M$ that is incident to a vertex of $\hat S\setminus \hat B$, and it may not insert new edges incident to vertices of $\hat S\setminus \hat B$ into $\hat M$. 

 From the properties of the initial $(h+1)$-level $z'$-subgraph system, and since we only move vertices from $\hat U$ to $\hat S$ but not in the other direction, it is easy to verify that Property \ref{prop: neighbors of Ai in M in Ni} will continue to hold for vertices that lied in the set $\hat A_{h+1}$ at the beginning of Step 2. 
Since a new vertex $v$ can only be added to $\hat A_{h+1}$ if, for every edge $(u,v)\in \hat M$ incident to it, $u\in \hat S$ holds, we get that, whenever a new vertex $v$ is added to $\hat A_{h+1}$, Property \ref{prop: neighbors of Ai in M in Ni} holds for it. As before, since we only move vertices from $\hat U$ to $\hat S$ but not in the other direction, once a vertex $v$ joins the set $\hat A_{h+1}$, Property \ref{prop: neighbors of Ai in M in Ni} will continue to hold for it for the remainder of the algorithm.

For all $1\leq i\leq h$, the vertex sets $\hat A_i$, $\hat N_i$ and $\hat R_i$ remain unchanged. Since, in the initial $(h+1)$-subgraph system, $\hat N_i\subseteq \hat A^{\geq i+1}\cup \hat B$ held, and since our algorithm may only add new vertices to set $\hat A^{\geq i+1}\cup \hat B$ but not remove them, $\hat N_i\subseteq \hat A^{\geq i+1}\cup \hat B$ will continue to hold, and Property \ref{prop: neighbors of Ai in M in Ni} will continue to hold for $\hat A_i$. Moreover, since the set $\hat R_i$ remains unchanged,  for every vertex $v\in \hat A_i$, the list $\hat L(v)$ that we computed in Step 1 continues to contain all vertices of $N_{G'}(v)\cap \hat R_i$, as required.

From our description, it is immediate to verify that Properties \ref{prop: small degrees in M}, \ref{prop: no edge of M w both endpoints in U} and \ref{prop: edges inc to S} will continue to hold throughout this step. 
Lastly, since Property \ref{prop: containment of Rs} held in the initial $(h+1)$-level subgraph system, and since we may only remove vertices from set $\hat U$, while setting $\hat R_{h+1}=\hat U$ throughout the algorithm, we get that Property \ref{prop: containment of Rs} will continue to hold as well.

Our goal is to perform the modification steps described above to ensure that Properties \ref{prop: degrees in U} and \ref{prop: U neighbors in Ak} also hold. We also need to update the lists $\hat \Lambda(u)$ for vertices $u$ that remain in $\hat U$ (which we do at the very end of the algorithm), and to update the lists $\hat L(v)$ for vertices $v\in \hat A_{h+1}$ (which we do continuously).

Specifically, recall that, for all $1\leq i\le h+1$, for every vertex $v\in \hat A_i$, we are required that $\hat L(v)=N_{G'}(v)\cap \hat R_i$, and that, throughout the algorithm, $\hat R_{h+1}= \hat U$ holds. Therefore, our goal is to ensure that, throughout the algorithm, for every vertex $v\in \hat A_{h+1}$, $\hat L(v)=N_{G'}(v)\cap \hat U$ holds; this property is guaranteed to hold in the initial $(h+1)$-level subgraph system. 
For every vertex $v\in \hat U$, we are required that $\hat \Lambda(v)=N_{G'}(v)\cap (\hat B\cup \hat U)$ holds; the initial $(h+1)$-level subgraph system only guarantees that for each such vertex $v$, $\hat \Lambda(v)=N_{G'}(v)\cap (\hat A_{h+1}\cup \hat B\cup \hat U)$, and that $\sum_{v\in \hat U}|\hat \Lambda(v)|\leq O(nz+|E_I|)$. We will only update the sets $\hat \Lambda(v)$ for  vertices $v\in \hat U$ at the end of the algorithm.

For every vertex $v\in \hat B$, we maintain a counter $n(v)$, that counts the number of edges $(u,v)\in \hat M$ with $u\in \hat U$ and the set $Z(v)$ of all vertices $u\in \hat U$ with $(u,v)\in \hat M$. We initialize all these counters $n(v)$ and vertex sets $Z(v)$ at the beginning of Step 2 in time $O(|\hat M|)\leq O(nz)$, and then maintain them explicitly.
We assume that every vertex that is moved from $\hat U$ into $\hat B$ is automatically inserted into $\hat N_{h+1}$; that every vertex that is removed from $\hat U$ is automatically deleted from $\hat R_{h+1}$; and that every vertex that is moved from $\hat B$ to $\hat A_{h+1}$ is automatically deleted from $\hat N_{h+1}$. We do not explicitly discuss maintaining the vertex sets $\hat N_{h+1}$ and $\hat R_{h+1}$.

Whenever our algorithm identifies a vertex $u\in \hat U$ with  $\hat m(u)\geq z'-h$, it calls Procedure $\procpromote(u)$, described in Figure \ref{fig:propromote}. The procedure moves $u$ from $\hat U$ to $\hat S$, where it joins either the set $\hat A_{h+1}$ (if all edges $(u,v)\in \hat M$ have $v\in \hat S$), or $\hat B$ (otherwise). The procedure also updates all data structures accordingly. Its running time is $\tilde O(|\hat \Lambda(u)|+z')$.

\program{Procedure \procpromote}{fig:propromote}
{
	{\bf Input:} a vertex $u\in \hat U$ with $\hat m(u)\geq z'-h$, and the list $\hat \Lambda(u)=N_{G'}(u)\cap (\hat A_{h+1}\cup \hat B\cup \hat U)$.
	
	\begin{itemize}
		\item Move $u$ from $\hat U$ to $\hat S$;
		\item If, for all edges $(u,v)\in \hat M$ incident to $u$, $v\in \hat S$ holds: add $u$ to $\hat A_{h+1}$.
		
		\item Otherwise:
		 add $u$ to $\hat B$,
			 initialize $Z(u)=\set{v\in \hat U\mid (u,v)\in \hat M}$ and $n(u)=|Z(u)|$.

		\item Compute the set $\hat L(u)=N_{G'}(u)\cap \hat U$ of vertices, by inspecting all vertices in the set $\hat \Lambda(u)=N_{G'}(u)\cap (\hat A_{h+1}\cup \hat B\cup \hat U)$.
	
		\item For every vertex $v\in \hat A_{h+1}\cup \hat B$ with $u\in \hat L(v)$: (note that each such vertex $v$ must lie in $\hat \Lambda(u)=N_{G'}(u)\cap (\hat A_{h+1}\cup \hat B\cup \hat U)$)
		
		\begin{itemize}
		\item delete $u$ from $\hat L(v)$;
		\item if $v\in \hat B$ and $(u,v)\in \hat M$, decrease $n(v)$ by $1$ and delete $u$ from $Z(v)$. If $n(v)$ reaches $0$, move $v$ from $\hat B$ to $\hat A_{h+1}$.
	\end{itemize}
\end{itemize}
}

At the beginning of the algorithm, we apply Procedure $\procpromote(u)$ to every vertex $u\in \hat U$ with $\hat m(u)\geq z'-h$, so all such vertices are removed from $\hat U$ (this step does not change the edge set $\hat M$). After that we process every vertex $u$ that remains in $\hat U$ one by one, using  Procedure \procprocess described in \Cref{fig:procprocess}. When a vertex $u\in \hat U$ is processed by the procedure, it first checks whether there are $z'-\hat m(u)$ edges in $G'$ connecting $u$ to other vertices of $\hat U$ (all such vertices of $\hat U$ must lie in the list $\hat \Lambda(u)=N_{G'}(u)\cap (\hat A_{h+1}\cup \hat B\cup \hat U)$). If so, then the procedure inserts into $\hat M$ exactly $z'-\hat m(u)$ edges of $G'$ connecting $u$ to vertices of $\hat U$, and then calls $\procpromote(u)$. If, additionally, due to these edge insertions, the counter $\hat m(u')$ for some other vertices $u'\in \hat U$ reaches $z'-h$, it calls $\procpromote(u')$ for all such vertices $u'$ as well. Assume now that $G'$ contains fewer than $z'-\hat m(u)$ edges connecting $u$ to other vertices of $\hat U$. The next step is to check whether vertex set $\hat B$ contains 
$z'-\hat m(u)$ vertices $v$ with $(u,v)\in E(G')\setminus\hat M$. If this is the case then we insert the corresponding edges $(u,v)$ into $\hat M$ (if an edge $(u,v)$ with $v\in \hat B$ is inserted into $\hat M$, we delete from $\hat M$ an edge connecting $v$ to some vertex of $\hat U$ to ensure that $\hat m(v)\leq z'$ continues to hold; such an edge must exist from Invariant \ref{inv: in B connection}). The procedure then updates all data structures accordingly.

\program{Procedure \procprocess}{fig:procprocess}
{
	{\bf Input:} a vertex $u\in \hat U$ and the list $\hat \Lambda(u)=N_{G'}(u)\cap (\hat A_{h+1}\cup\hat B\cup \hat U)$.
	
	\begin{itemize}
	\item Denote: $r=z'-\hat m(u)$.
    \item If $|N_{G'}(u)\cap \hat U|\ge r$: (observe that the set $N_{G'}(u)\cap \hat U=\hat\Lambda(u)\cap \hat U$ can be constructed by inspecting all vertices of $\hat \Lambda(u)$)
	\begin{itemize}
		\item Select an arbitary set $\set{v_1,\ldots,v_r}$ of $r$ vertices of $N_{G'}(u)\cap \hat U$;
		\item For all $1\leq i\leq r$, insert the edge $(v_i,u)$ into $\hat M$, so $\hat m(u)=z'$ now holds;
		\item call $\procpromote(u)$;
		\item For all $1\leq i\leq r$, if $\hat m(v_i)\geq z'-h$, call $\procpromote(v_i)$;
		\item Terminate the procedure.
		\end{itemize}
	
	We assume from now on that $|N_{G'}(u)\cap \hat U|<r$.
	 \item	If $|N_{G'}(u)\cap \hat B|\geq z'$: (observe that the set $N_{G'}(u)\cap \hat B=\hat\Lambda(u)\cap \hat B$ can be constructed by inspecting all vertices of $\hat \Lambda(u)$)
	\begin{itemize}
		\item Select an arbitrary set $\set{v_1,\ldots,v_r}$ of $r$ vertices of $N_{G'}(u)\cap \hat B$ such that, for all $1\leq i\leq r$, $(v_i,u)\not\in \hat M$;
		\item For all $1\leq i\le r$, select an arbitrary vertex $a_i\in Z(v_i)$ (recall that $a_i\in \hat U$ and $(v_i,a_i)\in \hat M$ must hold).
		\item For all $1\leq i\leq r$:
		\begin{itemize}
			\item insert the edge $(v_i,u)$ into $\hat M$ and insert $u$ into $Z(v_i)$;
			\item delete the edge $(v_i,a_i)$ from $\hat M$ and delete $a_i$ from $Z(v_i)$;
					
		\end{itemize}
		\item call $\procpromote(u)$ (notice that $\hat m(u)=z'$ now holds);
	
	\end{itemize}
	\end{itemize}
}

It is easy to verify that Procedure $\procprocess(u)$ can be implemented to run in time $O(|\hat \Lambda(u)|+z')$, excluding the time spent on Procedure $\procpromote$. If $u$ remains in $\hat U$ at the end of Procedure $\procprocess(u)$, then we are guaranteed that $|N_{G'}(u)\cap \hat U|\leq z'$ and $|N_{G'}(u)\cap \hat B|\leq z'$ holds at this time. Note that, over the course of the algorithm for Step 2, vertices may leave the set $\hat U$ but they may never join it. Therefore, $|N_{G'}(u)\cap \hat U|\leq z'$ will continue to hold until the end of Step 2. Moreover, the only vertices that may join the vertex set $\hat B$ are  vertices that lie in $\hat U$. 
Therefore, over the remainder of Step~2 at most $|N_{G'}(u)\cap \hat U|\le z'$ additional neighbors of $u$ can move into $\hat B$, and hence $|N_{G'}(u)\cap \hat B|\leq 2z'$ must hold at the end of Step 2. This ensures that, at the end of Step 2, Properties \ref{prop: degrees in U} and \ref{prop: U neighbors in Ak} hold.

Finally, for every vertex $u$ that remains in $\hat U$ at the end of Step 2, we inspect the list $\hat \Lambda(u)=N_{G'}(u)\cap (\hat A_{h+1}\cup \hat B\cup \hat U)$, and delete from it the vertices of $\hat A_{h+1}$, to ensure that $\hat \Lambda(u)=N_{G'}(u)\cap ( \hat B\cup \hat U)$ holds as required. From our discussion, we obtain a valid $(h+1)$-level $z'$-subgraph system $\hat \sset$ for the graph $G'=(G\cup E_I)\setminus (E_D\setminus E_D')$. 
Recall that edge set $E'_D$ remained unchanged over the course of Step 2, so, at the end of the algorithm,  $|E_D'|\leq \frac{|E_D|\cdot z'}{z}$ holds as required. We return $E'_D$ and $\sset'=\hat \sset$ as the algorithm's outcome.

As observed already, the running time of Procedure $\procprocess(u)$ is $\tilde O(|\hat \Lambda(u)|+z')$, excluding the time spent on  Procedure $\procpromote$. Procedure $\procpromote$ may be applied at most once to every vertex $u'$ that lies in $\hat U$ at the beginning of Step 2, and its running time is $\tilde O(|\hat \Lambda(u')|+z')$. 
The total running time of the algorithm for Step 2 is then bounded by: $\tilde O(nz'+|E_I|+\sum_{u\in \hat U}|\hat \Lambda(u)|)\leq  \tilde O(nz+|E_I|)$, and the running time of the entire algorithm is bounded by $\tilde O(n^{1+o(1)}\cdot z+|E_I|+|E_D|)$.

\bibliographystyle{alpha}
\bibliography{ref}
\end{document}